\documentclass[a4paper,12pt]{article}
\usepackage[utf8]{inputenc}

\usepackage{amsmath}
\usepackage{amsfonts}
\usepackage{amsthm}
    \newtheorem{lemma}{Lemma}

\usepackage{csquotes}

\usepackage[style=authoryear,maxcitenames=2,maxbibnames=99,uniquelist=false,backend=biber]{biblatex}
\usepackage{titlesec}

\usepackage{graphicx}
\usepackage[export]{adjustbox}
\usepackage{hyperref}
\titlelabel{\thetitle.\quad}
\usepackage{multirow}
\usepackage[figurename=Fig.,labelsep=endash,labelfont=bf]{caption}
\usepackage{longtable}
\usepackage{adjustbox}
\usepackage{caption}

\usepackage{appendix}
\providecommand{\keywords}[1]
{
	\small      
	\textbf{\textit{Keywords:}} #1
}
\providecommand{\JELcodes}[1]
{
	\small      
	\textbf{\textit{JEL codes:}} #1
}
\usepackage{authblk}
\newtheorem{proposition}{Proposition}
\newtheorem*{remark}{Remark}
\title{\LARGE\bf Nonlinear fiscal multiplier controlling for policy and economic regimes}
\author{Szymon P. Chudziak$^{1,2}$, PhD}
\affil{\footnotesize{$^1$Polish Economic Institute, \\ szymon.chudziak@pie.net.pl \\
		$^2$Department of Economics I, Collegium of Economic Analysis, SGH Warsaw School of Economics \\ schudz@sgh.waw.pl}}
\date{\footnotesize \today}

\begin{document}
	\maketitle
	\begin{abstract}
		A new form of the fiscal multiplier suited to data expressed as gross growth rates is proposed, together with a structural vector autoregression that isolates discretionary fiscal policy from regime and rule-based components. Because cumulating growth-rate responses over a horizon involves a product rather than a sum, the proposed multiplier is multiplicative: both numerator and denominator are polynomials in the size of the fiscal impulse, so the multiplier varies nonlinearly with the scale of the expenditure programme -- a dependence that conventional additive multipliers cannot capture. The multiplier requires no discounting, no ad hoc conversion factor, and no auxiliary estimate of an average expenditure share of output, and it is computed directly from nominal, non-seasonally-adjusted data, matching the terms in which government budgets are actually expressed. A version of the narrative approach to identification is introduced. It allows to decompose the influence of fiscal policy into a regime component, a standard policy rule embedded in the autoregressive structure, and discretionary shocks. Removing the regime component from the residuals yields smaller identified shocks, and correspondingly smaller multipliers, than found in SVARs not controlling for regime changes -- closing much of the gap between SVAR- and DSGE-based estimates. It is found that multipliers vary substantially across countries and across types of government expenditure, that regime volatility itself differs markedly by country, and that several economies display cyclicality in spending and multiplier paths that only becomes visible once responses are cumulated multiplicatively. The resulting multipliers are nominal, derived from unadjusted data, proven to be sensitive to shock size, and purged of regime effects, offering a tool with direct relevance for budget planning and fiscal policy evaluation.
	\end{abstract}
	
	\noindent{\footnotesize \JELcodes{E62, H5}}
	
	\noindent{\footnotesize \keywords{fiscal multiplier, fiscal policy, regimes, narrative identification}}
	
	\section{Introduction}\label{s1}

Empirical literature on fiscal policy has produced a remarkably wide range of estimated fiscal multipliers, obtained by a correspondingly wide range of methods. Part of this dispersion is a matter of measurement conventions, part of it is a matter of identification, and part of it is actual cross-country heterogeneity. Underlying all three is a question that is rarely posed explicitly: how exactly is fiscal policy accounted for in a structural vector autoregression? Is a single equation for government expenditure sufficient to represent it?

Different countries have distinct structures of their economies, and therefore different dynamic responses to innovations. These reactions reflect structural changes and characteristics of a given economy that are uncontrolled, or not fully controlled, in the estimated model. Policy changes are likely to affect the estimates, given that they are non-random and are not captured by the autoregressive part of the VAR. It is unclear, however, how exactly they affect the results of an estimation. Fiscal policy of any country has cyclical and noncyclical components, and the question is whether this distinction can be made visible in the estimated multipliers.

A second question concerns the form of the multiplier itself. Additive multipliers are, by construction, linear in the numerator and in the denominator, so that they do not depend on the size of the shock. Researchers often use a discounted, real version of the multiplier. Governments, however, need nominal versions for the evaluation of policy effects and of their impact on the budget, and it is not obvious which discount rates should be applied. Conversely to the so-called policy-relevant multipliers of \textcite{mountford_what_2009}, calculated as ratios of present discounted value of impulse responses, governments need nominal reaction functions and nominal multipliers, because they have to operate with a nominal budget under unknown future inflation. Moreover, it is not clear what interest rates, or what measure of inflation, should be used in such a case: the GDP deflator, or an index tailored to the sectoral characteristics of government expenditure. If calculated multipliers are to be useful for policy purposes, they would ideally be expressed in nominal terms.

A similar issue is the use of detrended and seasonally adjusted data for estimation. Such data is deprived of a great deal of information on the actual dynamics of a time series. In the context of fiscal policy, this means inaccurate predictions and inference. It would therefore be beneficial to obtain credible results using nonfiltered data. Levels cannot be used, because they are nonstationary, but gross growth rates can, and they have the additional property that they can be used directly to reproduce levels.

Another issue related to the large range of estimates of multipliers in economic studies is the potential dependence of the size of a multiplier on the type of government spending. Consumption of the government may have different effects on output in the long run than aggregate government expenditure, and defence spending may be less efficient in stimulating output growth.

This paper contributes to the literature in the following ways. First, a new, multiplicative, form of the multiplier is proposed, consistent with the form of the data used. The proposed multiplier makes it possible to discover a nonlinear dependence between the size of the multiplier and the size of the shock to government expenditure, and it is the appropriate form when the model is estimated on gross growth rates. Second, controlling for policy and economic regimes in an SVAR allows to decompose the part of fiscal policy that is not captured by the autoregression into a nondiscretionary part, represented by appropriate intervention variables, and a discretionary part, that is, the structural shocks. The identified shocks are therefore smaller than in the usual SVARs, which do not control for regime changes. This lowers the estimates of multipliers and closes the gap between the values obtained from SVAR and from DSGE models: the fiscal multipliers reported here are in most cases smaller than one, and rarely cross one. Third, different types of government expenditure are shown to have unalike multipliers. Fourth, distinct economies are characterised by distinct values of fiscal multipliers, which is shown for the six largest economies of the European Union. Fifth, the approach is robust when used with non-adjusted data in the form of gross growth rates, so that it may form a basis for empirical tools with direct real-world policy relevance, as governments operate with budgets expressed in nominal terms. Sixth, it is argued that the entire estimated model, not only impulse responses and the derived multipliers, characterises economic policy and it is best to evaluate it alongside the dynamic multipliers.

This paper is organised as follows. Section \ref{s2} contains the discussion of the existing ways of measuring fiscal multipliers and the difficulties associated with them. The model, the intervention variables that control for policy and economic regimes, and the proposed multiplicative multiplier are presented in section \ref{s3}. Results are reported in section \ref{s4}, while their discussion is located in section \ref{s5}, followed by conclusions in section \ref{s6}. Appendix \ref{appndx} contains some additional results.

\section{Measuring fiscal multipliers}\label{s2}

Several distinct objects are called fiscal multipliers in the literature, and they are calculated in different ways. One of them is the quasi-multiplier in the sense of \textcite{blanchard_empirical_2002}. Another is the additive multiplier, described as the true multiplier by \textcite{ramey_ten_2019}; it has a linear numerator and a linear denominator -- sums of the elements of appropriate impulse response functions. This formulation is therefore independent of the size of the shock.

Researchers often use a discounted, real version of the multiplier, as in \textcite{mountford_what_2009}, and as reported in the discussion of the literature by \textcite{ramey_ten_2019}, whereas governments need a nominal one, for the reasons already discussed in the introduction. \textcite{mountford_what_2009} moved the literature forward by introducing policy-relevant multipliers, calculated as the present discounted value of the output response over time divided by the present discounted value of the government spending response over time to the shock. In most applications, different interest rates used for this present discounted value, including a zero discount rate, give nearly identical multipliers, because the timing of the government spending and output responses is very similar; these multipliers are often known as present value or cumulative multipliers. A difference of even a few percentage points is, however, relevant in the context of a government budget and of fiscal policy. It can amount to many millions of a country's currency, which makes it relevant for government budget planning and policy evaluation.

A related practice criticised by \textcite{ramey_ten_2019} is the conversion of elasticities into multipliers with an ad hoc conversion factor, the average of $Y/G$ over the sample. In \textcite{owyang_are_2013} biases that could arise from this practice were identified; in their historical sample, $Y/G$ varied significantly. \textcite{sims_state-dependent_2018,sims_output_2018} also found that this practice tends to bias multipliers differentially, making them seem much higher during recessions.

As \textcite{sims_output_2018} observe, in the empirical literature it is common not to estimate multipliers directly, but to first estimate the elasticity of output with respect to tax revenue and then to transform it by multiplying by the inverse of the average tax revenue share of output, as in \textcite{blanchard_empirical_2002} and \textcite{mountford_what_2009}. A criticism that \textcite{ramey_government_2018} make of the existing empirical literature on state-dependent multipliers concerns precisely this conversion of elasticities into multipliers. Most empirical work uses logs of variables and estimates the elasticity of output with respect to government spending; this elasticity is then converted into a multiplier by post-multiplying by the inverse of the average government spending share of output. \textcite{ramey_government_2018} argue that this approach is likely to make the output multiplier artificially high in recessions, because the government spending share of output is countercyclical. The analysis of \textcite{sims_output_2018} shows the quantitative relevance of this criticism. They argue that converting elasticities into multipliers using a fixed government spending share would result in the output multiplier varying between 0.8 and 1.1 across states, which they regard as incorrect in their context, and they argue that the multiplier is almost constant when monetary policy is characterised by a Taylor rule.

Ramey's conclusion is that a specification in logarithms requires the ad hoc conversion factor. It is far from clear, however, why this should lead to replacing the logarithms of government spending, GDP, and taxes with the ratios of each of those variables to the \textcite{ramey_government_2018} polynomial trend estimate of potential GDP \parencite{ramey_ten_2019}. Such ratios are not structural variables, and it is not evident what the shocks mean in a model formulated in these terms.

Even cumulative multipliers do not fully reflect the consequences for the government budget. If an increase in government spending raises GDP, then a rise in tax revenues is to be expected. Thus, even without an exogenous increase in tax rates, the government budget deficit should be expected to rise by less than the total amount of government spending. This is one of the reasons why the model proposed in this paper controls for tax and fiscal policy regimes directly.

A further strand of the literature concerns fiscal foresight. Most of the literature tries to address anticipation whenever feasible, either by constructing measures of news, from narratives or from bond spreads, or by including professional forecasts of government spending in order to mitigate the problem. \textcite{ramey_identifying_2011} provides the clearest demonstration of why this matters for government spending. She shows that standard Cholesky-identified VARs and the \textcite{ramey_costly_1998} narrative war-dates approach give opposed answers about the response of consumption and real wages to government spending shocks, and that the reconciling factor is timing rather than the underlying economic mechanism: both the war-dates variable and professional forecasters' predictions Granger-cause the VAR-identified spending shocks, which shows that the VAR shocks are simply picking up the news several quarters too late. Motivated by this finding, \textcite{ramey_identifying_2011} constructed a continuous narrative measure of the expected present discounted value of changes in defense spending running from 1939 to 2008, together with a second, forecast-error-based measure for the post-Korean-War period. Once government spending shocks are dated at the point the news arrives rather than at the point spending actually changes, consumption, the real wage, and investment fall on impact, as the neoclassical model predicts, and the implied government spending multiplier lies between 0.6 and 1.2 depending on whether the two world wars are included in the sample. Artificially delaying the timing of the narrative dates is shown to reproduce the Keynesian-looking VAR results, which confirms that timing, not identification scheme per se, drives the earlier disagreement in the literature. This paper is therefore a central piece of evidence for the argument, made throughout the present review, that failing to control for the timing of news about future fiscal actions -- rather than the choice between SVAR and narrative methods as such -- is what generates much of the disagreement over the sign and size of fiscal effects.

It is worth asking, however, what anticipation is, and what the possibility of estimating a correlation between some variables and future shocks amounts to. Fiscal spending is usually directed, and therefore reaches only a fraction of the economy, so it is not clear how firms or households in general could prepare for it. If they accumulate funds in order to start spending more, this is the effect of the shock anyway. The cases in which agents do start spending before a programme has begun are those described by \textcite{house_phased-tax_2006} and \textcite{mertens_dynamic_2013}, who showed that while unanticipated tax cuts have expansionary effects on output, phased-in tax cuts depress output during the phase-in period, because firms and consumers delay their activity until tax rates are lower. This is exactly what can be controlled by means of the fiscal regime approach proposed in this paper.

The narrative approach to identifying policy shocks by purging a policy variable of its systematically forecastable component is not confined to fiscal policy. \textcite{romer_new_2004} apply essentially the same logic to monetary policy: they combine the narrative record of FOMC deliberations with the Federal Reserve's internal Greenbook forecasts to construct a measure of the intended federal funds rate that is free of both the endogenous within-meeting response of the funds rate to current conditions and the anticipatory response of policy to the Federal Reserve's own forecasts of future output and inflation. The resulting shock series produces output and price responses to monetary policy that are considerably larger and, for prices, faster than those obtained from the actual funds rate, and the conventional "price puzzle" found with the raw funds rate disappears even without controlling for commodity prices. This is direct, cross-domain evidence for the more general point made in this review: policy variables that are not first purged of their systematic, forecastable component -- whether because of automatic feedback rules or because policymakers act on private information about the future -- yield biased and attenuated estimates of policy effects, and narrative methods that use the policymaker's own information set to construct the purified shock are one of the few ways of addressing this problem convincingly.

\textcite{romer_macroeconomic_2010} apply the same narrative logic to tax policy in a way that is close in spirit to, but methodologically independent of, the approach of Ramey and Shapiro. Rather than looking for a small number of large, discrete natural experiments, they use the full postwar legislative record -- presidential speeches, the Economic Report of the President, and Congressional committee reports -- to classify every significant postwar legislated U.S. tax change according to its stated motivation. Tax changes taken to offset a change in spending or to counteract other current or prospective developments are classified as endogenous, while tax changes taken to reduce an inherited budget deficit or to promote long-run growth are classified as exogenous, on the grounds that these latter motivations are unlikely to be systematically correlated with other determinants of output growth. The resulting series of 54 exogenous quarterly tax changes is shown not to be Granger-caused by output, government spending, oil prices, or (with the exception of a few observations around the 1981 tax cuts) measures of monetary policy shocks, which supports the claim that the classification succeeds in isolating changes that are plausibly uncorrelated with the error term in an output regression. 

Their main result is that an exogenous tax increase of 1 percent of GDP lowers output by close to 3 percent within about two years, an effect that is far larger and more persistent than the one obtained when the same regressions are run using broader tax measures such as legislated tax changes as a whole or the change in cyclically adjusted revenues, both of which include tax actions that are correlated with other determinants of output. The output effects are shown to operate mainly through investment rather than consumption, to be tied more closely to the date taxes actually change than to the date of legislative passage, and to have weakened somewhat after 1980, a pattern the authors partly attribute to a more aggressive monetary policy response to fiscal shocks in the more recent period. This finding -- extremely large and persistent multipliers for tax changes identified narratively -- is precisely the empirical counterpart to the tension noted at the end of this review: very high narrative-based tax multipliers of this kind are difficult to reconcile with the modest government spending multipliers reported elsewhere in the literature unless the two are modelled jointly, along with the associated tax and fiscal policy regime, rather than estimated as if they were independent objects.

Any analysis that seeks to measure a causal effect must confront identification issues. An example of the problem that arises here is that if governments increase spending in response to a recession, then the simple correlation between government spending and GDP will confound the positive causal effect of government spending on GDP with the negative causal effect of GDP on government spending. In the past, the standard macroeconomic approach used to extract the exogenous rise in government spending was a structural vector autoregression. In most applications, this approach is based on the assumption that the exogenous part of government spending is the part not forecast by lagged values of spending, GDP, and taxes. Alternatively, in order to identify exogenous movements in taxes, \textcite{blanchard_empirical_2002} used external estimates of the elasticity of tax revenue to income, which allowed identification of the component of taxes that was not induced by movements in GDP. Several papers have highlighted potential problems with these widely used methods. First, failing to account for fiscal foresight could lead to biased estimates. Second, the tax multiplier estimates were very sensitive to the value of the external tax elasticity estimate used, as shown for example by \textcite{mertens_fiscal_2014} and \textcite{caldara_analytics_2017}. These concerns led to the development of other identification methods, using natural experiments and narrative methods, and, as a result, the standard SVAR identification approach is no longer the first resort in the literature on fiscal multipliers. The approach proposed here belongs to the narrative family, while retaining the SVAR framework.

A further alternative to the Cholesky-ordering SVAR is identification via restrictions, which imposes only the qualitative comovements that theory predicts on impact, rather than a specific recursive ordering or a narrative date. \textcite{pappa_effects_2009} applied this approach to the question of how fiscal shocks are transmitted to employment and the real wage, restricting government consumption, investment, and employment shocks to raise output and the budget deficit on impact -- restrictions that hold in both a real business cycle and a New Keynesian version of her model -- and then tracing out the resulting dynamics for the labour market using both aggregate U.S. data and a panel of U.S. states. She finds that shocks to government consumption and investment raise both real wages and employment contemporaneously, in the aggregate and in the large majority of states, a pattern that is difficult to reconcile with the standard neoclassical prediction of a negative wealth effect on the real wage, though it is broadly consistent with sticky-price New Keynesian transmission. The results for government employment shocks are more mixed: while aggregate employment and the real wage both rise, in roughly one third of U.S. states an increase in government employment actually crowds out private employment. This heterogeneity across the components of spending, and across identification schemes, reinforces the point made throughout this review that estimated multipliers -- and the sign of the associated real-wage response -- can differ enormously depending on both what kind of spending shock is examined and how it is identified, apart from the state-dependence and foresight issues discussed above.

Despite its popularity in the fiscal and monetary VAR literatures alike -- \textcite{baumeister_drawing_2020} counted close to a hundred studies that identify structural shocks solely, or primarily, from sign restrictions, with \textcite{pappa_effects_2009} itself appearing on that list -- this paper does not adopt sign-restriction identification, for reasons set out in a sequence of papers by \textcite{baumeister_sign_2015}, \textcite{baumeister_inference_2018}, and \textcite{baumeister_drawing_2020}. Their central point is that an assumption about signs is, by construction, not enough to identify a structural model: it only ever delivers a set of admissible structural parameters, all of which fit the reduced-form data and the sign restrictions equally well, so that the object actually recovered is what the set-identification literature calls the identified set rather than a single structural VAR. \textcite{baumeister_sign_2015} show analytically that the standard algorithm for generating draws from this set -- drawing an orthonormal rotation matrix from the uniform Haar distribution and combining it with the Cholesky factor of the reduced-form covariance matrix -- is not, as commonly believed, an uninformative or agnostic procedure. Because impulse responses and elasticities are nonlinear functions of the rotation angle, a uniform prior on the angle induces a highly non-uniform, often Cauchy-shaped, prior on the impulse responses and structural elasticities themselves; in an application to short-run labour demand and supply elasticities they show this implicit prior can favour extreme values over economically plausible ones, or vice versa, purely as an artefact of the number of variables in the system. Crucially, they demonstrate that this influence of the unacknowledged prior does not vanish even asymptotically: with an infinite sample, the posterior for a sign-identified structural parameter is confined to the identified set, and within that set is simply proportional to the prior, so that no amount of data can ever substitute for economically defensible identifying information once a model is only set- rather than point-identified.

\textcite{baumeister_inference_2018} sharpened this critique further by showing, using a three-equation monetary VAR, that a researcher who imposes no sign restrictions whatsoever -- keeping every single candidate rotation -- can still generate what looks like a sensible median impulse response with a tight 68\% credible band purely from the randomness of the algorithm's own random-number generator, with no information from the data or from economic theory playing any role at all; adding conventional sign restrictions merely truncates this same manufactured distribution rather than genuinely sharpening it. They show that reporting the full identified set, rather than its median or a credible interval carved arbitrarily out of it, is often strikingly uninformative -- the identified set for the effect of a monetary contraction on several variables of interest is shown to include both zero and values of either sign at short horizons -- which reveals how little the sign restrictions alone are actually pinning down. Their proposed remedy, followed up in \textcite{baumeister_drawing_2020} in the specific context of a study of U.S. monetary policy and emerging-market capital flows that combines sign and zero restrictions, is not to abandon Bayesian methods but to make the prior distribution over the structural parameters explicit, economically motivated, and defensible, rather than allowing an arbitrary and mathematically incidental feature of the estimation algorithm (the Haar measure on rotation matrices) to masquerade as an assumption-free identification scheme.

These results are directly relevant to the identification choices made in this paper. Since fiscal shocks of the kind studied here are, at best, only partially identified by plausible sign restrictions on the comovement of output, spending, deficits, and taxes the \textcite{baumeister_sign_2015,baumeister_inference_2018,baumeister_drawing_2020} results imply that any point estimates or credible sets reported from such a model would necessarily be driven substantially by an unacknowledged and mathematically arbitrary prior embedded in the standard rotation-based algorithm, rather than by the data or by the sign restrictions themselves. Rather than adopt this approach and then have to defend an implicit prior that was never chosen deliberately, the fiscal and tax regimes used for identification in this paper are constructed instead, using the narrative and regime-based information discussed above, which corresponds to the kind of explicit, economically motivated, and openly defended identifying information that \textcite{baumeister_sign_2015} and \textcite{baumeister_inference_2018} argued is the only sound basis for structural inference once a model is not fully identified by hard zero or timing restrictions.

A closely related body of work asks whether the size of the spending multiplier itself depends on the state of the business cycle. \textcite{auerbach_measuring_2012} estimated a smooth-transition regime-switching VAR that allows both the contemporaneous covariance matrix and the lag polynomial to differ between recessions and expansions, and report output multipliers that are far larger in recessions (of the order of 1 to 1.5, and as high as 2.5 at long horizons for some specifications) than in expansions (close to zero or even negative), a difference that becomes still more pronounced once the government spending series is purged of the component predicted by real-time Survey of Professional Forecasters and Federal Reserve Greenbook forecasts. \textcite{ramey_can_2011} surveyed this and the broader multiplier literature and concludes, after reviewing both the theoretical range implied by neoclassical and Keynesian models and the empirical estimates available at the time, that the plausible range for the multiplier on a temporary, deficit-financed increase in government purchases is between 0.8 and 1.5, while noting that the evidence for materially larger multipliers during recessions is itself sensitive to the assumption, built into regime-switching VARs such as that of \textcite{auerbach_measuring_2012}, that the economy remains in the same regime for the entire horizon over which the multiplier is computed.

\textcite{ramey_government_2018} took up this critique directly: using a purpose-built quarterly U.S. dataset extending back to 1889 and Jordà local-projection methods that do not require this constant-regime assumption, they find no evidence that government spending multipliers are higher in periods of economic slack -- estimated multipliers lie mostly between 0.3 and 0.8 in both high- and low-unemployment states -- and only mixed evidence of multipliers above unity near the zero lower bound, mainly once the rationing period of World War II is excluded from the sample. They further showed that when the \textcite{auerbach_measuring_2012} smooth-transition VAR is re-estimated allowing the regime indicator to evolve endogenously with the size and persistence of the shock, rather than being held fixed at its initial value, the estimated recession multiplier falls back close to unity, which accounts for much of the difference between the two sets of results. Taken together, these three papers illustrate that apparent state-dependence in the fiscal multiplier is at least partly an artefact of how impulse responses are constructed from a nonlinear model, a methodological warning that parallels the timing-related sensitivity documented by \textcite{ramey_identifying_2011} for the simple linear VAR case.

A different and, for the purposes of this paper, more directly relevant notion of regime-dependence concerns not the business cycle but the prevailing monetary-fiscal policy mix itself. \textcite{mumtaz_fiscal_2020} estimated a battery of SVARs for the United States -- a recursive scheme, the Blanchard-Perotti elasticity-based scheme, a Mountford-Uhlig sign-restricted VAR, and a proxy SVAR that uses the \textcite{mertens_empirical_2012} tax instrument and the \textcite{ramey_identifying_2011} defense news series as external instruments for tax and spending shocks respectively -- and show, across every one of these identification schemes, that the response of real stock prices to government spending and tax shocks changed sign around 1980. This finding survives a long list of robustness checks. To explain the shift, they estimate a New Keynesian model augmented with a fiscal sector over the same two subsamples and classify the estimated policy rules using the \textcite{leeper_equilibria_1991} active/passive taxonomy: prior to 1980 fiscal policy is estimated to be active and monetary policy passive, so that a spending shock raises inflation expectations, lowers the real interest rate, and thereby raises consumption, investment, and equity prices, generating a multiplier above unity; after 1980 the regime flips to active monetary and passive fiscal policy, so that the same spending shock now raises the real rate and depresses both private demand and equity prices, generating a much smaller multiplier. 

This is direct structural evidence that the sign and size of the same government spending shock's effects depend on the prevailing fiscal-monetary regime rather than being a stable structural constant, which reinforces the case made throughout this review -- and the specific choice made in the model proposed here -- for identifying fiscal shocks jointly with the tax and monetary-fiscal policy regime in place, rather than estimating a single multiplier and treating regime shifts as a robustness check to be run after the fact. It is also worth noting that one of the four identification schemes used by \textcite{mumtaz_fiscal_2020}, the Mountford-Uhlig sign-restriction VAR, is itself subject to the \textcite{baumeister_sign_2015,baumeister_inference_2018} critique discussed above; the fact that their conclusions are corroborated by the narrative-instrument and elasticity-based schemes as well is therefore important, since it is these latter, more fully identified schemes -- rather than the sign-restricted one -- that carry the real evidentiary weight for the regime-shift finding.

Identification of fiscal policy shocks using natural experiment methods has its own limitations. Such settings are confined to the specific case which provides the source of identification, and their results are usually not generalisable to the majority of standard fiscal policy. A natural-experiment approach can therefore be of little use for governments and for budget planning. 

The form of the data used in the analysis raises a separate set of difficulties. Detrended and seasonally adjusted data is spectrally distorted, and the reactions of an economy estimated on such data will differ from the actual ones. Governments need forecasts that are as accurate as possible. Detrending and adjusting the data narrows the confidence intervals of impulse response functions, but in the light of planning government expenses this may not be desirable, as it gives the misleading impression that the trajectories of the variables shown are accurate, whereas non-adjusted data with trend display a much wider range of possible paths. Finally, data in levels, in first differences, or in the logarithmic equivalents of the two do not allow the nonlinearity of multipliers with respect to the size of the shock to be discovered, which is possible when using data in the form of gross or net growth rates, paired with an appropriate multiplier form, presented in this paper.

In \textcite{ramey_macroeconomic_2016} and \textcite{ramey_government_2018}, it was shown that cumulative multipliers can be estimated in a one-step instrumental variables method based on local projections: cumulative GDP up to horizon $h$ is regressed on cumulative government spending up to horizon $h$, using an SVAR shock or a narrative variable as an instrument. \textcite{ramey_government_2018} implement this directly by using both their historical military news series (an updated version of the \textcite{ramey_identifying_2011} defense news variable) and a Blanchard-Perotti-type shock as instruments, separately and in combination, and show that the two instruments have complementary strengths -- the Blanchard-Perotti shock is more relevant at short horizons, the narrative news shock at longer ones -- which is itself further evidence of the timing issues emphasised by \textcite{ramey_identifying_2011}. This does not mean, however, that SVAR estimation is mathematically equivalent to instrumental variables estimation. Local projections and SVARs theoretically produce the same impulse response functions, but if instrumental variables are used, the finite sample bias of the instrumental variables estimator is incurred, which is particularly severe in short macroeconomic series.

Finally, considering multipliers of tax changes and multipliers of government spending separately makes little sense, unless one assumes that public debt absorbs all the difference. The small government spending multipliers reported in the literature -- for instance the 0.3 to 1.5 range surveyed by \textcite{ramey_can_2011} and the sub-unity estimates of \textcite{ramey_government_2018} -- and the very high tax change multipliers obtained from empirical models such as that of \textcite{romer_macroeconomic_2010}, are mutually inconsistent, because the large drops of output after tax increases shown by the impulse response functions imply large increases of GDP after tax decreases, as impulse response functions are symmetric by definition. The source of the problem, as argued in this paper, is that tax and fiscal policy regimes have not been controlled for in SVARs.

\section{The approach}\label{s3}

\subsection{Data and the form of the variables}\label{s3.1}

The model is estimated on annual, nominal, non-detrended and non-adjusted national accounts data. Quarterly series are aggregated to annual ones by summation of four consecutive quarters, and government expenditure by function is taken directly at annual frequency. Because the levels of the series are nonstationary, and because seasonal adjustment and detrending would remove exactly the information that is relevant for budget planning, all variables enter the model as annual gross growth rates. For a series $X_{t}$, with the gross growth rate denoted as $x_t$, the level of the variable is obviously reproduced exactly by $X_{t}=X_{t_{0}}\prod_{s=t_{0}+1}^{t}x_{s}$, that is, no information on the level is lost by the transformation. Moreover, given that the  values of impulse responses from a SVAR of gross growth rates of variables of interest are percentage points, the correct form of multiplier is different from the usual ratio of sums of elements of impulse responses up to a given horison. This multiplier form is introduced in subsection \ref{s3.6}.

The vector of endogenous variables is six-dimensional,
\begin{equation}\label{eq:1}
	{y}_{t}\;=\;\bigl(g^{d}_{t},\;g^{n}_{t},\;c_{t},\;i_{t},\;r_{t},\;y_{t}\bigr)^{\prime},
\end{equation}
where the elements are the gross growth rates of, respectively,
\begin{itemize}
	\item $g^{d}_{t}$ -- general government expenditure on defence, $G^{d}_{t}$;
	\item $g^{n}_{t}$ -- non-defence general government expenditure, $G^{n}_{t}=TGGE_{t}-G^{d}_{t}$, that is, total general government expenditure less expenditure on defence;
	\item $c_{t}$ -- final consumption expenditure of households and NPISHs, $C_{t}$;
	\item $i_{t}$ -- gross capital formation, $I_{t}$;
	\item $r_{t}$ -- total general government revenue, $R_{t}$;
	\item $y_{t}$ -- gross domestic product at market prices, $Y_{t}$.
\end{itemize}
Net exports are not included among the endogenous variables. The reason is that the net export series changes sign within the sample, so that its gross growth rate is not a meaningful quantity: a ratio of two negative values is positive, and the transformation therefore no longer conveys the direction of change of the series.

The sample used in the estimation for Poland covers the years 1995-2024 in levels, that is, $T=29$ observations on gross growth rates, from 1996 to 2024, and $T-p=28$ effective observations. Expenditure on defence is available from Eurostat only from 2000 onwards in the narrower military defence classification, which is why the general government defence function is used here.

\subsection{The VARX model with intervention variables}\label{s3.2}

The reduced form of the model is a vector autoregression of order $p$ augmented by a vector of intervention, that is, regime, variables,
\begin{equation}\label{eq:2}
	{y}_{t}\;=\;{c}\;+\;\sum_{l=1}^{p}\Phi_{l}\,{y}_{t-l}\;+\;\Psi\,{d}_{t}\;+\;{u}_{t},
	\qquad {u}_{t}\sim({0},\Sigma_{u}),
\end{equation}
where ${c}$ is an $N\times 1$ vector of constants with $N=6$, $\Phi_{l}$ are $N\times N$ matrices of autoregressive coefficients, ${d}_{t}=(d_{1t},\dots,d_{Mt})^{\prime}$ is the $M\times 1$ vector of intervention variables described in Section \ref{s3.3}, and $\Psi$ is the $N\times M$ matrix of their coefficients.

The lag order is set to $p=1$. This is not a purely statistical choice. The structural assumptions made here about the three fiscal variables, $g^{d}_{t}$, $g^{n}_{t}$ and $r_{t}$, are that only the previous period matters for them: budgets are prepared and executed on an annual basis, and the systematic component of the fiscal rule relates the current year's expenditure and revenue to the outcome of the immediately preceding year. Since a standard VAR imposes a single, common lag order on all variables and all equations, a model with $L>1$ would be inconsistent with these structural assumptions, and $p=1$ is therefore imposed on the whole system. A consequence of this restriction, discussed in Section \ref{s4}, is that some of the estimated coefficients in the remaining equations are cumulative: they absorb the effects of more than one lag of a variable in equations for which a longer lag structure would be appropriate.

Stacking the $T-p$ effective observations, the model given in equation \ref{eq:2} is written as
\begin{equation}\label{eq:3}
	{Y}\;=\;{X}{B}\;+\;{U},
	\qquad
	{x}_{t}^{\prime}=\bigl(1,\;{y}_{t-1}^{\prime},\dots,{y}_{t-p}^{\prime},\;{d}_{t}^{\prime}\bigr),
\end{equation}
with ${Y}$ of dimension $(T-p)\times N$, ${X}$ of dimension $(T-p)\times K$ and $K=1+Np+M$. The coefficient matrix is estimated by least squares,
\begin{equation}\label{eq:4}
	\widehat{{B}}\;=\;\bigl({X}^{\prime}{X}\bigr)^{+}{X}^{\prime}{Y},
	\qquad
	\widehat{{U}}\;=\;{Y}-{X}\widehat{{B}},
	\qquad
	\widehat{\Sigma}_{u}\;=\;\frac{\widehat{{U}}^{\prime}\widehat{{U}}}{T-p-K},
\end{equation}
where $(\cdot)^{+}$ denotes the Moore-Penrose pseudoinverse, calculated by singular value decomposition. The pseudoinverse is used because for Italy, Spain and Poland the intervention variables constructed below form an exhaustive partition of the sample, so that their sum equals the constant regressor and ${X}$ has rank $K-1$. This is a deliberate modelling choice, given the assumed form of the growth process expressed by equation \ref{eq:8}. The minimum-norm least squares solution is then reported; the fitted values, the residuals, the estimated covariance matrix, the identified shocks and all impulse responses are invariant with respect to this normalisation, and only the split of the estimated level between the constant ${c}$ and the regime coefficients $\Psi$ is affected by it. For Germany and the Netherlands no subset of the overlapping regimes partitions the sample, so that ${X}$ has full column rank $K$ and the pseudoinverse coincides with the ordinary inverse. The reasons for which this collinearity is retained rather than removed, and the exact properties of the resulting estimator, are set out in Section \ref{s3.3}.

\subsection{Intervention variables and the narrative identification of the shocks}\label{s3.3}

The intervention variables are the core of the identification strategy. Each of them is an indicator of a policy-economic regime,
\begin{equation}\label{eq:5}
	d_{mt}\;=\;
	\begin{cases}
		1, & t\in\mathcal{T}_{m},\\
		0, & t\notin\mathcal{T}_{m},
	\end{cases}
	\qquad m=1,\dots,M,
\end{equation}
where $\mathcal{T}_{1},\dots,\mathcal{T}_{M}$ are consecutive subperiods whose union is the whole sample, $\bigcup_{m=1}^{M}\mathcal{T}_{m}=\{1,\dots,T\}$; in Italy, Spain and Poland they are in addition disjoint, while in Germany, France and the Netherlands some of them overlap with others or are nested within longer ones. The subperiods are dated narratively, from the legislative record of the tax system and from the record of other government policies and of the economic conditions in which they were conducted. A new regime begins in the year in which a change of the tax code, or of the framework in which fiscal policy is conducted, comes into force, and lasts until the next such change.

For Germany the sample is divided into fourteen regimes, two of which overlap the long reform period of 1999-2005, because the reduction of income taxation carried out in those years was a multi-annual programme within which separate, discrete changes took place. The first regime covers the years in which the solidarity surcharge, reintroduced as a permanent surcharge on income and corporate taxes at 7.5 per cent of the underlying liability, was in force at that rate; it closes with the reduction of the surcharge to 5.5 per cent in 1998, at which level it largely remained thereafter. The second intervention period covers the gradual reduction of personal income tax rates, in the course of which the top marginal rate fell from 53 per cent in the late 1990s to 42 per cent by 2005, together with the accompanying reform of corporate taxation. Within it, the third regime isolates the corporate tax reform of 2001, which replaced the imputation system with a shareholder relief system and cut the corporate rate substantially. 

The fourth legal-economic-policy interval for Germany covers the Agenda 2010 and Hartz period, whose tax and social contribution changes were designed to raise labour market incentives and to reduce non-wage labour costs, and extends to the year preceding the increase of the standard value added tax rate. That increase, from 16 to 19 per cent and one of the largest post-war tax increases in Germany, is isolated as the fifth regime; the business tax reform of 2008, which lowered the effective corporate burden and revised the trade tax rules, as the sixth; and the introduction of the flat withholding tax of 25 per cent on most interest, dividends and capital gains of private investors as the seventh. 

The eighth regime covers the successive adjustments of inheritance taxation, of child allowances and of the tax treatment of family businesses, and the ninth the introduction of the statutory minimum wage, which was not a tax reform but had substantial payroll and social contribution implications. The tenth regime isolates the reform of investment taxation adopted in 2017 and effective from 2018, which changed the treatment of mutual funds and exchange traded funds; the eleventh covers the increases of the basic allowance and the bracket adjustments made to offset cold progression; and the twelfth the coronavirus pandemic, during which the standard and reduced value added tax rates were temporarily lowered, loss carrybacks were expanded, and the solidarity surcharge was abolished for the majority of taxpayers. 

The thirteenth regime covers the inflation relief acts, which repeatedly raised the basic allowance and adjusted the brackets upwards during the inflation surge, and the fourteenth the inflation compensation reform of 2024, the return of the value added tax on most restaurant services from 7 to 19 per cent, and the technical amendments of the Annual Tax Act of that year.

\begin{table}[htbp]
	\centering
	\caption{Dating of the policy-economic regimes, Germany. The intervention variables are indexed as in Table \ref{tab:psi_germany}.}
	\label{tab:regimes_germany}
	\small
	\begin{tabular}{ll}
		\hline
		Intervention variable & Period \\
		\hline
		$d_{1}$  & 1996-1998 \\
		$d_{2}$  & 1999-2005 \\
		$d_{3}$  & 2001-2002 \\
		$d_{4}$  & 2003-2006 \\
		$d_{5}$  & 2007 \\
		$d_{6}$  & 2008 \\
		$d_{7}$  & 2009 \\
		$d_{8}$  & 2010-2014 \\
		$d_{9}$  & 2015-2016 \\
		$d_{10}$ & 2017 \\
		$d_{11}$ & 2018-2019 \\
		$d_{12}$ & 2020-2021 \\
		$d_{13}$ & 2022-2023 \\
		$d_{14}$ & 2024 \\
		\hline
	\end{tabular}
\end{table}

For France the sample is divided into twelve regimes. The first covers the fiscal consolidation conducted under the Jupp\'e government, which raised some indirect taxes and social contributions. The second covers the progressive reductions of personal income tax rates and the targeted tax credits introduced under the Jospin government to support employment and consumption. The third economic-policy-legal interval isolates the year of the pension reform, with its adjustments to social contributions and to the tax treatment of retirement saving, and the fourth the reform of the corporate tax base, which changed the depreciation and corporate tax rules. 

The fifth regime for France covers the tax shield, which limited the combined burden of direct taxes to a share of income and was intended to reduce tax-induced capital flight, together with the law on work, employment and purchasing power of 2007, which exempted overtime pay, reduced inheritance taxation and introduced incentives for investment and employment. The sixth intervention period covers the reform of local business taxation, as well as the world economic crisis that began in 2008. The seventh covers the temporary surcharges on high incomes and on corporations introduced to address the post-crisis deficits, the higher taxation of top incomes and of capital income under the Hollande government, and the alignment of the taxation of capital income with that of labour income. The eighth regime covers the Responsibility and Solidarity Pact, under which employer social contributions and business tax burdens were gradually reduced, and within it the ninth isolates the year in which the reform of wealth taxation was announced. 

The tenth interval covers the introduction of the single flat levy of 30 per cent on many forms of capital income and the replacement of the broad wealth tax by the narrower tax on real estate wealth. The ninth regime is therefore nested in the eighth, and this is the only overlap among the French regimes; elsewhere, overlapping regimes occur only in Germany and in the Netherlands. The eleventh intervention interval covers the gradual reduction of the standard rate of corporate income tax, together with the emergency measures of the pandemic year, and the twelfth the energy-related tax reductions and fuel rebates of the inflation surge and the Finance Act of 2024, with its provisions on green investment, on the taxation of multinationals and on the implementation of international tax agreements.

\begin{table}[htbp]
	\centering
	\caption{Dating of the policy-economic regimes, France. The intervention variables are indexed as in Table \ref{tab:psi_france}.}
	\label{tab:regimes_france}
	\small
	\begin{tabular}{ll}
		\hline
		Intervention variable & Period \\
		\hline
		$d_{1}$  & 1996-1997 \\
		$d_{2}$  & 1998-2002 \\
		$d_{3}$  & 2003 \\
		$d_{4}$  & 2004-2005 \\
		$d_{5}$  & 2006-2007 \\
		$d_{6}$  & 2008-2010 \\
		$d_{7}$  & 2011-2013 \\
		$d_{8}$  & 2014-2017 \\
		$d_{9}$  & 2017 \\
		$d_{10}$ & 2018 \\
		$d_{11}$ & 2019-2022 \\
		$d_{12}$ & 2023-2024 \\
		\hline
	\end{tabular}
\end{table}

For Italy the sample is divided into eleven intervals, which form an exhaustive partition. The first regime covers the initial stage of the reform conducted under Finance Minister Visco, which simplified tax administration, strengthened the measures against evasion and introduced the regional tax on productive activities in 1998, together with the years immediately preceding it, in which that reform was prepared. The second economic-policy-legal regime covers the second year of the Visco reforms and the first phase of the Tremonti reforms, with their incentives for investment and their reductions of personal and business taxation. The third covers the years in which the number of personal income tax brackets was reduced together with the marginal rates, and in which the old corporate income tax was replaced by the modern corporate income tax. The fourth regime covers the fiscal package of the Prodi government, which broadened the tax bases and improved compliance, and the fifth -- the expansion of electronic reporting and financial monitoring, closing with the consolidation measures adopted during the sovereign debt crisis. 

The sixth regime covers the initial world crisis years and the reintroduction and expansion of municipal property taxation, the seventh the single year in which the monthly tax credit of eighty euro for employees was introduced together with reductions of labour taxation, and the eighth the reductions of labour costs and employer burdens that accompanied the Jobs Act, the gradual reduction of the corporate rate, and the special regime allowing wealthy foreign residents to pay a fixed annual amount on foreign income. 

The ninth intervention interval covers the broader use of the flat-rate regime for small businesses and the self-employed, the measures of the pandemic, the reduction of the number of personal income tax brackets, and the partial abolition of the regional tax on productive activities for many sole proprietors. The tenth economic-policy-legal regime isolates the year in which parliament approved the framework law delegating a broad restructuring of the tax system, and the eleventh the reduction from four to three personal income tax brackets, together with the first phase of the implementation of that framework.

\begin{table}[htbp]
	\centering
	\caption{Dating of the policy-economic regimes, Italy. The intervention variables are indexed as in Table \ref{tab:psi_italy}.}
	\label{tab:regimes_italy}
	\small
	\begin{tabular}{ll}
		\hline
		Intervention variable & Period \\
		\hline
		$d_{1}$  & 1996-2000 \\
		$d_{2}$  & 2001-2002 \\
		$d_{3}$  & 2003-2005 \\
		$d_{4}$  & 2006-2007 \\
		$d_{5}$  & 2008-2011 \\
		$d_{6}$  & 2012-2013 \\
		$d_{7}$  & 2014 \\
		$d_{8}$  & 2015-2018 \\
		$d_{9}$  & 2019-2022 \\
		$d_{10}$ & 2023 \\
		$d_{11}$ & 2024 \\
		\hline
	\end{tabular}
\end{table}

For Spain the sample is divided into eight regimes, which form an exhaustive partition. The first covers the reductions of personal income tax rates undertaken by the Aznar government and the measures adopted to support growth and convergence towards the euro area. The second covers the major reform of the personal income tax, which simplified the system and reduced the marginal rates. The third interval covers the reform of corporate and personal taxation, the expanded deductions and the favourable treatment granted to small and medium-sized enterprises, and closes with the reform that introduced the modern savings income base, under which capital income is taxed separately. 

The fourth economic-policy-legal regime for Spain covers the crisis years, including the temporary stimulus and relief that followed the global financial crisis, the increase of the standard value added tax rate and of the reduced rate, the consolidation package of the sovereign debt crisis, the temporary surcharge on personal income tax rates, and the further increase of the value added tax rates. The fifth regime covers the preparation and the first stage of the comprehensive reform that reduced personal income tax rates, lowered corporate taxation and simplified the system.
The sixth interval marks the completion of the reform, including reducing the corporate tax rate, introducing measures against evasion, aggressive tax planning as well as preparatory work on the taxation of digital platforms and financial transactions. 

The seventh economic-policy-legal regime represents the pandemic years, in which deferrals, guarantees and payment flexibility were granted, and in which the tax on certain digital services, the tax on acquisitions of shares in large listed companies and the higher rates on high incomes and some categories of capital income were introduced. The eighth covers the temporary reductions of taxation on electricity and energy products adopted against the inflation shock, the extraordinary levies on large energy companies and financial institutions, and the continuing adjustments related to energy taxation and to the implementation of international agreements, including the global minimum tax rules.

\begin{table}[htbp]
	\centering
	\caption{Dating of the policy-economic regimes, Spain. The intervention variables are indexed as in Table \ref{tab:psi_spain}.}
	\label{tab:regimes_spain}
	\small
	\begin{tabular}{ll}
		\hline
		Intervention variable & Period \\
		\hline
		$d_{1}$ & 1996-1998 \\
		$d_{2}$ & 1999-2001 \\
		$d_{3}$ & 2002-2007 \\
		$d_{4}$ & 2008-2012 \\
		$d_{5}$ & 2013-2014 \\
		$d_{6}$ & 2015-2019 \\
		$d_{7}$ & 2020-2021 \\
		$d_{8}$ & 2022-2024 \\
		\hline
	\end{tabular}
\end{table}

For the Netherlands the sample is divided into fourteen regimes. They cover the sample exhaustively, and at two points a regime overlaps the one that follows it, since one framework was still being implemented while the next was already in force. The first covers the shift of the financial responsibility for sickness and for disability from the collective schemes to the individual employer, together with the strong employment growth of the late 1990s. The second economic-policy-legal interval reflects the downturn that followed the end of that expansion and the administrative reorganisation of social security, which separated the administration of benefits from the provision of reintegration services. The third regime, which overlaps the second, represents the second and more restrictive wave of disability reform: the procedural obligations imposed on employers and employees during the first years of sickness, the lengthening of the period of continued wage payment, the tightening of the assessment criteria, and finally the replacement of the old disability scheme by one that distinguishes fully and permanently incapacitated workers from partially disabled workers with residual earning capacity. The fourth interval isolates the change of the surrounding framework, in particular the shortening of the maximum duration of unemployment benefits, the withdrawal of the fiscal advantages of early retirement and the reform of health insurance, all of which affected participation from that year onwards. 

The fifth regime for the Netherlands covers the financial crisis, the support given to the banking sector and the temporary schemes that allowed firms to reduce working time instead of dismissing workers. The sixth covers the reversal of that stance: the beginning of fiscal consolidation, the freezing of public sector wages and of the indexation of benefits, and the reform of the scheme for young disabled persons, which reoriented entitlement towards residual work capacity. The seventh interval isolates the second recession, which followed the European sovereign debt crisis, the additional consolidation measures taken in response to it and the decision to raise the statutory retirement age in steps from the following year. The eighth and the ninth regimes, which overlap in one year, cover the redesign of the Dutch welfare state: the extension of experience rating to workers without an employer, the decentralisation of social assistance, of sheltered employment and of long-term care to the municipalities, the new rules on flexible contracts and on dismissal, and a further reduction of the maximum duration of unemployment benefits. 

The tenth regime covers the years of sustained expansion, with rising employment, falling unemployment and increasing labour market tightness, and with the slowing of the scheduled increase of the state pension age. The eleventh covers the pandemic and the emergency support deployed in response to it, in the form of wage subsidies, of income support for the self-employed and of compensation for fixed costs; the twelfth the continuation and the gradual scaling down of that support during the successive lockdowns; and the thirteenth its termination, together with the energy price shock, which produced double-digit inflation in a labour market that was simultaneously at record tightness. The fourteenth regime controls for the policy response to that shock: the capping of energy prices and the compensation that accompanied it, an exceptionally large increase of the statutory minimum wage with the corresponding uprating of the benefits linked to it, and the beginning of the transition to the new pension system. The resulting subperiods are given in table \ref{tab:regimes_netherlands}.

\begin{table}[htbp]
	\centering
	\caption{Dating of the policy-economic regimes, the Netherlands. The intervention variables are indexed as in Table \ref{tab:psi_netherlands}.}
	\label{tab:regimes_netherlands}
	\small
	\begin{tabular}{ll}
		\hline
		Intervention variable & Period \\
		\hline
		$d_{1}$  & 1996-2000 \\
		$d_{2}$  & 2001-2003 \\
		$d_{3}$  & 2002-2006 \\
		$d_{4}$  & 2007 \\
		$d_{5}$  & 2008-2009 \\
		$d_{6}$  & 2010-2011 \\
		$d_{7}$  & 2012 \\
		$d_{8}$  & 2013-2015 \\
		$d_{9}$  & 2015-2016 \\
		$d_{10}$ & 2017-2019 \\
		$d_{11}$ & 2020 \\
		$d_{12}$ & 2021 \\
		$d_{13}$ & 2022 \\
		$d_{14}$ & 2023-2024 \\
		\hline
	\end{tabular}
\end{table}

For Poland the sample is divided into seven regimes, which form an exhaustive partition. The first covers the years preceding the reform of indirect taxation. Before 2004 the corporate income tax was changed almost every year, so that there is no justification for assigning separate intervention variables to individual years within that period. The second economic-policy-legal interval opens with the change of the value added tax in 2004 and covers the years in which the new structure of indirect taxation was in force. The third interval covers the modification of the personal income tax act in 2009, the two years in which that modification, and the conditions of the global financial crisis that coincided with it. The fourth group of years opens with the change of value added tax in 2011 and extends over the years in which the higher rates then adopted remained in force. The fifth regime opens with the introduction of an additional rate of the corporate income tax in 2017, which altered the taxation of smaller companies; the modification of the personal income tax act in 2019 falls within this regime and was not given a separate indicator, since it did not change the framework established at its beginning. The sixth regime opens with the modification of the personal income tax act in 2020 and coincides with the coronavirus pandemic. This interval covers also the further modification of the personal income tax act in 2022 falls within it. The seventh regime covers the years following the pandemic, in which the conditions under which the budget was planned and executed differed from those of the preceding three years. 	

\begin{table}[htbp]
	\centering
	\caption{Dating of the policy-economic regimes, Poland. The intervention variables are indexed as in Table \ref{tab:psi_poland}.}
	\label{tab:regimes_poland}
	\small
	\begin{tabular}{ll}
		\hline
		Intervention variable & Period \\
		\hline
		$d_{1}$ & 1996-2003 \\
		$d_{2}$ & 2004-2008 \\
		$d_{3}$ & 2009-2010 \\
		$d_{4}$ & 2011-2016 \\
		$d_{5}$ & 2017-2019 \\
		$d_{6}$ & 2020-2022 \\
		$d_{7}$ & 2023-2024 \\
		\hline
	\end{tabular}
\end{table}

The role of these variables can be seen most clearly by rewriting the model as a decomposition of fiscal policy into three parts. For $p=1$,
\begin{equation}\label{eq:6}
	{y}_{t}\;=\;\underbrace{{c}+\Phi_{1}{y}_{t-1}}_{\text{policy rule}}\;+\;\underbrace{\Psi{d}_{t}}_{\text{regime}}\;+\;\underbrace{{u}_{t}}_{\text{discretionary}} .
\end{equation}
The autoregressive part approximates the systematic fiscal policy rule, that is, the part of expenditure and revenue that is a function of the state of the economy in the preceding year. The term $\Psi{d}_{t}$ captures the nondiscretionary shift of the level of the fiscal variables that follows from the regime in force: a change of tax law, or of the framework of government policy, alters the level around which the fiscal variables fluctuate without being a discretionary reaction of the government in that particular year. What remains, ${u}_{t}$, is the genuinely discretionary, unanticipated part of policy.

In an SVAR without intervention variables, the second of these three components has nowhere to go. Being non-random, and not being captured by the autoregression, it is absorbed into the reduced-form residuals and hence into the identified structural shocks. The estimated shocks then mix regime changes with discretionary fiscal policy, and are correspondingly too large. Controlling for the regimes narratively removes this component from the residuals, so that the identified shocks are smaller and correspond more closely to what a government would recognise as a discretionary decision. This is also the sense in which the approach handles the phenomena that the fiscal foresight literature attributes to anticipation: a phased-in tax change, of the kind analysed by \textcite{house_phased-tax_2006} \textcite{mertens_dynamic_2013}, is a change of regime with a known date, and is represented here as such, rather than as a sequence of shocks that agents somehow foresee.

The partition structure of the intervention variables has a consequence for estimation that has to be made explicit, because it is not a nuisance of the specification but a direct expression of the process assumed here. In every country the regimes cover the sample exhaustively, and in Germany, France and the Netherlands some of them in addition overlap with others, or are nested within longer ones. What matters for estimation is not the presence of overlaps as such, but whether the constant regressor can be written as a linear combination of the intervention variables over the effective sample $t=p+1,\dots,T$. Writing $\mathcal{S}\subseteq\{1,\dots,M\}$ for a subset of the indicators, the relevant condition is the existence of an $\mathcal{S}$ with
\begin{equation}\label{eq:7}
	\sum_{m\in\mathcal{S}}d_{mt}\;=\;1\qquad\text{for every }t=p+1,\dots,T,
\end{equation}
that is, the existence of a subset of the regimes that partitions the effective sample exactly. It should be stressed that the specification \ref{eq:8} is the same for all six countries: in each of them the regimes cover the sample exhaustively, so that in each of them a basic exponential trend level and a sequence of economic and policy regime-specific effects altering it are assumed to coexist, and the countries differ only in whether the legislative calendar happens to make the two separately estimable. 

It is crucial to note that the economically interpretable value in each of the VARX equations is the sum of the basic exponential trend and the policy-economic regime. Taken separately, their values may be large in absolute value, but their sum is not.

For Spain, Poland and Italy the regimes do not overlap at all and \ref{eq:7} holds with $\mathcal{S}=\{1,\dots,M\}$. For France it holds with $\mathcal{S}$ equal to all the indicators except the one that isolates the nested regime, since removing the nested indicator leaves a partition. In these four countries the column of ones in \ref{eq:3} is an exact sum of columns of intervention variables, ${X}^{\prime}{X}$ is singular, and the constant and the regime coefficients are not separately identified by the data alone; the estimate is obtained through regularisation, by choosing the vector of coefficients with smallest Euclidean norm. For Germany and for the Netherlands, whose overlapping regimes are arranged so that no subset of them partitions the sample, \ref{eq:7} does not hold; the design matrix ${X}$ of \ref{eq:3} then has full column rank $K$, and the minimum-norm solution coincides with ordinary least squares. 

A second, and quite different, source of rank deficiency has to be avoided at the dating stage. The first observation on the gross growth rates is consumed as the initial condition of the autoregression, so that the effective sample begins in 1997 and a regime confined to 1996 alone would contribute a column of zeros to ${X}$. Such a regime would not merely be normalised, it would not be estimable at all, since the data would contain no observation in which it was in force. The dating adopted here therefore never isolates the first year of the sample: wherever the narrative record would suggest treating it separately, as in the Italian years preceding the Visco reform, it is merged with the regime that follows it. %The condition to be satisfied is thus not only \ref{eq:7} but also that every $\mathcal{T}_{m}$ intersect $\{p+1,\dots,T\}$.

The collinearity is deliberate and is retained. The process assumed for the gross growth rates is, equation by equation,
\begin{equation}\label{eq:8}
	y_{it}\;=\;\underbrace{c_{i}}_{\text{basic exponential trend}}\;+\;\underbrace{\sum_{l=1}^{p}\boldsymbol{\phi}_{il}^{\prime}{y}_{t-l}}_{\text{policy and propagation rule}}\;+\;\underbrace{\sum_{m=1}^{M}\psi_{im}d_{mt}}_{\text{policy-economic regime}}\;+\;u_{it},
	\qquad i=1,\dots,N,
\end{equation}
in which the constant $c_{i}$ is the exponential trend of the corresponding level series, since a constant in a gross growth rate is a constant proportional rate of growth, and $\psi_{im}$ is the shift of the level of that growth rate induced by the legal and policy framework in force. Both components are present in the process by assumption: an economy grows exponentially at some underlying rate, and the framework within which fiscal policy is conducted moves that rate up or down for as long as it is in force. Whether the two happen to be separately identified is decided not by economics but by the timing of legal procedures. In a country in which reforms were phased in, so that one framework was still being implemented while the next was already being legislated, the regimes overlap and the two components are separately estimable, as in Germany and the Netherlands; in a country in which each framework was replaced cleanly by the next, as in Italy, Spain and Poland, they are not, and the same is true when the only overlap is a short regime nested within a longer one, as in France. Nothing in the economics of the two components changes between these cases, and it would be arbitrary to specify a different process for a country merely because its legislative calendar produced no overlaps.

The problem is therefore of the same family as the simultaneous estimation of demand and supply, in that two structurally coexisting components have to be recovered from a single observed series. It differs from it in that there is no two-way causality here. What is at issue is the structural coexistence of the exponential growth rate and of the effect of the legal and economic-policy conditions, which becomes collinear whenever, for a given country, some subset of the regimes partitions the sample, as it does when the regimes do not overlap at all or overlap only through a nested regime.

The alternatives to retaining both components were considered and rejected. Omitting the constant attributes the whole of the trend to the regime coefficients and makes $\Psi$ incomparable with the estimates obtained for the countries whose designs have full rank. Omitting one intervention variable turns the remaining coefficients into relative constants, with respect to an arbitrarily chosen baseline regime, which is precisely what the reading of $\Psi$ in Section \ref{s4.3} must avoid, since what is of interest there is the magnitude of the level shift attributable to each regime and not its difference from a reference regime. Introducing artificial overlaps in order to restore full rank falsifies the dating and thereby the narrative identification itself. Truncating the sample so as to break the partition is not justifiable given $T-p=28$ effective observations.

What is done instead is to estimate \ref{eq:3} by regularised, minimum-norm least squares. Let the singular value decomposition of the design matrix be ${X}={U}{S}{V}^{\prime}$ with singular values $s_{1}\geq\dots\geq s_{K}\geq 0$, and let
\begin{equation}\label{eq:9}
	r\;=\;\#\{\,j:\;s_{j}>\tau\,\},\qquad \tau\;=\;\epsilon\cdot\max(T-p,K)\cdot s_{1},
\end{equation}
where $\epsilon$ is the machine precision of the floating-point representation. The estimator is
\begin{equation}\label{eq:10}
	\widehat{{B}}\;=\;{V}_{r}{S}_{r}^{-1}{U}_{r}^{\prime}{Y}\;=\;\bigl({X}^{\prime}{X}\bigr)^{+}{X}^{\prime}{Y},
\end{equation}
the subscript $r$ denoting the truncation to the first $r$ singular values and the corresponding singular vectors. The regularisation consists in discarding the directions of the column space along which the design carries no information, that is, in the truncation at $\tau$; among the infinitely many least squares solutions it selects the one of smallest Frobenius norm. In the implementation this is what the least squares routine of the numerical linear algebra library performs by default, and no ridge penalty is imposed, so that the fit itself is not shrunk. The covariance estimator in \ref{eq:4} is computed with the residuals of \ref{eq:10}.

The consequences of this choice are stated precisely in the following proposition, which covers the countries whose intervention variables satisfy \ref{eq:7}.

\begin{proposition}\label{prop:minnorm}
	Write ${e}_{j}$ for the $j$-th unit vector of $\mathbb{R}^{K}$, put $k_{m}:=1+Np+m$, and let $\mathcal{Z}$ be the set of indices of the intervention variables that are identically zero on the effective sample. Suppose that \ref{eq:7} holds for some subset $\mathcal{S}$ with $\mathcal{S}\cap\mathcal{Z}=\emptyset$, and that the null space $\mathcal{N}({X})$ is spanned by
	\begin{equation}\label{eq:11}
		{v}_{0}\;=\;{e}_{1}-\sum_{m\in\mathcal{S}}{e}_{k_{m}}
		\qquad\text{together with}\qquad
		\{{e}_{k_{m}}\}_{m\in\mathcal{Z}},
	\end{equation}
	so that $\operatorname{rank}({X})=K-1-\lvert\mathcal{Z}\rvert$. Then
	\begin{enumerate}
		\item the fitted values ${X}\widehat{{B}}$, the residuals $\widehat{{U}}$, the covariance estimate $\widehat{\Sigma}_{u}$, its Cholesky factor ${P}$, the structural shocks, the autoregressive matrices $\Phi_{1},\dots,\Phi_{p}$ and all impulse responses $\Theta_{h}$ are the same for every least squares solution and are therefore invariant with respect to the regularisation;
		\item the deterministic quantities identified by the data are the period levels ${c}+\Psi{d}_{t}$, $t=p+1,\dots,T$; in particular ${c}+\boldsymbol{\psi}_{m}$ is identified for $m\in\mathcal{S}$, where $\boldsymbol{\psi}_{m}$ is the $m$-th column of $\Psi$, while ${c}$ and $\boldsymbol{\psi}_{m}$, $m\in\mathcal{S}$, are not identified separately; for a regime $m\notin\mathcal{S}\cup\mathcal{Z}$, nested in a regime of $\mathcal{S}$, $\boldsymbol{\psi}_{m}$ is identified on its own and ${c}+\boldsymbol{\psi}_{m}$ is not; for $m\in\mathcal{Z}$ nothing is identified;
		\item the minimum-norm solution \ref{eq:10} is the unique least squares solution orthogonal to $\mathcal{N}({X})$, that is, the unique one for which
		\begin{equation}\label{eq:12}
			\widehat{{c}}\;=\;\sum_{m\in\mathcal{S}}\widehat{\boldsymbol{\psi}}_{m}
			\qquad\text{and}\qquad
			\widehat{\boldsymbol{\psi}}_{m}={0}\ \ \text{for }m\in\mathcal{Z}.
		\end{equation}
	\end{enumerate}
\end{proposition}

\begin{proof}
	Proof is moved to the appendix.
	%That ${X}{v}_{0}={0}$ is \ref{eq:7} written in matrix form, and ${X}{e}_{1+Np+m}={0}$ for $m\in\mathcal{Z}$ by the definition of $\mathcal{Z}$. Writing ${V}$ for the matrix whose columns are the vectors in \ref{eq:11}, the set of least squares solutions is $\{\widehat{{B}}+{V}{A}:{A}\in\mathbb{R}^{q\times N}\}$ with $q=1+\lvert\mathcal{Z}\rvert$, and every member of it yields the same ${X}\widehat{{B}}$, hence the same residuals, the same $\widehat{\Sigma}_{u}$ and, since the identification of Section \ref{s3.4} and the impulse responses of Section \ref{s3.5} are functions of $\widehat{\Sigma}_{u}$ and of $\Phi_{1},\dots,\Phi_{p}$ alone, the same shocks and impulse responses; this proves 1., the autoregressive blocks of $\widehat{{B}}$ being unaffected because every column of ${V}$ has zeros in the corresponding positions. Part 2. follows because adding ${V}{A}$ adds a common vector to ${c}$ and subtracts it from each $\boldsymbol{\psi}_{m}$ with $m\in\mathcal{S}$, leaving the sums ${c}+\boldsymbol{\psi}_{m}$ unchanged, while $\boldsymbol{\psi}_{m}$ with $m\in\mathcal{Z}$ may be shifted arbitrarily without affecting the fit at all. For 3., minimising $\lVert\widehat{{B}}+{V}{A}\rVert_{F}^{2}$ over ${A}$ gives ${A}=-({V}^{\prime}{V})^{-1}{V}^{\prime}\widehat{{B}}$ and hence ${V}^{\prime}\widehat{{B}}={0}$ at the optimum; writing this out column by column with \ref{eq:11} gives \ref{eq:12}, the columns of ${V}$ being mutually orthogonal here.
\end{proof}

%	\begin{remark}
	%		Equation \ref{eq:12} can be verified directly in the results reported below. For Spain, Poland and Italy, where $\mathcal{S}=\{1,\dots,M\}$, the entries of any row of $\widehat{\Psi}$ in Tables \ref{tab:psi_spain}, \ref{tab:psi_poland} and \ref{tab:psi_italy} sum to the constant of the corresponding equation in Tables \ref{tab:var_spain}, \ref{tab:var_poland} and \ref{tab:var_italy}: in the GDP equation the sums are $1.10073$, $0.63911$ and $0.38520$, which are exactly the reported constants. For France the same holds once the indicator of the nested regime is omitted from the sum, the GDP equation giving $0.89608$ against a constant of $0.89609$. Germany and the Netherlands are the only countries for which no such relation holds, their designs being of full rank: in the Dutch GDP equation the fourteen entries of $\widehat{\Psi}$ sum to $-0.08019$ against a constant of $0.54589$.
	%	\end{remark}

Part (i) of Proposition \ref{prop:minnorm} is what makes the whole construction usable. The objects of interest of this paper, the impulse responses of Section \ref{s3.5} and the multipliers of Section \ref{s3.6}, do not depend on the normalisation at all; only the split of the estimated deterministic level between the constant and the regime coefficients does. The same applies to the bootstrap of Section \ref{s3.7}: every artificial sample is generated from the estimated constant, the estimated autoregressive matrices and the estimated $\Psi$, using the same matrix of intervention variables, and is then re-estimated by the same regularised least squares, so that each replication is subject to the identical rank deficiency and the identical normalisation, and the resampled impulse responses are comparable across replications by part (i). Replications in which the bootstrap covariance matrix fails to be positive definite, and for which the Cholesky factorisation therefore does not exist, are discarded, and the number of successful replications is reported.

Part (ii) is the reason for the reading of the estimates in Sections \ref{s4.3} and \ref{s4.4}. What the data determine for the four countries with rank-deficient designs is the regime-specific deterministic component ${c}+\boldsymbol{\psi}_{m}$, and it is the plausibility of these implied components, rather than of the constant taken in isolation, that is used here to judge between specifications. The criterion is that the implied deterministic part of the gross growth rate should be positive and of a plausible order of magnitude, that is, close to unity, since the endogenous variables are gross growth rates and an exponential trend far from unity would imply a rate of growth that no economy in the sample exhibited. Judged by this criterion, the specification adopted here performs better than the two alternatives. A model estimated without intervention variables leaves the level effects of the successive frameworks in the residuals and in the autoregressive part, and the implied trend absorbs them; a model in which the regimes of a country whose legislative record shows no overlaps are given an artificial one, so as to restore full rank without any support in the legislative record, yields implied trends that are less plausible still, and the overlapping dates are not defensible in the narrative account given above. The specification with the full, exhaustive partition and regularised estimation is retained on this ground, in addition to the identification argument of the preceding paragraphs.

\subsection{Identification of the structural shocks}\label{s3.4}

The structural shocks are identified recursively. Let ${P}$ be the lower-triangular Cholesky factor of the reduced-form covariance matrix,
\begin{equation}\label{eq:13}
	\widehat{\Sigma}_{u}\;=\;{P}{P}^{\prime},
	\qquad
	{u}_{t}\;=\;{P}\boldsymbol{\varepsilon}_{t},
	\qquad
	\mathrm{E}\bigl(\boldsymbol{\varepsilon}_{t}\boldsymbol{\varepsilon}_{t}^{\prime}\bigr)={I}_{N},
\end{equation}
so that the estimated structural shocks are obtained as $\widehat{\boldsymbol{\varepsilon}}_{t}={P}^{-1}\widehat{{u}}_{t}$. With the ordering of \ref{eq:1}, the two categories of government expenditure are placed first, followed by private consumption, gross capital formation, government revenue and, last, GDP. The recursive structure therefore implies that structural innovations to defence and to non-defence government expenditure are not affected contemporaneously by innovations to the private variables, to government revenue or to output, while GDP responds contemporaneously to all the shocks of the system. In the sample used here the covariance matrix of the estimated structural shocks is the identity matrix up to numerical error, of the order of $10^{-16}$, which confirms the normalisation.

\subsection{Impulse responses and the scaling of the shock}\label{s3.5}

The moving average representation of the reduced form is obtained recursively from the autoregressive matrices,
\begin{equation}\label{eq:14}
	\Xi_{0}={I}_{N},
	\qquad
	\Xi_{h}=\sum_{l=1}^{\min(h,p)}\Phi_{l}\,\Xi_{h-l},
	\qquad h=1,\dots,H,
\end{equation}
and the structural impulse response matrices are
\begin{equation}\label{eq:15}
	\Theta_{h}\;=\;\Xi_{h}{P},
	\qquad h=0,1,\dots,H .
\end{equation}
The element $[\Theta_{h}]_{ij}$ is the response of the gross growth rate of variable $i$ at horizon $h$ to a unit structural shock in equation $j$. Because the endogenous variables are gross growth rates, these responses are deviations of growth rates, not of levels.

The responses are then rescaled to an impulse of a given size. Let $\delta$ denote the size of the impulse and $\sigma_{j}$ the standard deviation of the $j$-th structural shock, which equals unity under the normalisation \ref{eq:13}. The scaled responses are
\begin{equation}\label{eq:16}
	\theta^{(\delta)}_{i,j,h}\;=\;\frac{\delta}{\sigma_{j}}\,[\Theta_{h}]_{ij},
\end{equation}
and they are interpreted as the reaction of the system to a one-period increase of the gross growth rate of variable $j$ by $\delta$, that is, by $100\delta$ percentage points. In the results reported below the shocked variable is non-defence general government expenditure, $j=2$, and the impulse sizes considered are $\delta\in\{0.01,0.02,\dots,0.10\}$, that is, from one to ten percentage points.

\subsection{The multiplicative multiplier}\label{s3.6}

When the variables of the model are gross growth rates, the cumulation of a response over a horizon is a product and not a sum. If the growth rate of variable $i$ deviates from its baseline path by $\theta^{(\delta)}_{i,j,h}$ in each period $h$, then the level of that variable at the end of horizon $H$ deviates from the baseline level by the factor
\begin{equation}\label{eq:17}
	\prod_{h=0}^{H}\Bigl(1+\theta^{(\delta)}_{i,j,h}\Bigr).
\end{equation}
The multiplier proposed here is the ratio of the cumulated relative level deviation of the variable of interest to the cumulated relative level deviation of the variable that has been shocked,
\begin{equation}\label{eq:multiplier}
	m_{i}\bigl(H,\delta\bigr)\;=\;
	\frac{\displaystyle\prod_{h=0}^{H}\Bigl(1+\theta^{(\delta)}_{i,j,h}\Bigr)}
	{\displaystyle\prod_{h=0}^{H}\Bigl(1+\theta^{(\delta)}_{j,j,h}\Bigr)} ,
	\qquad j=2 .
\end{equation}
For $i=6$ this is the fiscal multiplier of GDP with respect to an impulse of non-defence general government expenditure; for the remaining values of $i$ the same expression gives the corresponding multiplier of defence expenditure, private consumption, gross capital formation and government revenue. By construction $m_{j}(H,\delta)\equiv 1$ for the shocked variable itself, which provides a numerical check on the computation.

Two properties of \ref{eq:multiplier} deserve emphasis. First, it is consistent with the form of the data: it is the appropriate way of cumulating responses expressed as gross growth rates, and it does not require any discounting, any ad hoc conversion factor, or any auxiliary estimate of an average expenditure share of output. It is a nominal object, computed from nominal series, and therefore directly comparable with the nominal quantities with which a government budget operates. Second, and unlike the additive multiplier, it is nonlinear in the size of the shock. Expanding \ref{eq:17}, both the numerator and the denominator of \ref{eq:multiplier} are polynomials of degree $H+1$ in $\delta$,
\begin{equation}\label{eq:18}
	\prod_{h=0}^{H}\Bigl(1+\delta\,\tilde{\theta}_{i,j,h}\Bigr)
	\;=\;1+\delta\sum_{h=0}^{H}\tilde{\theta}_{i,j,h}
	+\delta^{2}\!\!\sum_{0\le h<k\le H}\!\!\tilde{\theta}_{i,j,h}\tilde{\theta}_{i,j,k}+\dots,
	\qquad \tilde{\theta}_{i,j,h}=[\Theta_{h}]_{ij}/\sigma_{j},
\end{equation}
so that the ratio does not reduce to a constant independent of $\delta$. The size of the shock therefore matters for the size and for the shape of the multiplier, which is a property that cannot be recovered from a specification in levels, in first differences, or in their logarithmic counterparts.

A further reason for defining the multiplier as in \ref{eq:multiplier}, rather than as a ratio of impulse responses, is numerical. Responses of gross growth rates oscillate around zero and change sign from one horizon to the next, and their magnitudes fall by several orders of magnitude as the horizon grows; a pointwise ratio of two such responses is therefore not an informative quantity. The multiplicative form aggregates the entire path up to $H$ and, as reported in Section \ref{s4}, produces stable values for every horizon considered.

\subsection{Bootstrap inference}\label{s3.7}

Confidence intervals for the impulse responses are obtained by a residual-based, recursive-design bootstrap. The estimated residuals are centred, $\widehat{{u}}^{c}_{t}=\widehat{{u}}_{t}-\overline{\widehat{{u}}}$, and resampled with replacement. For each replication $b=1,\dots,B$ a pseudo-sample is generated recursively from
\begin{equation}\label{eq:19}
	{y}^{(b)}_{t}\;=\;\widehat{{c}}\;+\;\sum_{l=1}^{p}\widehat{\Phi}_{l}\,{y}^{(b)}_{t-l}\;+\;\widehat{\Psi}\,{d}_{t}\;+\;{u}^{(b)}_{t},
\end{equation}
with the first $p$ observations of the actual sample as initial values and with the actual, unresampled intervention variables, so that the regime structure is held fixed across replications. The model is re-estimated on each pseudo-sample and the structural impulse responses $\Theta^{(b)}_{h}$ are recomputed; replications in which the bootstrap covariance matrix is not positive definite, and for which the Cholesky decomposition therefore fails, are discarded. The reported bands are the pointwise percentile intervals, that is, the $\alpha/2$ and $1-\alpha/2$ empirical quantiles of $\{\Theta^{(b)}_{h}\}_{b=1}^{B}$. In the computations reported below $B=1000$ and $\alpha=0.05$.

\subsection{Interpreting the large estimated constants: the defence expenditure equation}\label{s4.5}

The constants reported in Tables \ref{tab:var_germany}-\ref{tab:var_poland} are large in the defence expenditure equation, and in several countries far larger than unity.%: $2.49$ for Germany, $3.75$ for France, $5.34$ for Italy, $37.47$ for the Netherlands, and $-9.997$ for Spain. 
Taken at face value these would be impossible as growth factors. They are not to be read that way, for two separate reasons.

The first is the normalisation. By Proposition \ref{prop:minnorm} the constant of an equation is, in the four countries with rank-deficient designs, the sum of the coefficients of the regimes in $\mathcal{S}$ in that equation, that is, of all the regimes in Italy, Spain and Poland and of all but the nested ninth regime in France, so its magnitude scales with the number of regimes and with the size of the level effects and carries no independent interpretation. In Germany and in the Netherlands, whose designs have full rank, the constant is estimated in its own right and the first reason does not apply; the second does, and it is sufficient. The second, and more important, is that in a dynamic model the constant is not the trend. The deterministic component of the gross growth rate implied by regime $m$ is the stationary mean of \ref{eq:2} with ${d}_{t}$ held at that regime,
\begin{equation}\label{eq:20}
	\boldsymbol{\mu}_{m}\;=\;\Bigl({I}_{N}-\sum_{l=1}^{p}\Phi_{l}\Bigr)^{-1}\bigl({c}+\boldsymbol{\psi}_{m}\bigr),
\end{equation}
and it is $\boldsymbol{\mu}_{m}$, not ${c}$, that is invariant to the normalisation and comparable across countries. %Computed from Tables \ref{tab:psi_spain}-\ref{tab:var_poland}, the implied mean gross growth factors of GDP lie between $0.93$ and $1.05$ in every country and in almost every regime, which is the criterion of Section \ref{s3.3} and is satisfied throughout. 
The large constants of the defence equation are accompanied, in the same equation, by autoregressive coefficients of comparable size and opposite sign, and \ref{eq:20} nets the two against each other.

What remains after this netting is nevertheless a genuine feature of the data: the implied deterministic component of defence expenditure is far more dispersed across regimes than that of any other variable. %ranging from about $0.63$ to about $1.66$ against the range of $0.93$ to $1.05$ for GDP. 
Two circumstances account for most of it.

The first is the size of the defence sector. A gross growth rate is a ratio, and its volatility is governed by the size of the base rather than by the size of the change. In a country that had reduced its armed forces to a residual share of general government expenditure in the decades after the end of the Cold War, an absolute increase that is small relative to the budget as a whole is very large relative to the defence function itself, and appears in the model as a growth factor far from unity. The rearmament that followed 2022 is visible in exactly this form, and in exactly the order that this reasoning predicts.%: the implied defence growth factor of the last regime is about $1.66$ for Germany, $1.39$ for Poland, $1.30$ for the Netherlands, $1.13$ for Italy, $1.09$ for Spain and $1.05$ for France, so that 
The countries whose defence sector had been run down furthest before 2020, and which stood closest to the eastern border, show the largest values, while France, which maintained a large defence establishment and an independent deterrent throughout, shows the smallest. In the German series the thirteen regimes preceding the last lie between $0.85$ and $1.23$, so that the jump to $1.66$ at the end of the sample stands alone against an otherwise flat background. The Dutch series is the more dispersed of the two: seven of its fourteen regimes have implied factors below unity, that is, a defence sector contracting in nominal terms while the rest of the economy grew, the lowest of them being $0.63$ in 2021, and it is against that background that the $1.30$ of 2023--2024 is to be read, as a reversal rather than as the continuation of a trend. Only the pandemic regime, at $1.55$, exceeds it, and there the factor owes as much to the denominator as to defence spending itself, since the implied factor of Dutch GDP in that regime, $0.99$, is the lowest of the fourteen.

The second reason is the budgetary position. In a debt-burdened country, defence expenditure is a residual claimant on the budget. Debt service, pensions and the other legally mandated transfers are fixed in the short run, so that when consolidation is required the adjustment falls disproportionately on the discretionary functions, of which defence is the largest; and when external circumstances force an increase, it has to be made abruptly rather than by a gradual accumulation, because it was not built into the medium-term plans. Such a series does not have a trend in any useful sense: it has episodes. A model of the form \ref{eq:2} represents an episodic series by a large deterministic component in the regimes in which the episodes fall, offset by strong mean reversion in the autoregressive part, and this is precisely the configuration observed in the defence equations of Italy, Spain and, to a lesser extent, Germany. It is also why the defence multipliers reported in Section \ref{s4.2} should be read with more caution than the others: a shock to a small and episodic series is estimated from few effective episodes, and the corresponding bootstrap intervals of Section \ref{s3.7} are correspondingly wide.

Neither circumstance is an argument against the specification. Both are arguments for reading the constants through \ref{eq:20} rather than directly, and for treating the dispersion of the implied defence components as information about the defence budgets of the six countries rather than as a symptom of misspecification.

\section{Results}\label{s4}

\subsection{The nonlinearity of the multiplier with respect to the size of the shock}\label{s4.1}

The first result concerns the dependence of the multiplier on the size of the shock.\begin{figure}[htbp]
	\centering
	% The figure files are produced by the estimation code as SVG; a PDF or PNG
	% version of the same file has to be available for inclusion.
	\includegraphics[width=\textwidth]{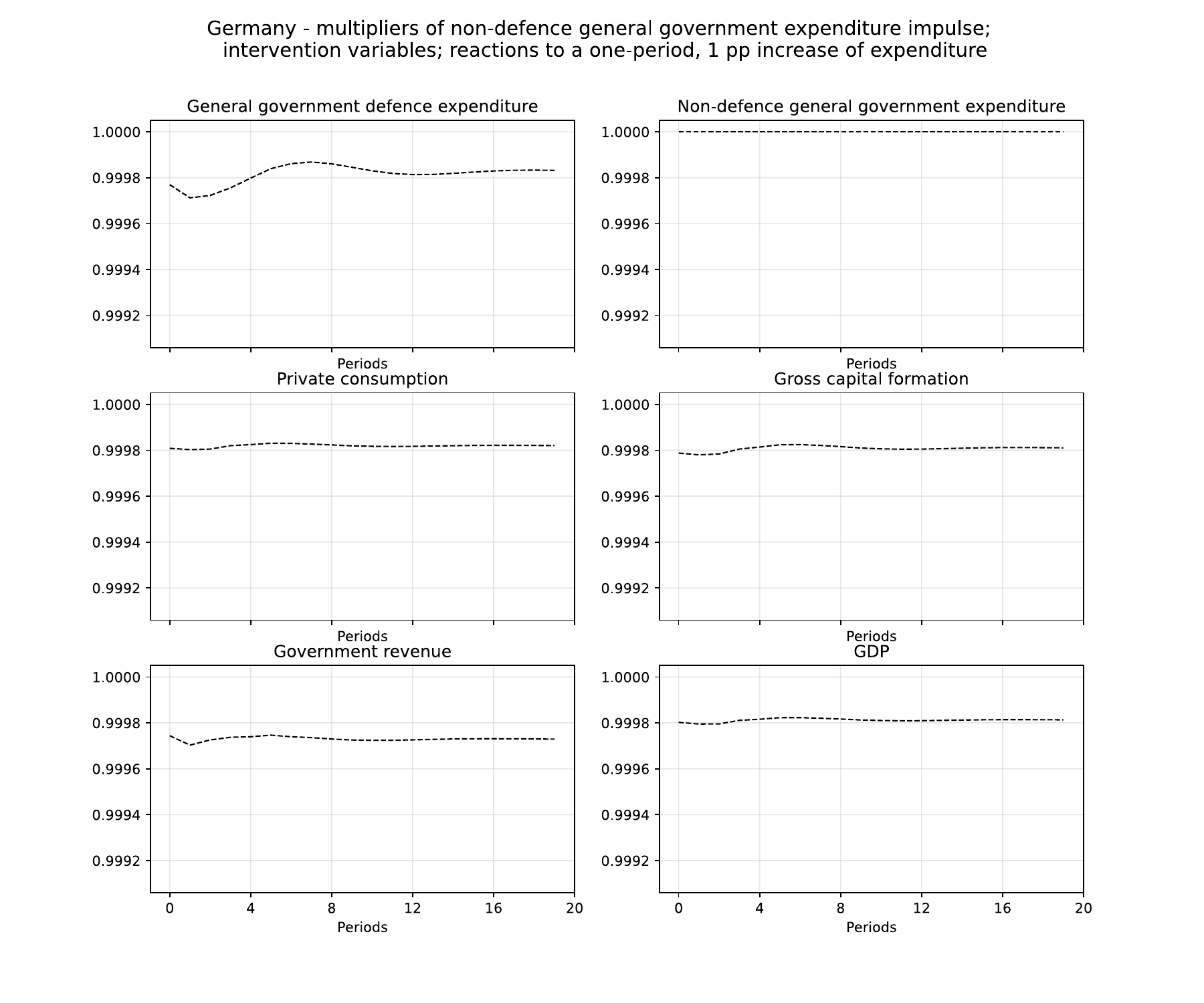}
	\caption{Germany. Multipliers of an impulse of non-defence general government expenditure, with intervention variables; reaction to a one-period increase of expenditure of 1 percentage point. The units of the impulse are percentage points.}
	\label{fig:de001}
\end{figure} Figures \ref{fig:de001} and \ref{fig:de003} present, for Germany, the multipliers \ref{eq:multiplier} of all six variables computed for a one-period increase of non-defence general government expenditure of one percentage point and of three percentage points respectively, as functions of the horizon.

\begin{figure}[htbp]
	\centering
	\includegraphics[width=\textwidth]{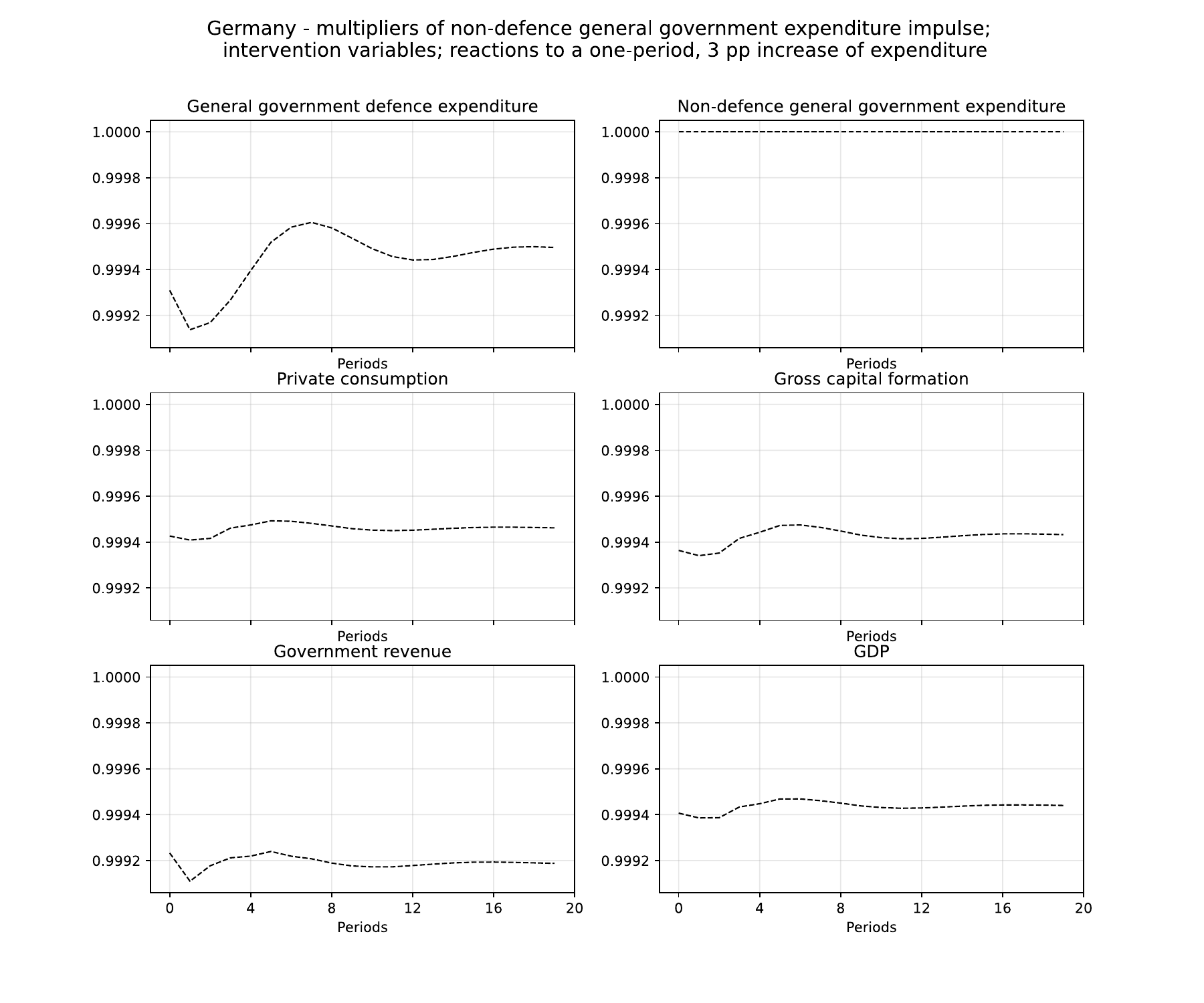}
	\caption{Germany. Multipliers of an impulse of non-defence general government expenditure, with intervention variables; reaction to a one-period increase of expenditure of 3 percentage points. The units of the impulse are percentage points.}
	\label{fig:de003}
\end{figure}

Two features are visible. First, the size of the multiplier differs for all horizons: the value of the multiplier obtained for an impulse of one percentage point is not the value obtained for an impulse of three percentage points, at any horizon. Second, the difference is not proportional to the difference in the size of the impulse. The shapes of the multiplier corresponding to a one-percentage-point shock and of the multiplier corresponding to a three-percentage-point shock differ nonlinearly, which is precisely the property implied by \ref{eq:18} and which an additive multiplier, being linear in both the numerator and the denominator, cannot exhibit.

\subsection{Differences of multiplier sizes across countries}\label{s4.2}

The second result concerns cross-country heterogeneity. Figures \ref{fig:m_de}-\ref{fig:m_pl} report the multipliers of an impulse of non-defence general government expenditure of one percentage point for the six largest economies of the European Union: Germany, France, Italy, Spain, the Netherlands and Poland.

\begin{figure}[htbp]
	\centering
	\includegraphics[width=\textwidth]{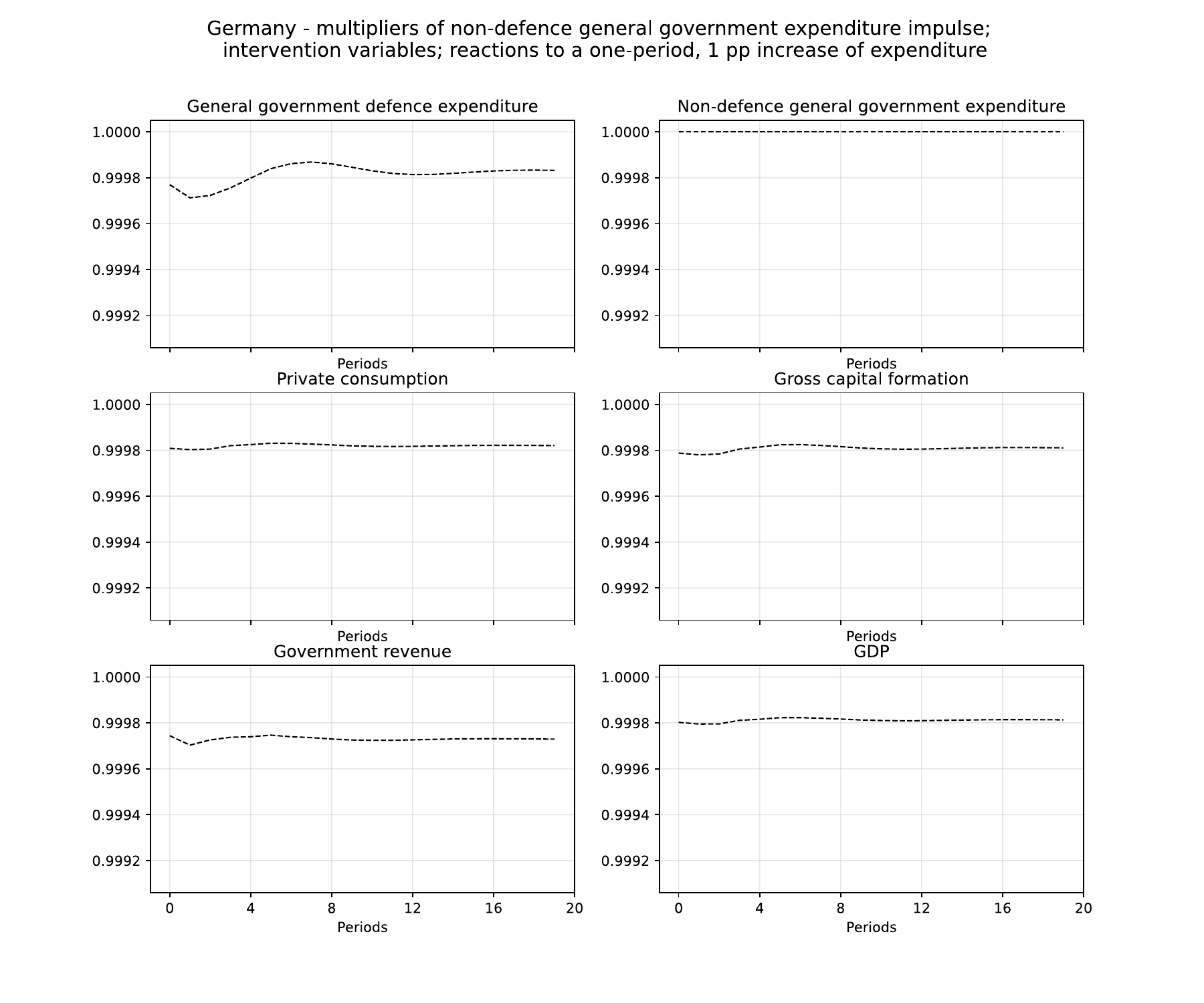}
	\caption{Germany. Multipliers of an impulse of non-defence general government expenditure of 1 percentage point.}
	\label{fig:m_de}
\end{figure}

\begin{figure}[htbp]
	\centering
	\includegraphics[width=\textwidth]{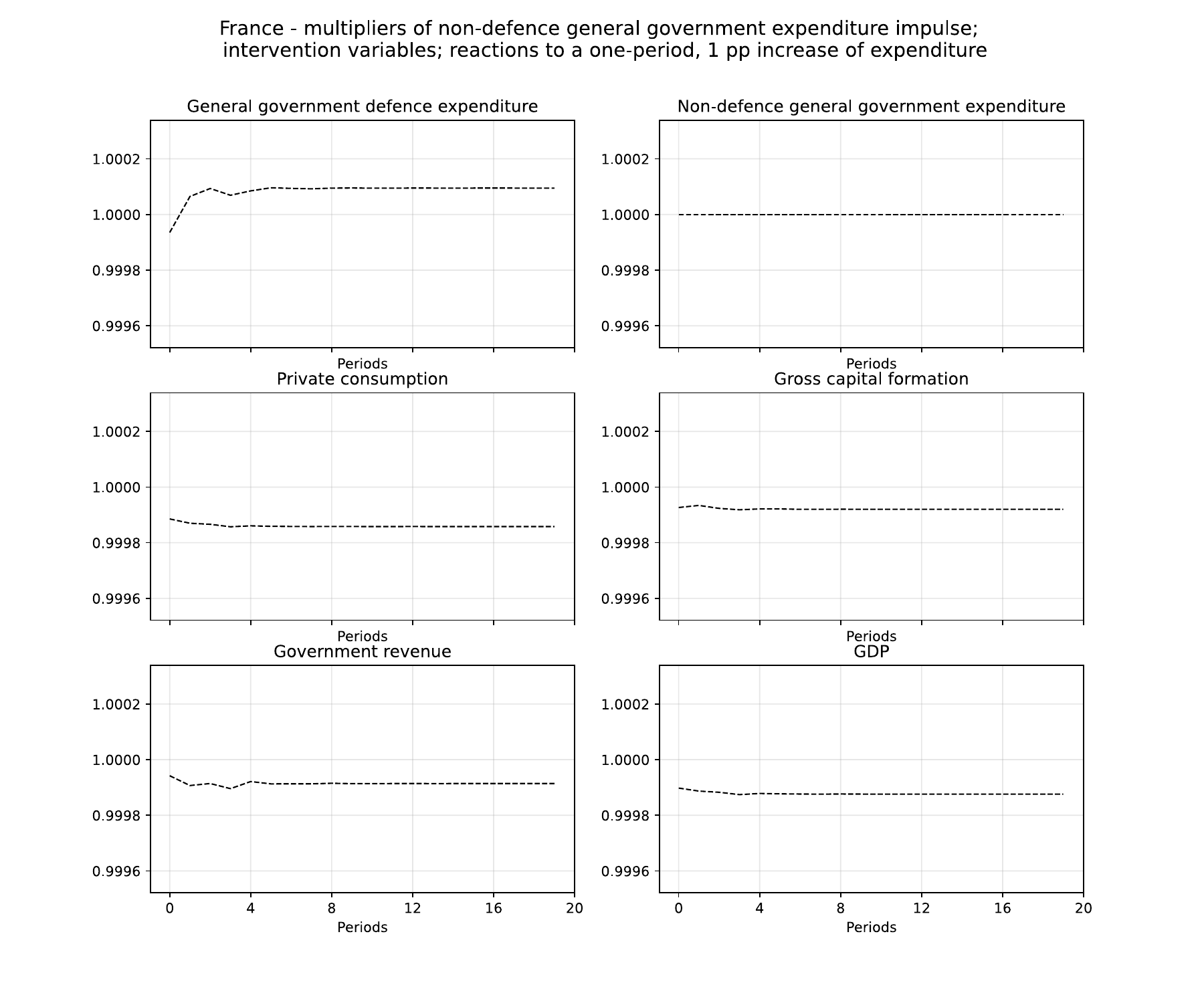}
	\caption{France. Multipliers of an impulse of non-defence general government expenditure of 1 percentage point.}
	\label{fig:m_fr}
\end{figure}

\begin{figure}[htbp]
	\centering
	\includegraphics[width=\textwidth]{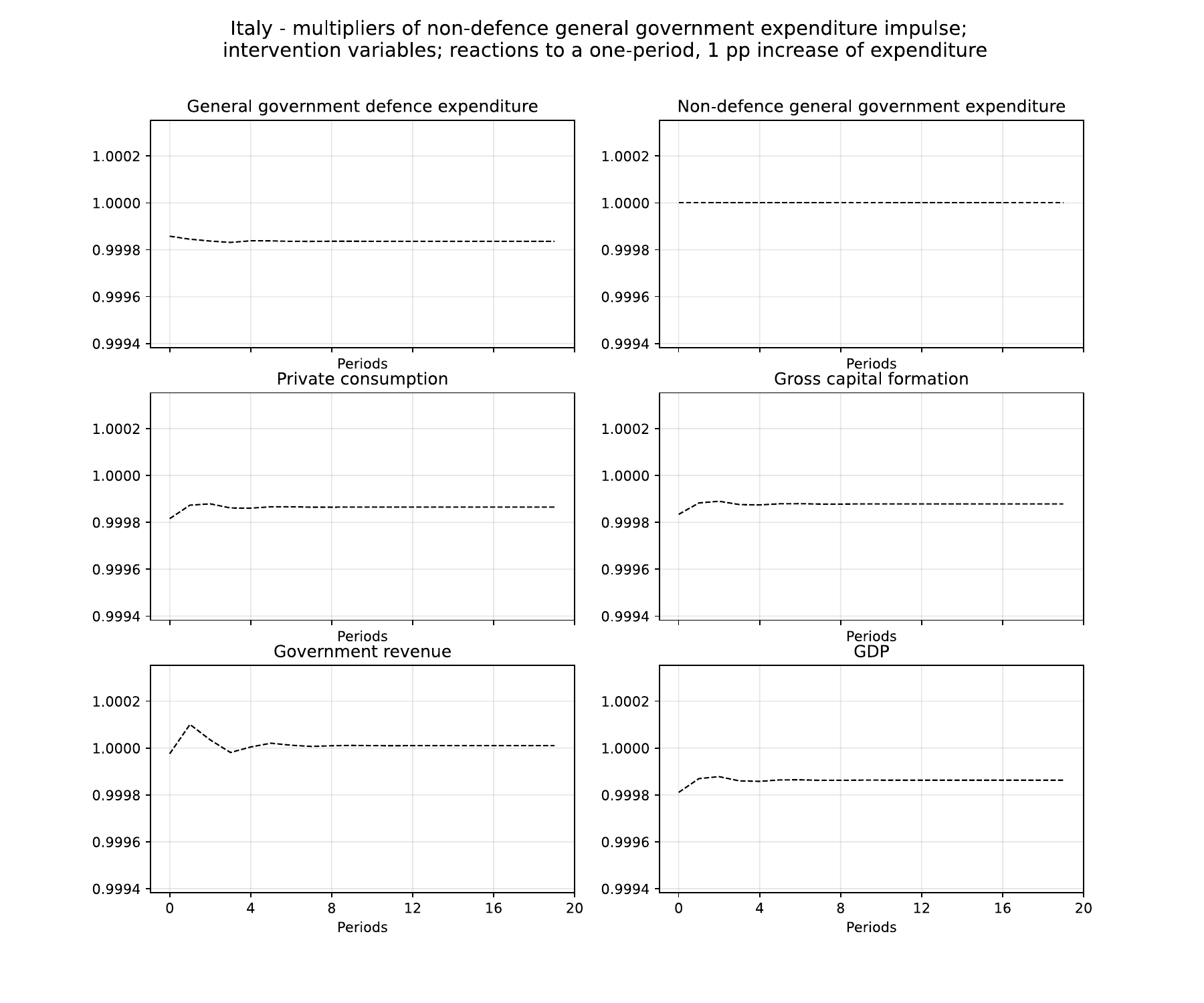}
	\caption{Italy. Multipliers of an impulse of non-defence general government expenditure of 1 percentage point.}
	\label{fig:m_it}
\end{figure}

\begin{figure}[htbp]
	\centering
	\includegraphics[width=\textwidth]{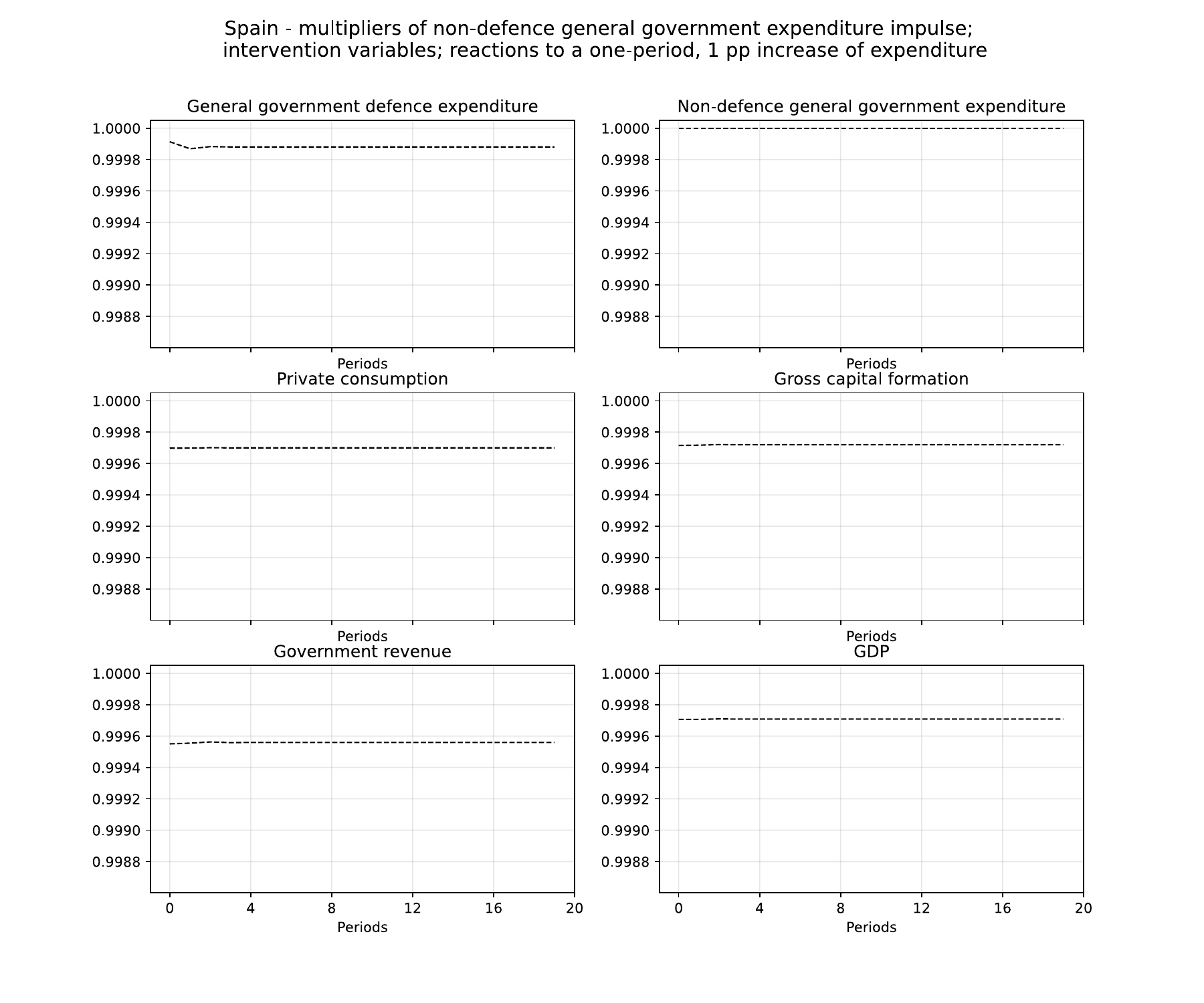}
	\caption{Spain. Multipliers of an impulse of non-defence general government expenditure of 1 percentage point.}
	\label{fig:m_es}
\end{figure}

\begin{figure}[htbp]
	\centering
	\includegraphics[width=\textwidth]{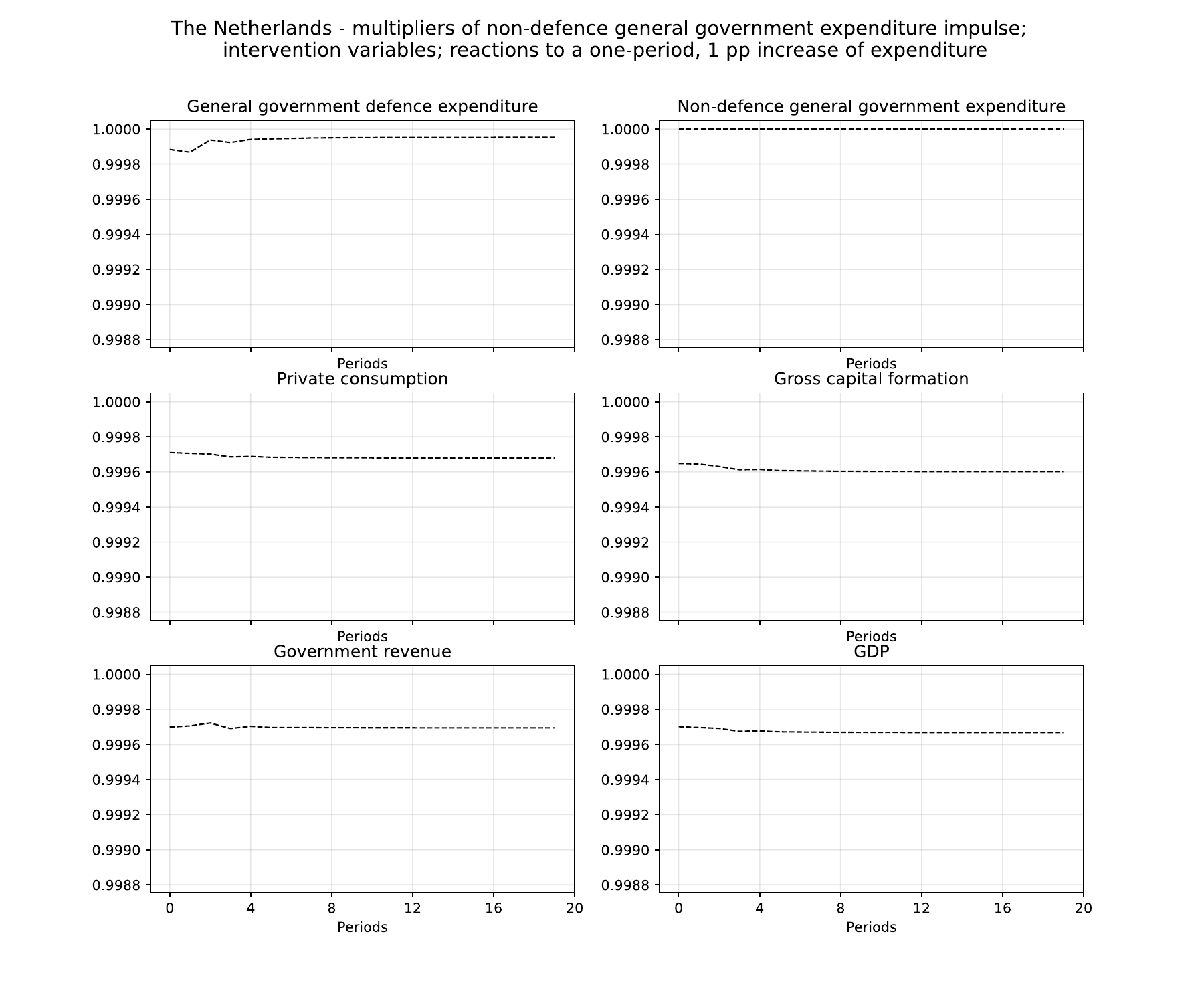}
	\caption{The Netherlands. Multipliers of an impulse of non-defence general government expenditure of 1 percentage point.}
	\label{fig:m_nl}
\end{figure}

\begin{figure}[htbp]
	\centering
	\includegraphics[width=\textwidth]{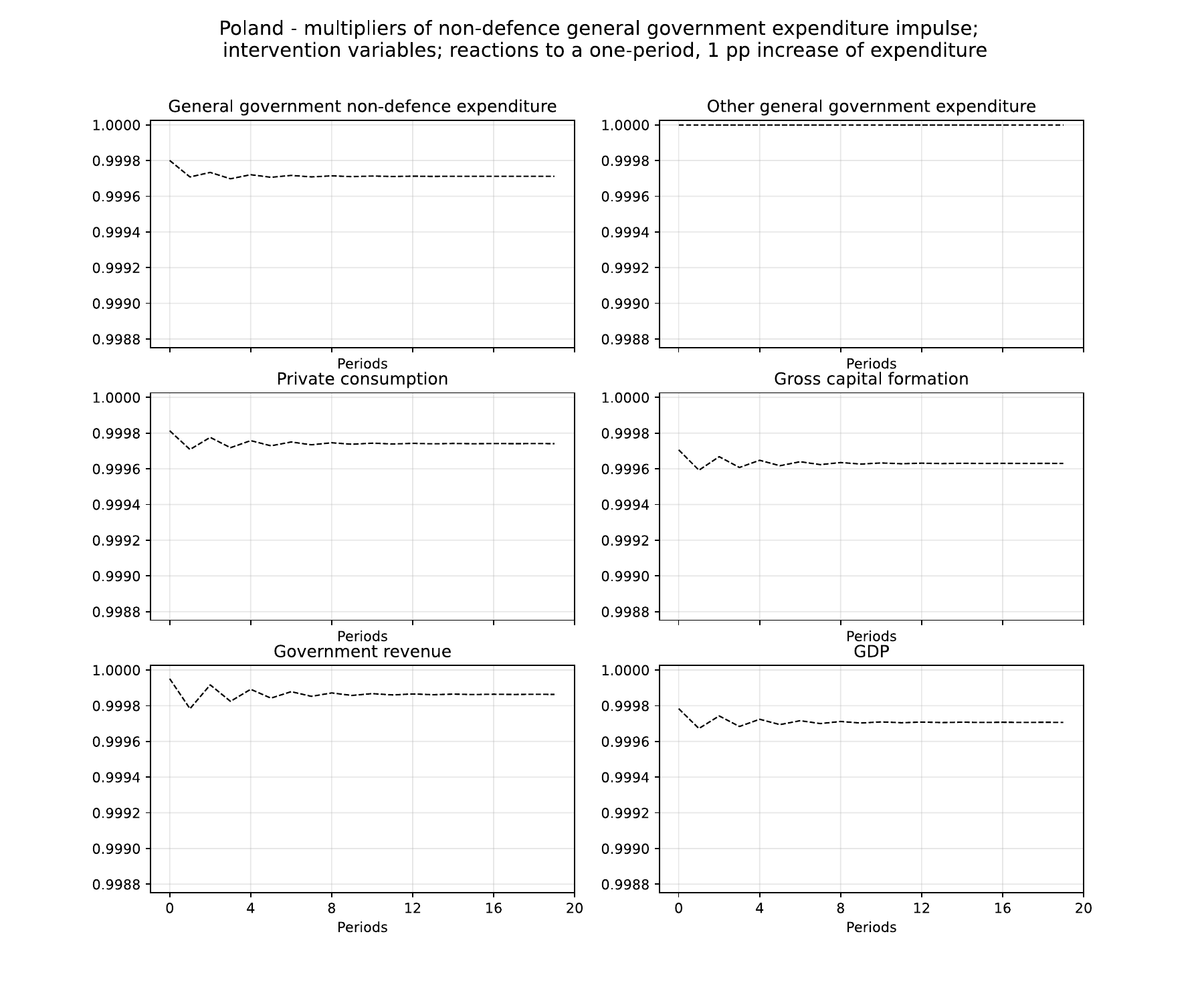}
	\caption{Poland. Multipliers of an impulse of non-defence general government expenditure of 1 percentage point.}
	\label{fig:m_pl}
\end{figure}

As is clearly shown by the Figures \ref{fig:m_de}-\ref{fig:m_pl}, fiscal multipliers differ across countries. The cause is the structural and institutional differences between these economies, which the intervention variables control for only in so far as they concern the policy and economic regimes of each country separately. For all six economies, however, the multipliers stabilise within the twenty-year horizon.

\subsection{Estimates of the policy-economic regime-controlling intervention variables}\label{s4.3}

The estimates of the matrix $\Psi$ of coefficients on the intervention variables are reported, country by country, in Tables \ref{tab:psi_spain}-\ref{tab:psi_germany}. Since the rows of $\Psi$ correspond to the equations of the model, the magnitude of the entries in a given row measures how strongly the policy and economic regime shifts the level of the corresponding variable, and hence how much of the variation of that variable would have been attributed to the autoregressive and to the discretionary part of the model had the regimes not been controlled for.

In the case of Germany, the large values in the first equation are concentrated in the last three years of the sample, which points to a large regime change at the end of the sample. In general, changes in tax and economic conditions altered the German economy only to a small extent.

\begin{table}[htbp]
	\centering
	\caption{Estimated coefficients $\hat{\Psi}$ on the policy-economic regime intervention variables, Germany. Rows correspond to the equations of the model, columns to the consecutive regimes $d_{1},\dots,d_{14}$.}
	\label{tab:psi_germany}
	\begin{adjustbox}{max width=\textwidth}
		\small
		\begin{tabular}{lrrrrrrrrrrrrrr}
			\hline
			Equation & $d_{1}$ & $d_{2}$ & $d_{3}$ & $d_{4}$ & $d_{5}$ & $d_{6}$ & $d_{7}$ & $d_{8}$ & $d_{9}$ & $d_{10}$ & $d_{11}$ & $d_{12}$ & $d_{13}$ & $d_{14}$ \\
			\hline
			Defence expenditure & $-$0.09147 & $-$0.03572 & $-$0.06441 & $-$0.06982 & $-$0.22868 & $-$0.00892 & $-$0.02571 & $-$0.08986 & $-$0.05317 & 0.00463 & 0.00399 & 0.11017 & 0.19221 & 0.09039 \\
			Non-defence expenditure & $-$0.01477 & $-$0.01030 & 0.00479 & $-$0.01842 & $-$0.02069 & 0.01891 & 0.00204 & $-$0.02706 & 0.01370 & $-$0.01520 & $-$0.01646 & 0.00382 & $-$0.03078 & $-$0.00592 \\
			Private consumption & 0.00228 & 0.00053 & 0.00258 & 0.00274 & 0.00505 & 0.00739 & 0.00796 & 0.00645 & 0.00557 & 0.00364 & 0.00340 & $-$0.00691 & $-$0.00635 & $-$0.00199 \\
			Gross capital formation & $-$0.00754 & 0.00010 & 0.00575 & 0.01128 & 0.00552 & 0.00064 & $-$0.00258 & 0.00616 & $-$0.00201 & 0.00386 & $-$0.01049 & $-$0.01460 & $-$0.00342 & 0.01315 \\
			Government revenue & $-$0.04593 & $-$0.01879 & $-$0.05086 & $-$0.02075 & 0.02734 & $-$0.02868 & $-$0.01032 & 0.03723 & 0.01969 & 0.05310 & 0.03597 & $-$0.04759 & $-$0.00076 & $-$0.00937 \\
			GDP & 0.00505 & 0.00393 & 0.00506 & 0.00809 & 0.00914 & 0.01189 & 0.00987 & 0.01121 & 0.00985 & 0.00859 & 0.00699 & $-$0.00426 & $-$0.00555 & $-$0.00529 \\
			\hline
		\end{tabular}
	\end{adjustbox}
\end{table}

For France, the large values are confined to the first equation, which indicates a great volatility of the autoregressive and discretionary parts of public defence expenditure. In general, changes in tax and economic conditions had a strong impact on the French economy, but a small one on the fiscal policy regime, as can be seen from the estimates for the second equation.

\begin{table}[htbp]
	\centering
	\caption{Estimated coefficients $\hat{\Psi}$ on the policy-economic regime intervention variables, France. Rows correspond to the equations of the model, columns to the consecutive regimes $d_{1},\dots,d_{12}$.}
	\label{tab:psi_france}
	\begin{adjustbox}{max width=\textwidth}
		\small
		\begin{tabular}{lrrrrrrrrrrrr}
			\hline
			Equation & $d_{1}$ & $d_{2}$ & $d_{3}$ & $d_{4}$ & $d_{5}$ & $d_{6}$ & $d_{7}$ & $d_{8}$ & $d_{9}$ & $d_{10}$ & $d_{11}$ & $d_{12}$ \\
			\hline
			Defence expenditure & 0.33926 & 0.34866 & 0.20783 & 0.32258 & 0.34242 & 0.34884 & 0.30816 & 0.37064 & $-$0.07291 & 0.34345  & 0.40661 & 0.41563 \\
			Non-defence expenditure & $-$0.01479 & 0.00344 & 0.01298 & 0.02274 & 0.01476 & 0.01164 & 0.00115 & $-$0.00977 & 0.00678& $-$0.01608  & 0.00404 & $-$0.00553 \\
			Private consumption & 0.05826 & 0.05591 & 0.05880 & 0.05875 & 0.05932 & 0.05509 & 0.06371 & 0.05902 & 0.00144 & 0.05763  & 0.05827 & 0.05916 \\
			Gross capital formation & 0.04809 & 0.04010 & 0.06190 & 0.08264 & 0.07695 & 0.07067 & 0.06783 & 0.08775 & $-$0.01798 & 0.05174  & 0.04359 & 0.04231 \\
			Government revenue & 0.07028 & 0.02722 & $-$0.05475 & $-$0.00879 & 0.03119 & $-$0.02733 & 0.03371 & 0.02277 & 0.08807 & $-$0.01683  & 0.03239 & 0.03230 \\
			GDP & 0.08122 & 0.07810 & 0.08099 & 0.08160 & 0.08339 & 0.07993 & 0.08516 & 0.08593 & $-$0.00023& 0.07921  & 0.08007 & 0.08048 \\
			\hline
		\end{tabular}
	\end{adjustbox}
\end{table}

As for Italy, the estimates are again large in absolute value for the first two equations and for the fifth, so that the same conclusion applies to the volatility of the autoregressive and discretionary parts of government expenditure and revenue. The coefficients in the defence expenditure and government revenue equations lie between $0.39$ and $0.62$ and between $0.56$ and $0.68$ respectively in every one of the eleven regimes, while those in the private consumption and GDP equations are of the order of four hundredths; the successive policy and economic frameworks therefore shift the level of the fiscal variables by an order of magnitude more than they shift the level of the private variables. Within each equation, on the other hand, the entries vary comparatively little from regime to regime, so that what the intervention variables remove from the residuals is largely a common level effect of the legal and policy framework rather than a sequence of sharply distinct shifts. In general, changes in tax and economic conditions had a strong impact on the Italian economy, and a very large one on public spending.

\begin{table}[htbp]
	\centering
	\caption{Estimated coefficients $\hat{\Psi}$ on the policy-economic regime intervention variables, Italy. Rows correspond to the equations of the model, columns to the consecutive regimes $d_{1},\dots,d_{11}$.}
	\label{tab:psi_italy}
	\begin{adjustbox}{max width=\textwidth}
		\small
		\begin{tabular}{lrrrrrrrrrrr}
			\hline
			Equation & $d_{1}$ & $d_{2}$ & $d_{3}$ & $d_{4}$ & $d_{5}$ & $d_{6}$ & $d_{7}$ & $d_{8}$ & $d_{9}$ & $d_{10}$ & $d_{11}$ \\
			\hline
			Defence expenditure & 0.46396 & 0.50699 & 0.55881 & 0.45210 & 0.49026 & 0.39083 & 0.44740 & 0.49850 & 0.44555 & 0.47579 & 0.61120 \\
			Non-defence expenditure & $-$0.08198 & $-$0.05536 & $-$0.02631 & $-$0.03728 & $-$0.04613 & $-$0.04907 & $-$0.03083 & $-$0.06575 & $-$0.02047 & $-$0.02859 & $-$0.05127 \\
			Private consumption & 0.04101 & 0.04388 & 0.04191 & 0.04029 & 0.04077 & 0.04379 & 0.04197 & 0.04064 & 0.03956 & 0.04247 & 0.03905 \\
			Gross capital formation & $-$0.07544 & $-$0.05955 & $-$0.04521 & $-$0.06247 & $-$0.04897 & $-$0.04182 & $-$0.05077 & $-$0.04769 & $-$0.05505 & $-$0.03025 & $-$0.03951 \\
			Government revenue & 0.67443 & 0.61194 & 0.61106 & 0.62774 & 0.56128 & 0.60769 & 0.59287 & 0.63486 & 0.61051 & 0.62105 & 0.55600 \\
			GDP & 0.03557 & 0.03631 & 0.03421 & 0.03226 & 0.03225 & 0.03636 & 0.03383 & 0.03773 & 0.03828 & 0.03655 & 0.03186 \\
			\hline
		\end{tabular}
	\end{adjustbox}
\end{table}

The estimates for Spain are large in absolute value for the first two equations and for the fifth. This indicates a great volatility of the autoregressive and discretionary parts of government expenditure, and of government revenue. In general, changes in tax and economic conditions had a strong and negative impact on the Spanish economy, and a large to very large one on public spending.
% Do sprawdzenia / to be verified.

\begin{table}[htbp]
	\centering
	\caption{Estimated coefficients $\hat{\Psi}$ on the policy-economic regime intervention variables, Spain. Rows correspond to the equations of the model, columns to the consecutive regimes $d_{1},\dots,d_{8}$.}
	\label{tab:psi_spain}
	\begin{adjustbox}{max width=\textwidth}
		\small
		\begin{tabular}{lrrrrrrrr}
			\hline
			Equation & $d_{1}$ & $d_{2}$ & $d_{3}$ & $d_{4}$ & $d_{5}$ & $d_{6}$ & $d_{7}$ & $d_{8}$ \\
			\hline
			Defence expenditure & $-$1.20534 & $-$1.20605 & $-$1.24336 & $-$1.31785 & $-$1.33570 & $-$1.25144 & $-$1.23244 & $-$1.20480 \\
			Non-defence expenditure & $-$0.32965 & $-$0.32227 & $-$0.30209 & $-$0.30949 & $-$0.39621 & $-$0.35612 & $-$0.30396 & $-$0.33265 \\
			Private consumption & 0.13222 & 0.13229 & 0.13777 & 0.14241 & 0.14356 & 0.13904 & 0.13357 & 0.13523 \\
			Gross capital formation & 0.09741 & 0.10257 & 0.10991 & 0.10800 & 0.10512 & 0.10381 & 0.10402 & 0.10862 \\
			Government revenue & 1.19521 & 1.19932 & 1.22796 & 1.07024 & 1.12600 & 1.18043 & 1.17649 & 1.23967 \\
			GDP & 0.13441 & 0.13337 & 0.13681 & 0.13874 & 0.14246 & 0.14078 & 0.13745 & 0.13671 \\
			\hline
		\end{tabular}
	\end{adjustbox}
\end{table}

For the Netherlands the estimates are small in every equation, and in this respect the Dutch case resembles the German one rather than the Spanish or the Italian. The largest entries lie in the defence expenditure equation, between $-0.36953$ and $0.01119$, and in the government revenue equation, between $-0.12471$ and $-0.01056$, while the entries of the equations for private consumption and for GDP are of the order of a few thousandths. The signs are systematic: government revenue, gross capital formation and GDP are shifted downwards by every one of the fourteen regimes, and defence expenditure by all of them but one, whereas private consumption is shifted upwards by all of them but two. The largest individual shifts fall in the years of the pandemic and of its aftermath, the regime for 2021 carrying $-0.36953$ in the defence equation and the regime for 2020 carrying $-0.12471$ in the revenue equation and $0.09221$ in the non-defence expenditure equation, the largest entry of that row. The successive policy and economic frameworks therefore moved the level of the Dutch fiscal variables only slightly, and the level of the private variables less still, so that little of the observed variation of the Dutch system is regime-driven, and correspondingly little would have been misattributed to the autoregressive and to the discretionary part had the regimes not been controlled for.

\begin{table}[htbp]
	\centering
	\caption{Estimated coefficients $\hat{\Psi}$ on the policy-economic regime intervention variables, Netherlands. Rows correspond to the equations of the model, columns to the consecutive regimes $d_{1},\dots,d_{14}$.}
	\label{tab:psi_netherlands}
	\begin{adjustbox}{max width=\textwidth}
		\small
		\begin{tabular}{lrrrrrrrrrrrrrr}
			\hline
			Equation & $d_{1}$ & $d_{2}$ & $d_{3}$ & $d_{4}$ & $d_{5}$ & $d_{6}$ & $d_{7}$ & $d_{8}$ & $d_{9}$ & $d_{10}$ & $d_{11}$ & $d_{12}$ & $d_{13}$ & $d_{14}$ \\
			\hline
			Defence expenditure & $-$0.15286 & $-$0.11922 & $-$0.02508 & $-$0.00715 & $-$0.11460 & $-$0.14111 & $-$0.04587 & $-$0.03124 & $-$0.04993 & 0.01119 & $-$0.03018 & $-$0.36953 & $-$0.17633 & $-$0.09172 \\
			Non-defence expenditure & 0.00188 & 0.04419 & 0.02996 & 0.00372 & 0.02024 & $-$0.02657 & $-$0.00316 & 0.00400 & $-$0.00039 & 0.04006 & 0.09221 & $-$0.00637 & 0.01534 & 0.04023 \\
			Private consumption & 0.00285 & 0.00498 & 0.00680 & 0.00546 & 0.00776 & 0.00556 & 0.00727 & 0.00363 & 0.00081 & 0.00327 & 0.00139 & $-$0.00003 & 0.00139 & $-$0.00181 \\
			Gross capital formation & $-$0.00430 & $-$0.02024 & $-$0.02458 & $-$0.03368 & $-$0.04271 & $-$0.00013 & $-$0.00768 & $-$0.02956 & $-$0.03243 & $-$0.02722 & $-$0.04082 & $-$0.04083 & $-$0.06268 & $-$0.05668 \\
			Government revenue & $-$0.06290 & $-$0.06610 & $-$0.01056 & $-$0.03425 & $-$0.07197 & $-$0.07573 & $-$0.06701 & $-$0.04485 & $-$0.03695 & $-$0.04884 & $-$0.12471 & $-$0.07233 & $-$0.01959 & $-$0.06307 \\
			GDP & $-$0.00502 & $-$0.00222 & $-$0.00571 & $-$0.00636 & $-$0.00559 & $-$0.00114 & $-$0.00172 & $-$0.00406 & $-$0.00363 & $-$0.00779 & $-$0.00998 & $-$0.00528 & $-$0.00976 & $-$0.01193 \\
			\hline
		\end{tabular}
	\end{adjustbox}
\end{table}

For Poland, the estimates are large for the first two equations and for the fifth; the entries in the equation for gross capital formation are large as well. In general, changes in tax and economic conditions had a strong impact on the Polish economy, and a large to very large one on public spending.

\begin{table}[htbp]
	\centering
	\caption{Estimated coefficients $\hat{\Psi}$ on the policy-economic regime intervention variables, Poland. Rows correspond to the equations of the model, columns to the consecutive regimes $d_{1},\dots,d_{7}$.}
	\label{tab:psi_poland}
	\begin{adjustbox}{max width=\textwidth}
		\small
		\begin{tabular}{lrrrrrrr}
			\hline
			Equation & $d_{1}$ & $d_{2}$ & $d_{3}$ & $d_{4}$ & $d_{5}$ & $d_{6}$ & $d_{7}$ \\
			\hline
			Defence expenditure & 0.21974 & 0.23061 & 0.01325 & 0.06790 & 0.00409 & 0.08759 & 0.54017 \\
			Non-defence expenditure & 0.12359 & 0.11708 & 0.14794 & 0.05462 & 0.10504 & 0.15702 & 0.20763 \\
			Private consumption & 0.02315 & 0.00976 & 0.01832 & 0.01064 & $-$0.00179 & 0.00687 & 0.00282 \\
			Gross capital formation & 0.19339 & 0.15852 & 0.14013 & 0.13846 & 0.10581 & 0.11871 & 0.11974 \\
			Government revenue & 0.10627 & 0.17551 & 0.07635 & 0.14231 & 0.22226 & 0.18041 & 0.18362 \\
			GDP & 0.11538 & 0.09967 & 0.09208 & 0.09451 & 0.07871 & 0.07978 & 0.07899 \\
			\hline
		\end{tabular}
	\end{adjustbox}
\end{table}

\subsection{Estimates of the autoregressive approximation of the dynamics}\label{s4.4}

The estimated constants and autoregressive coefficients are reported in Tables \ref{tab:var_germany}-\ref{tab:var_poland}. In those tables the columns are the equations of the model and the rows are the regressors, so that an equation is read down a column: the entry in the row labelled $c_{t-1}$ and in the column headed $r$ is the coefficient of lagged private consumption in the equation for government revenue. This is the opposite arrangement to that of Tables \ref{tab:psi_spain}-\ref{tab:psi_germany}, in which the equations are the rows and the regimes the columns; the two orientations are dictated by the number of regimes, which in the German case would make a transposed table unreadable. In the discussion below an equation is therefore named by its dependent variable and not by a row index. Orders of magnitude are counted as powers of ten over the forty-two entries of each table, that is, over the six constants and the thirty-six autoregressive coefficients, and the coefficients on the intervention variables, discussed in Section \ref{s4.3}, are not included in the counts. As a summary measure of how strongly an equation depends on the state of the economy in the preceding year, the sum of the absolute values of its six autoregressive coefficients, $\sum_{j}\lvert\hat{\phi}_{ij}\rvert$, is reported alongside.

For Germany the entries are distributed over the orders of magnitude as follows: three of order $10^{1}$, eleven of order $10^{0}$, sixteen of order $10^{-1}$, ten of order $10^{-2}$, one of order $10^{-3}$ and one of order $10^{-4}$. The sums of absolute autoregressive coefficients are $21.20$ in the defence expenditure equation, $12.35$ in the government revenue equation, $7.09$ in the non-defence expenditure equation, $4.30$ in the equation for gross capital formation, $1.18$ in the GDP equation and $0.70$ in the private consumption equation. The three fiscal equations are thus far more strongly tied to the previous year's state of the economy than the three private-sector equations, and the two extremes differ by a factor of thirty. The largest single entry of the system is $-10.15$, the coefficient of lagged private consumption in the revenue equation, followed closely by $-10.14$ on lagged GDP and $5.27$ on lagged consumption in the defence equation and by $5.60$ on lagged GDP in the non-defence equation. The coefficients above unity in absolute value in the two expenditure equations most likely result from the composite effect of regular social transfers, subsidies, compensation of employees and pensions, which respond to consumption and to the government's new debt-taking. Given that many wages and transfers of this kind are not adjusted every period, these are cumulative estimates of what structurally ought to be equations with more than one lag on some of the variables. In the GDP equation the coefficients on lagged consumption, $-0.58$, and on lagged investment, $-0.25$, are both negative, which reflects a two-year cyclical pattern in private demand rather than a detrimental effect of consumption or investment on future output; an autoregression of order one cannot represent a cyclical pattern, and, as explained in Section \ref{s3.2}, a standard VAR does not allow the lag order to be chosen separately for individual variables.

\begin{table}[htbp]
	\centering
	\caption{Estimated constant and autoregressive coefficients of the VARX(1) model, Germany. Columns correspond to the equations of the six endogenous variables.}
	\label{tab:var_germany}
	\begin{adjustbox}{max width=\textwidth}
		\small
		\begin{tabular}{lrrrrrr}
			\hline
			& $g^{d}$ & $g^{n}$ & $c$ & $i$ & $r$ & $y$ \\
			\hline
			Constant & 2.48680 & $-$3.95233 & 1.08077 & 1.55910 & 13.46873 & 1.48464 \\
			$g^{d}_{t-1}$ & 0.32096 & 0.09477 & 0.00065 & 0.00896 & $-$0.31652 & 0.02140 \\
			$g^{n}_{t-1}$ & $-$0.45899 & 0.12979 & 0.02010 & $-$0.03493 & $-$0.04929 & 0.01035 \\
			$c_{t-1}$ & 5.27162 & 0.19675 & $-$0.21785 & 1.30352 & $-$10.14560 & $-$0.57822 \\
			$i_{t-1}$ & 4.32514 & $-$1.04636 & $-$0.16914 & 0.55829 & $-$0.87574 & $-$0.25371 \\
			$r_{t-1}$ & $-$0.68616 & $-$0.02133 & 0.01441 & $-$0.01040 & $-$0.39178 & 0.03337 \\
			$y_{t-1}$ & $-$10.13913 & 5.60225 & 0.27323 & $-$2.37926 & $-$0.56716 & 0.28098 \\
			\hline
		\end{tabular}
	\end{adjustbox}
\end{table}

For France no entry is of order $10^{1}$; there are eight of order $10^{0}$, twenty of order $10^{-1}$, twelve of order $10^{-2}$, one of order $10^{-3}$ and one of order $10^{-4}$. The sums of absolute autoregressive coefficients are $14.64$ in the defence expenditure equation, $9.47$ in the government revenue equation, $2.31$ in the non-defence expenditure equation, $2.26$ in the equation for gross capital formation, $1.10$ in the GDP equation and $0.40$ in the private consumption equation. The ordering is the same as in Germany, but the spread between the fiscal and the private-sector equations is narrower, and no coefficient of the system reaches ten in absolute value. The largest entries are $-8.55$ on lagged consumption and $2.66$ on lagged GDP in the defence equation, and $-4.03$ on lagged consumption, $2.81$ on lagged GDP and $1.57$ on lagged non-defence expenditure in the revenue equation. The last of these is notable, since it is the only case in the six countries in which the previous year's growth of non-defence expenditure carries a coefficient above unity in the revenue equation, and it is consistent with a budgetary practice in which an expansion of expenditure is followed by measures on the revenue side. In the GDP equation the coefficient on lagged GDP is $0.54$ and that on lagged consumption $-0.48$, which again reflects the cumulative effect of more than one lag compressed into a model of order one.

\begin{table}[htbp]
	\centering
	\caption{Estimated constant and autoregressive coefficients of the VARX(1) model, France. Columns correspond to the equations of the six endogenous variables.}
	\label{tab:var_france}
	\begin{adjustbox}{max width=\textwidth}
		\small
		\begin{tabular}{lrrrrrr}
			\hline
			& $g^{d}$ & $g^{n}$ & $c$ & $i$ & $r$ & $y$ \\
			\hline
			Constant & 3.75408 & 0.02458 & 0.64393 & 0.67356 & 0.14216 & 0.89609 \\
			$g^{d}_{t-1}$ & $-$0.22915 & 0.10671 & $-$0.01071 & $-$0.02063 & $-$0.28953 & $-$0.02327 \\
			$g^{n}_{t-1}$ & $-$0.07728 & 0.20549 & $-$0.02699 & 0.14302 & 1.57199 & 0.02389 \\
			$c_{t-1}$ & $-$8.54824 & 0.80912 & 0.18183 & $-$0.53550 & $-$4.02990 & $-$0.47754 \\
			$i_{t-1}$ & 2.50776 & $-$0.65789 & 0.05080 & $-$0.43940 & 0.76337 & $-$0.03369 \\
			$r_{t-1}$ & 0.62002 & 0.01341 & $-$0.01111 & 0.05322 & 0.00066 & $-$0.00154 \\
			$y_{t-1}$ & 2.65912 & 0.51317 & 0.12101 & 1.06975 & 2.81033 & 0.54343 \\
			\hline
		\end{tabular}
	\end{adjustbox}
\end{table}

For Italy there is one entry of order $10^{1}$, nine of order $10^{0}$, twenty of order $10^{-1}$, seven of order $10^{-2}$ and five of order $10^{-3}$.

The feature that dominates the Italian estimates is the size of the autoregressive part of the three fiscal equations relative to that of the three private-sector equations. Summing the absolute values of the six autoregressive coefficients within each equation gives $5.52$ for defence expenditure, $7.27$ for non-defence expenditure and $17.13$ for government revenue, against $1.06$ for private consumption, $2.74$ for gross capital formation and $1.15$ for GDP. The systematic, non-discretionary part of Italian fiscal policy is thus an order of magnitude more responsive to the state of the economy in the preceding year than the private-sector equations are, and this remains true after the successive policy and economic frameworks have been controlled for by the intervention variables, so that it cannot be attributed to the regime shifts themselves. The individual coefficients show the same thing: the equation for government revenue carries $-11.53$ on lagged GDP, $2.60$ on lagged consumption and $2.14$ on lagged investment, and the equations for the two categories of expenditure carry $-2.17$ and $-1.87$ and $-2.51$ and $3.997$ on the same private-sector variables. The regular, rule-like parts of Italian government expenditure and revenue are therefore extremely volatile. Two readings of this are possible and the model does not discriminate between them: either economic policy in Italy is itself highly unstable, so that what is treated here as a systematic rule is in fact a rule that was repeatedly re-set within the regimes identified in Section \ref{s3.3}, or the Italian economy is volatile enough that a rule of moderate elasticity generates fiscal series of this amplitude. In either case the implication for the multiplier is the same, since a large autoregressive part means that a correspondingly small share of the observed movement of the fiscal variables is available to be identified as a discretionary shock.

As for the interpretation of the individual entries, the coefficients larger than one in absolute value on public spending and on government revenue most likely result from the composite effect of regular social transfers, subsidies, paid compensation of employees and pensions on consumption and on new debt-taking of the government. Given that many wages and such transfers are not adjusted every period, this estimate is a cumulative estimate of what structurally ought to be an equation with more than one lag on some of the variables. The negative coefficient of the gross growth rate of government revenue in the equation for the gross growth rate of GDP is indicative of the negative impact of tax collections and debt-taking. Again, the large coefficient on the gross growth of investment shows the cumulative effects of investment in more than one previous period; selecting a different lag order for individual variables is not possible in a standard VAR, hence this form and these cumulative estimates. The negative impact of past consumption on present GDP can be indicative of spending on imports and of too little saving by households.

\begin{table}[htbp]
	\centering
	\caption{Estimated constant and autoregressive coefficients of the VARX(1) model, Italy. Columns correspond to the equations of the six endogenous variables.}
	\label{tab:var_italy}
	\begin{adjustbox}{max width=\textwidth}
		\small
		\begin{tabular}{lrrrrrr}
			\hline
			& $g^{d}$ & $g^{n}$ & $c$ & $i$ & $r$ & $y$ \\
			\hline
			Constant & 5.34138 & $-$0.49304 & 0.45535 & $-$0.55673 & 6.70942 & 0.38520 \\
			$g^{d}_{t-1}$ & $-$0.31001 & 0.08777 & 0.00165 & $-$0.00719 & $-$0.15376 & $-$0.00684 \\
			$g^{n}_{t-1}$ & 0.39249 & 0.32743 & 0.00825 & 0.13047 & 0.62619 & $-$0.01384 \\
			$c_{t-1}$ & $-$2.17210 & $-$0.17681 & 0.74748 & 1.24881 & 2.60432 & 0.79400 \\
			$i_{t-1}$ & $-$1.86505 & $-$2.51029 & 0.02493 & $-$0.53941 & 2.13766 & 0.05734 \\
			$r_{t-1}$ & $-$0.21354 & $-$0.17300 & 0.00117 & $-$0.02125 & 0.07419 & 0.01408 \\
			$y_{t-1}$ & $-$0.56415 & 3.99663 & $-$0.27495 & 0.79347 & $-$11.52994 & $-$0.26110 \\
			\hline
		\end{tabular}
	\end{adjustbox}
\end{table}

For Spain there are two entries of order $10^{1}$, twelve of order $10^{0}$, fifteen of order $10^{-1}$, seven of order $10^{-2}$, four of order $10^{-3}$ and two of order $10^{-4}$. The sums of absolute autoregressive coefficients are $20.68$ in the government revenue equation, $16.62$ in the defence expenditure equation, $11.12$ in the non-defence expenditure equation, $3.87$ in the equation for gross capital formation, $1.05$ in the private consumption equation and $0.22$ in the GDP equation. The contrast between the fiscal and the private-sector equations is the sharpest of the six countries: the revenue equation is almost a hundred times more strongly tied to the previous year's state of the economy than the GDP equation is. The largest entries are $-14.42$ on lagged GDP and $4.90$ on lagged consumption in the revenue equation, $11.29$ on lagged GDP in the defence equation and $7.19$ on lagged GDP in the non-defence equation. The coefficients above unity in absolute value in the expenditure and revenue equations most likely result, as in Germany, from the composite effect of regular social transfers, subsidies, compensation of employees and pensions on consumption and on new debt-taking, and are cumulative estimates of what structurally ought to be equations with more than one lag. The GDP equation, by contrast, is nearly free of autoregressive content once the regimes are controlled for: its six coefficients sum in absolute value to $0.22$, the largest of them being $-0.17$ on lagged GDP itself. Spanish output, on this estimate, is almost entirely explained by the policy and economic framework and by the contemporaneous shock, and hardly at all by the previous year's growth rates.

\begin{table}[htbp]
	\centering
	\caption{Estimated constant and autoregressive coefficients of the VARX(1) model, Spain. Columns correspond to the equations of the six endogenous variables.}
	\label{tab:var_spain}
	\begin{adjustbox}{max width=\textwidth}
		\small
		\begin{tabular}{lrrrrrr}
			\hline
			& $g^{d}$ & $g^{n}$ & $c$ & $i$ & $r$ & $y$ \\
			\hline
			Constant & $-$9.99698 & $-$2.65243 & 1.09610 & 0.83945 & 9.41533 & 1.10073 \\
			$g^{d}_{t-1}$ & $-$0.20082 & 0.00700 & $-$0.00485 & 0.01640 & $-$0.01560 & $-$0.00159 \\
			$g^{n}_{t-1}$ & 0.77337 & 0.13968 & $-$0.01400 & $-$0.00095 & $-$0.22948 & $-$0.00081 \\
			$c_{t-1}$ & $-$2.09340 & $-$2.64704 & $-$0.41187 & 1.46002 & 4.89584 & $-$0.01945 \\
			$i_{t-1}$ & 2.01146 & $-$0.93666 & $-$0.18857 & 0.49838 & 0.76683 & $-$0.02953 \\
			$r_{t-1}$ & 0.24676 & 0.19259 & $-$0.01056 & $-$0.02324 & $-$0.35406 & 0.00326 \\
			$y_{t-1}$ & 11.29242 & 7.19341 & 0.41741 & $-$1.87365 & $-$14.42145 & $-$0.16877 \\
			\hline
		\end{tabular}
	\end{adjustbox}
\end{table}

For the Netherlands there are four entries of order $10^{1}$, eight of order $10^{0}$, fourteen of order $10^{-1}$, nine of order $10^{-2}$, six of order $10^{-3}$ and one of order $10^{-4}$. The sums of absolute autoregressive coefficients are $38.21$ in the defence expenditure equation, $11.62$ in the government revenue equation, $10.96$ in the non-defence expenditure equation, $3.65$ in the equation for gross capital formation, $1.25$ in the GDP equation and $0.38$ in the private consumption equation. The defence equation is the most volatile of any in the six countries, carrying $-19.49$ on lagged GDP, $-14.81$ on lagged consumption and $-1.82$ on lagged investment, and it is followed here, unlike in the other five countries, by the equation for non-defence expenditure, which carries $-6.35$ on lagged GDP and $4.37$ on lagged consumption, and only then by the revenue equation. This is most likely caused by the last observations of the sample, which show abrupt growth, and it suggests a structural change that remains present in the autoregressive part despite the inclusion of the policy-economic regime variables. The smallness of the Dutch entries of $\hat{\Psi}$ reported in Section \ref{s4.3} makes this reading the more compelling, since the regimes themselves absorb very little of the movement of the system, and what is not absorbed by them has to appear either in the autoregressive part or in the residuals. A shift in the composition of general government expenditure of this kind, and the rapid growth of spending on a sector that had been very small in the Netherlands, entails more than proportional increases in regular expenses in the initial stage of the expansion, as well as the restructuring and enlargement of the armed forces; the Netherlands announced in early 2025 that it plans to double the size of its army. The private consumption equation, at the other extreme, is almost inert, its six coefficients summing in absolute value to $0.38$, the largest of them being $0.20$ on lagged consumption itself.

\begin{table}[htbp]
	\centering
	\caption{Estimated constant and autoregressive coefficients of the VARX(1) model, Netherlands. Columns correspond to the equations of the six endogenous variables.}
	\label{tab:var_netherlands}
	\begin{adjustbox}{max width=\textwidth}
		\small
		\begin{tabular}{lrrrrrr}
			\hline
			& $g^{d}$ & $g^{n}$ & $c$ & $i$ & $r$ & $y$ \\
			\hline
			Constant & 37.47259 & 2.84134 & 0.97709 & $-$0.18509 & 11.85097 & 0.54589 \\
			$g^{d}_{t-1}$ & $-$0.78912 & $-$0.02243 & $-$0.00087 & $-$0.05460 & 0.01514 & $-$0.00602 \\
			$g^{n}_{t-1}$ & 1.20833 & 0.13667 & $-$0.00218 & $-$0.00885 & 0.03957 & $-$0.01427 \\
			$c_{t-1}$ & $-$14.81274 & 4.37074 & 0.19979 & $-$0.47261 & $-$4.80786 & $-$0.26765 \\
			$i_{t-1}$ & $-$1.82037 & 0.07500 & 0.00705 & $-$0.68468 & 0.46507 & $-$0.10217 \\
			$r_{t-1}$ & $-$0.08964 & $-$0.00575 & $-$0.00565 & 0.09293 & $-$0.22037 & 0.02284 \\
			$y_{t-1}$ & $-$19.49467 & $-$6.35127 & $-$0.16404 & 2.33481 & $-$6.07697 & 0.83491 \\
			\hline
		\end{tabular}
	\end{adjustbox}
\end{table}

For Poland there is one entry of order $10^{1}$, ten of order $10^{0}$, fourteen of order $10^{-1}$, fourteen of order $10^{-2}$, two of order $10^{-3}$ and one of order $10^{-4}$. The sums of absolute autoregressive coefficients are $22.17$ in the defence expenditure equation, $11.20$ in the government revenue equation, $5.82$ in the equation for gross capital formation, $2.79$ in the private consumption equation, $1.96$ in the non-defence expenditure equation and $0.67$ in the GDP equation. The defence equation carries $10.95$ on lagged GDP and $-7.40$ on lagged consumption, and, as in the Dutch case, the magnitude is most likely caused by the last observations of the sample. The structural change is of the same kind but of a different origin: the expansion of Polish defence spending during the Russian aggression against Ukraine entails large increases of the public defence sector, which are transferred into the coefficient through the regular expenses of that sector in the initial stage of its enlargement. The part of Polish defence expenditure that goes to domestic producers is large, which is why the effect passes into the rest of the system rather than into imports alone. The GDP equation, with a sum of $0.67$ and a largest entry of $0.46$ on lagged GDP itself, is by contrast the most stable of the six equations.

\begin{table}[htbp]
	\centering
	\caption{Estimated constant and autoregressive coefficients of the VARX(1) model, Poland. Columns correspond to the equations of the six endogenous variables.}
	\label{tab:var_poland}
	\begin{adjustbox}{max width=\textwidth}
		\small
		\begin{tabular}{lrrrrrr}
			\hline
			& $g^{d}$ & $g^{n}$ & $c$ & $i$ & $r$ & $y$ \\
			\hline
			Constant & 1.16334 & 0.91293 & 0.06977 & 0.97476 & 1.08672 & 0.63911 \\
			$g^{d}_{t-1}$ & $-$0.49114 & $-$0.05094 & $-$0.01468 & 0.05201 & $-$0.00036 & $-$0.00879 \\
			$g^{n}_{t-1}$ & 0.02774 & $-$0.40760 & 0.03582 & $-$0.07739 & 0.34578 & 0.01787 \\
			$c_{t-1}$ & $-$7.40205 & 0.54253 & $-$0.63376 & 2.22311 & 4.05529 & $-$0.09025 \\
			$i_{t-1}$ & $-$3.27673 & $-$0.12974 & $-$0.28404 & 0.58852 & 1.10422 & $-$0.08185 \\
			$r_{t-1}$ & 0.02586 & 0.45823 & 0.00489 & $-$0.04896 & $-$0.08815 & $-$0.01086 \\
			$y_{t-1}$ & 10.95014 & $-$0.36709 & 1.81483 & $-$2.83155 & $-$5.60133 & 0.45966 \\
			\hline
		\end{tabular}
	\end{adjustbox}
\end{table}

Three features are common to all six countries and are more informative than any individual entry. The first is that the largest coefficient of the system always lies in the equation for defence expenditure or in the equation for government revenue: $-10.15$ in the German revenue equation, $-8.55$ in the French defence equation, $-11.53$ in the Italian revenue equation, $-14.42$ in the Spanish revenue equation, $-19.49$ in the Dutch defence equation and $10.95$ in the Polish defence equation. The second is that the sums of absolute autoregressive coefficients are, in every country without exception, largest in these same two equations and smallest in the equations for private consumption and for GDP. The systematic part of fiscal policy is therefore an order of magnitude more responsive to the state of the economy in the preceding year than the private sector is, and this holds after the successive legal and policy frameworks have been removed by the intervention variables, so that it cannot be attributed to the regime shifts themselves.

The third feature concerns the direction of that responsiveness. Adding the three coefficients of the revenue equation that relate to private demand, $\hat{\phi}_{r,c}+\hat{\phi}_{r,i}+\hat{\phi}_{r,y}$, gives the response of the growth of government revenue to a common increase of one percentage point in the growth rates of private consumption, of gross capital formation and of GDP together. It is $-11.59$ for Germany, $-10.42$ for the Netherlands, $-8.76$ for Spain, $-6.79$ for Italy, $-0.46$ for France and $-0.44$ for Poland. The sign is negative in all six. Since the variables are gross growth rates, these numbers are to be read as the fall, in percentage points, of the growth of revenue: even the smallest of them, the French and the Polish, amount to about half a percentage point of revenue growth for every percentage point of private demand growth, which is not negligible against the observed dispersion of series whose gross growth rates lie for the most part between $0.8$ and $1.2$. The six countries therefore differ in the strength of the reaction by more than an order of magnitude, but not in its presence or in its direction.

A year of weak private demand is thus followed by a year in which the growth of government revenue is high, and a year of strong private demand by one in which it is low. Two things must be kept apart here. The co-movement of revenue with its own tax base within a single year, which is the automatic stabiliser in the narrow sense, is contemporaneous, and in a model of the form \ref{eq:2} it is carried by $\widehat{\Sigma}_{u}$ and by the recursive identification of Section \ref{s3.4}, not by $\Phi_{1}$. What $\Phi_{1}$ measures is the relation between this year's revenue and last year's private demand, once that contemporaneous co-movement has been taken out, and the estimates say that this lagged relation runs in the opposite direction. After a weak year the revenue side is tightened in the following budget in order to protect the deficit; after a strong year the pressure to do so is absent. That the effect survives the inclusion of the intervention variables is the point, since the level shifts produced by changes of the tax framework have already been removed by $\Psi{d}_{t}$. What is measured is the reaction of the government within a given framework to the previous year's outcome, and that reaction is rapid, counter-directional, and in four of the six countries large.

That this is a property of the revenue equation, and not a general tendency of the system to revert towards its mean, can be checked on the same estimates. The corresponding sum in the GDP equation, $\hat{\phi}_{y,c}+\hat{\phi}_{y,i}+\hat{\phi}_{y,y}$, is $-0.55$ for Germany, $0.03$ for France, $0.59$ for Italy, $-0.22$ for Spain, $0.47$ for the Netherlands and $0.29$ for Poland: small, and of mixed sign. A generic reversion of all the gross growth rates towards their means would show itself in both equations. It shows itself in one.

The size of these coefficients once again raises the question of how they are to be read. Because the variables of the model are gross growth rates, a coefficient is at first sight an elasticity: the number of percentage points by which the growth rate of the dependent variable changes when the growth rate of a regressor changes by one percentage point. It is, however, more than that, because the variables it links are macroeconomic aggregates of very different size. Suppose that in levels a change in $X_{j}$ is accompanied by a change $\Delta X_{i}=k\,\Delta X_{j}$ in $X_{i}$. Written in gross growth rates the same relation becomes
\begin{equation}\label{eq:21}
	x_{i,t}-1\;=\;k\,\frac{X_{j,t-1}}{X_{i,t-1}}\,\bigl(x_{j,t-1}-1\bigr),
\end{equation}
so that the coefficient estimated by the model is not $k$ but $k$ multiplied by the ratio of the levels of the two aggregates. A coefficient of this model therefore carries two things at once: the strength of the relation between the two variables, and their relative size. Expenditure on defence is of the order of one to two per cent of GDP over most of the sample, while private consumption is of the order of one half of it, so that the ratio in \ref{eq:21} is of the order of twenty-five to fifty, and a modest structural relation with $k$ well below unity appears in the model as a coefficient of order ten. The same mechanism, less dramatically, operates between government revenue and GDP.

It follows that requiring the coefficients of such a model to lie below unity would be an error. It would amount to requiring the aggregates that the model relates to be of the same size, which they are not, and it would force the estimator to misattribute to other regressors, or to the residuals, a relation that is genuinely present in the data. The relevant question is not whether a coefficient exceeds one but whether the implied relation between the levels is plausible. For example, by \ref{eq:21}, a coefficient of ten in the defence equation corresponds to an implied $k$ of about $0.2$ to $0.4$, which is not extreme at all.

One qualification must nevertheless be entered about the individual entries, as distinct from the sums just discussed. It is sometimes assumed that private consumption, gross capital formation and GDP must be collinear because they are related by an accounting identity. That identity holds in levels; it does not carry over to gross growth rates, in which the corresponding relation involves expenditure shares that change from year to year, so there is no theoretical reason why the three regressors should be collinear here. Whether they are close to collinear in a particular sample is an empirical question, and one that the estimator of Section \ref{s3.3} answers directly, since the singular values of ${X}$ are returned along with the coefficients, and since the bootstrap of Section \ref{s3.7} yields the sampling distribution of the individual entries of $\Phi_{1}$ alongside that of their sums. %The pairs of large offsetting entries visible in several of the revenue equations, $4.90$ on lagged consumption against $-14.42$ on lagged GDP in the Spanish case, are consistent with the split of a joint effect being less well determined than the joint effect itself, but they do not establish it, and the check should be made rather than assumed. Independently of how it comes out, 
The sums are what answers the question of interest here, namely the response of the fiscal variables to a common movement of private demand; and it is for the same reason that the analysis of Sections \ref{s4.1} and \ref{s4.2} rests on impulse responses, which are functions of the whole of $\Phi_{1}$ and of $\widehat{\Sigma}_{u}$, rather than on individual coefficients, and that the bootstrap intervals of Section \ref{s3.7} are computed for the responses rather than for the entries of $\Phi_{1}$.

\subsection{Reporting the distribution of probability mass around the multiplier}\label{s_mass_reporting}

A confidence interval for a fiscal multiplier is conventionally used only to answer a binary question: does the interval contain 1.0, or does it not? This is an impoverished way to communicate the estimate to a government that must decide, in practice, how large a spending programme to run. What matters for policy is not merely whether the multiplier is statistically distinguishable from unity, but which side of the deterministic path carries more probability mass -- that is, whether outcomes above the point estimate are more or less likely than outcomes below it.
\begin{figure}[htbp]
	\centering
	\includegraphics[width=\textwidth]{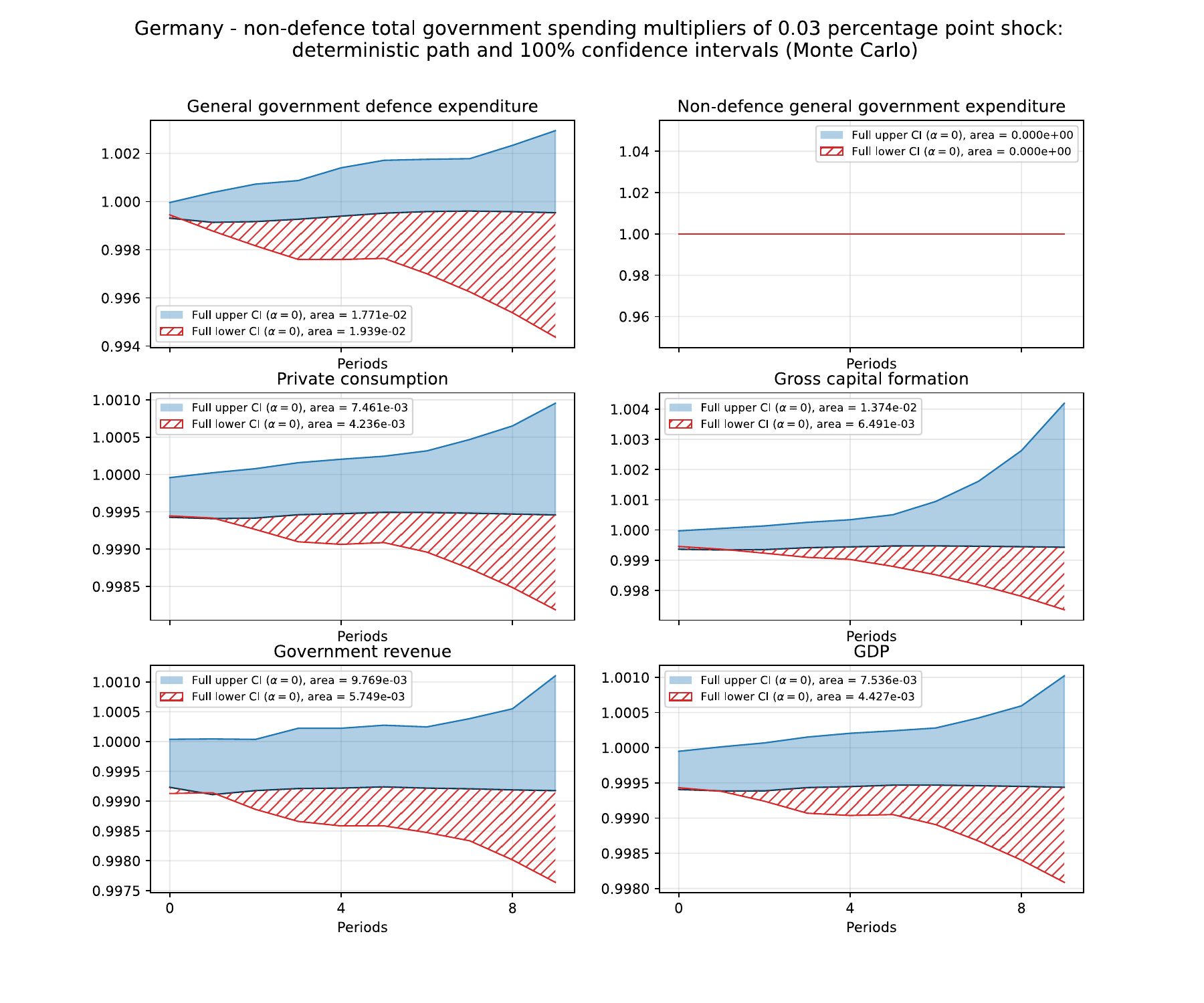}
	\caption{Germany -- non-defence total government spending multipliers of a 0.03 percentage point shock: deterministic path and 100\% confidence intervals (Monte Carlo).}
	\label{fig:mnznk_nondef}
\end{figure} A symmetric-looking interval around the deterministic path can conceal a highly asymmetric distribution of mass, and this asymmetry, not the mere width of the band, is what should inform a decision-maker weighing the upside risk of an expenditure programme against its downside risk. We therefore report, for every variable and horizon, the area enclosed between the deterministic path and the upper bound of the confidence interval separately from the area enclosed between the lower bound and the deterministic path, both computed by trapezoidal rule.

The deterministic path itself remains essential and should not be discarded once a confidence interval has been constructed around it, because it is this path, not the interval taken alone, that displays the actual dynamics generated by the estimated system: its curvature, its turning points, and the timing of its response to a shock. An interval by itself is agnostic about dynamics -- it only bounds a range of admissible outcomes at each horizon -- whereas the deterministic path shows the specific trajectory implied by the point estimates of the autoregressive coefficients and the covariance matrix, which is the object of direct interest to a government trying to anticipate the time profile of a policy's effects, not only its eventual range.

Table~\ref{tab:areas_h5} and Table~\ref{tab:areas_h20} report the two areas -- above the deterministic path (detrmn$\leftrightarrow$upper CI) and below it (lower CI$\leftrightarrow$detrmn) -- computed at horizons $H=5$ and $H=20$ respectively, for a shock of 0.03 percentage points to non-defence government expenditure. Comparing the two tables and Figure~\ref{fig:mnznk_nondef} shows that the relative size of the two masses is not stable over the horizon: at $H=5$ the area above the deterministic path exceeds the area below it for every variable with a nonzero response, whereas at $H=20$ this ordering is reversed for gross capital formation, where the upper mass (0.399) is roughly seven times the lower mass (0.060), yet has grown far more asymmetrically than at the short horizon (1.165 vs.\ 0.359), and government revenue nearly equalises the two masses (0.0393 vs.\ 0.0392) despite an initially asymmetric split. 

This reordering demonstrates that the balance of probability mass around a multiplier in multiplicative form, appropriate for gross growth rates data, is itself a function of the horizon. Consequently, a government evaluating a spending programme over several years needs the horizon-by-horizon decomposition of mass, not only the interval at the horizon conventionally reported in the literature.

\begin{table}[htbp]
	\centering
	\caption{Areas between the deterministic multiplier path and the 100\% confidence interval, horizon $H=5$ (0.03 pp shock to non-defence government expenditure)}
	\label{tab:areas_h5}
	\begin{tabular}{lcc}
		\hline
		Variable & Area of the upper CI & Area of the lower CI  \\
		\hline
		General government defence expenditure     & $2.371\times 10^{-3}$ & $1.788\times 10^{-3}$\\
		Non-defence general government expenditure & $0.000\times 10^{0}$  & $0.000\times 10^{0}$   \\
		Private consumption                        & $9.461\times 10^{-4}$ & $3.533\times 10^{-4}$  \\
		Gross capital formation                    & $1.165\times 10^{-3}$ & $3.587\times 10^{-4}$  \\
		Government revenue                         & $1.305\times 10^{-3}$ & $5.856\times 10^{-4}$  \\
		GDP                                         & $9.727\times 10^{-4}$ & $3.578\times 10^{-4}$ \\
		\hline
	\end{tabular}
\end{table}

\begin{table}[htbp]
	\centering
	\caption{Areas between the deterministic multiplier path and the 100\% confidence interval, horizon $H=20$ (0.03 pp shock to non-defence government expenditure)}
	\label{tab:areas_h20}
	\begin{tabular}{lcc}
		\hline
		Variable & Area of the upper CI & Area of the lower CI \\
		\hline
		General government defence expenditure     & $5.915\times 10^{-2}$ & $8.693\times 10^{-2}$ \\
		Non-defence general government expenditure & $0.000\times 10^{0}$  & $0.000\times 10^{0}$  \\
		Private consumption                        & $1.054\times 10^{-1}$ & $2.529\times 10^{-2}$ \\
		Gross capital formation                    & $3.995\times 10^{-1}$ & $5.966\times 10^{-2}$ \\
		Government revenue                         & $3.934\times 10^{-2}$ & $3.924\times 10^{-2}$ \\
		GDP                                         & $1.121\times 10^{-1}$ & $2.754\times 10^{-2}$ \\
		\hline
	\end{tabular}
\end{table}
\begin{figure}[htbp]
	\centering
	\includegraphics[width=\textwidth]{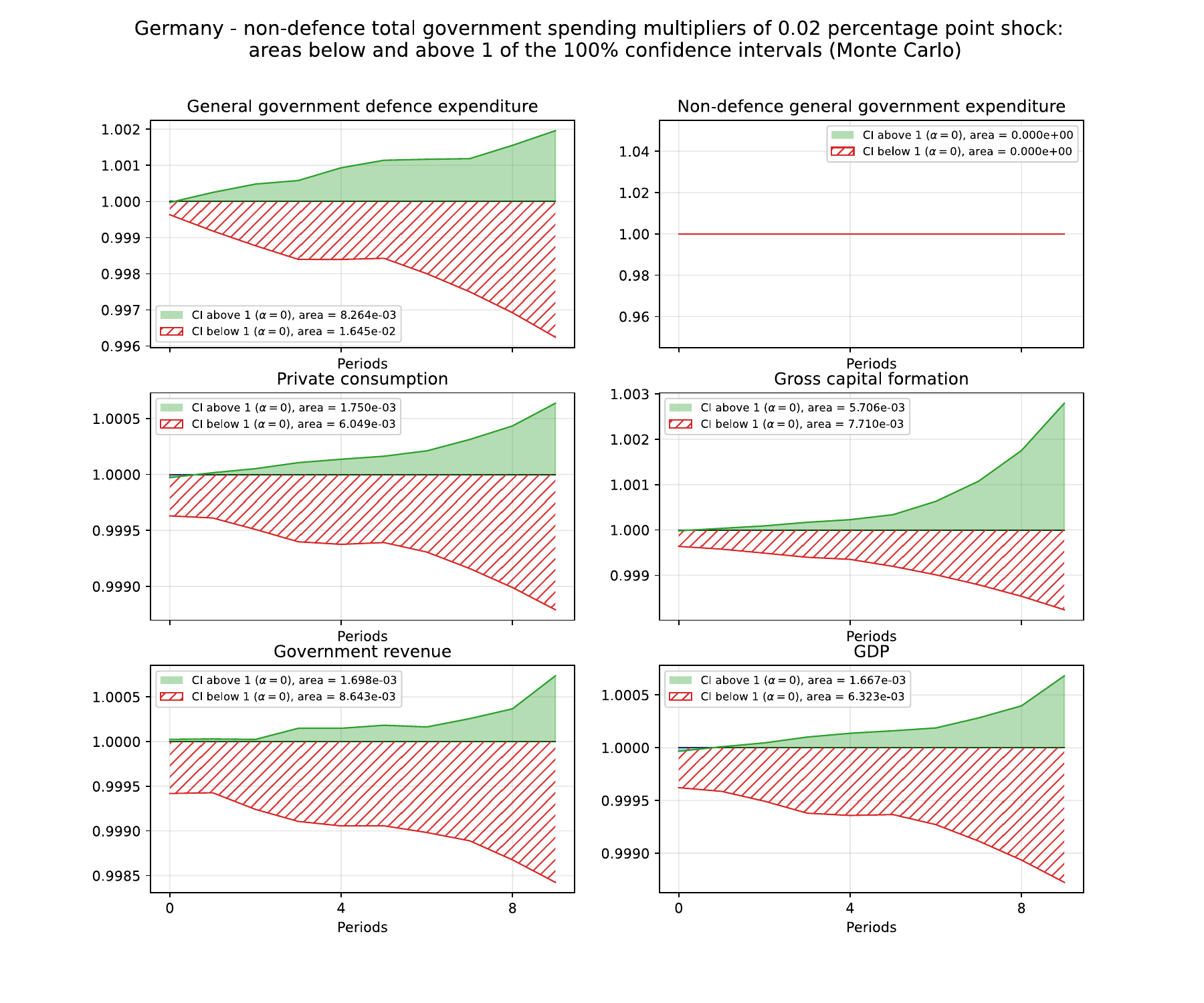}
	\caption{Germany -- non-defence total government spending multipliers of a 0.03 percentage point shock: deterministic path and 100\% confidence intervals (Monte Carlo).}
	\label{fig:mnznk_nondef_PU_PRWD}
\end{figure}
Nonetheless, the area above and below the multiplier's path does not fully convey the probability of the multiplier being positive for the economy, i.e., gross value of the multiplier being larger than 1.0. Thus, there is rationale for computing the mass of the confidence interval below and above the neutral, $1.0$ value of the multiplier, to more formally assess how likely a fiscal shock will increase macroeconomic aggregates, according to the model. Figure \ref{fig:mnznk_nondef_PU_PRWD} shows the results of such an analysis, confirming that focusing on the deterministic path may often be better than just evaluating statistical difference from a given level ($1.0$ in this case). That is, the probability of achieving positive results is lower than the probability of negative outcomes.

\section{Discussion}\label{s5}

The multiplier proposed in this paper is consistent with data in the form of gross, or net, growth rates. It is nonlinear in the size of the shock, which is a new result and one that is relevant for policy: the effect of a given unit of additional expenditure on output depends on how large the expenditure programme is, and this dependence cannot be recovered from the additive multipliers commonly used in the literature.

The multipliers differ across countries, and they differ with the size of the irregular, discretionary shocks to non-defence government expenditure. They are smaller than one, which is a consequence of the narrative approach adopted here: once the policy and economic regimes are controlled for, the identified discretionary shocks are smaller than the ones obtained from SVARs which do not control for regime changes, and the resulting multipliers are correspondingly smaller. This closes the gap between the values obtained from SVAR and from DSGE models.

It is worth setting the multiplier \ref{eq:multiplier} against the objects that the literature computes. The policy-relevant multipliers introduced by \textcite{mountford_what_2009} are the present discounted value of the output response over time divided by the present discounted value of the government spending response over time, and, as \textcite{ramey_macroeconomic_2016} observed, in most applications the choice of the interest rate used to discount, including a zero rate, gives nearly identical values, because the timing of the two responses is very similar. The multiplier proposed here cumulates the same two paths, but the cumulation is a product rather than a sum, because the variables of the model are gross growth rates and their levels are recovered by multiplication. Since no discounting enters, the questions of which interest rate to use, and of whether the deflator should be the GDP deflator or an index tailored to the sectoral composition of government expenditure, do not arise at all. What the resulting quantity measures is nominal, and it is directly comparable with the nominal budget with which a government operates under unknown future inflation. The difference is not merely presentational: as argued in Section \ref{s2}, a discrepancy of even a fraction of a percentage point translates into many millions of a country's currency in budget planning, and the estimates reported in Section \ref{s4} show discrepancies of that order between multipliers computed for impulses of different size.

The same comparison applies to the conversion of elasticities into multipliers. Most empirical work estimates the elasticity of output with respect to government spending or tax revenue from a specification in logarithms, and converts it by post-multiplying by the inverse of an average share of output, as in \textcite{blanchard_empirical_2002} and \textcite{mountford_what_2009}. \textcite{ramey_government_2018} argue that this makes the output multiplier artificially high in recessions, because the government spending share of output is countercyclical, and \textcite{owyang_are_2013} report that $Y/G$ varied from 2 to 24 with a mean of 8 in their historical sample. \textcite{sims_output_2018,sims_state-dependent_2018} show the quantitative relevance of this criticism. The multiplier \ref{eq:multiplier} needs no conversion factor whatsoever. Its numerator and its denominator are cumulated responses of the same kind of quantity, so the ratio is formed directly, no sample average enters it, and no state dependence can be manufactured by the conversion. This matters for the interpretation of the results of Section \ref{s4.1}: the dependence of the multiplier on the size of the impulse found here originates in the algebra of \ref{eq:18}, that is, in the data themselves, and not in a share used to rescale an elasticity. It also avoids the step of replacing the logarithms of government spending, of GDP and of taxes by their ratios to a polynomial trend estimate of potential GDP, which yields variables that are not structural and shocks whose meaning is correspondingly unclear.

The central methodological result is that the approach decomposes the impact of government policy into three parts, which the existing literature does not separate. Recalling \ref{eq:6}, the impact of the regime is measured by $\Psi{d}_{t}$, the impact of the standard policy rule by ${c}+\Phi_{1}{y}_{t-1}$, and the impact of discretionary policy by the structural shocks to the government equations of the SVAR, that is, by the first, the second and the fifth element of ${P}\boldsymbol{\varepsilon}_{t}$. Each of the three has a distinct interpretation. The regime part is nondiscretionary and non-random: it is the shift of the level of the fiscal variables that follows from the tax code and from the wider framework of government policy in force in a given subperiod, and it is dated narratively rather than extracted from the autocorrelation of the series. The policy rule part is the systematic reaction of expenditure and revenue to the state of the economy in the preceding year. The discretionary part is what remains once the first two have been accounted for, and it is this part, and only this part, that enters the calculation of a fiscal multiplier. In an SVAR without intervention variables the first and the third of these components are indistinguishable, and the multiplier is computed from their sum.

This decomposition speaks to several of the identification problems that led the literature away from structural vector autoregressions. The standard SVAR approach treats as exogenous the part of government spending not forecast by lagged values of spending, GDP and taxes, and the identification of exogenous movements in taxes has relied on external estimates of the elasticity of tax revenue to income, as in \textcite{blanchard_empirical_2002}; the estimates were then found to be very sensitive to the value of that external elasticity, as shown by \textcite{mertens_fiscal_2014} and \textcite{caldara_analytics_2017}. The model estimated here requires no external elasticity at all, so that source of sensitivity is absent, and the part of fiscal policy that is neither forecast by the autoregression nor genuinely discretionary is assigned to the regime term rather than being left in the residual. The same holds for fiscal foresight. Where \textcite{house_phased-tax_2006} and \textcite{mertens_empirical_2012} found that phased-in tax cuts depress output during the phase-in period, because firms and consumers delay their activity until rates are lower, the present approach represents a phased-in change directly, as a shift of regime with a known date, so that the delayed activity is part of the regime term and not a shock that agents somehow foresee. \textcite{ramey_ten_2019} claimed that cumulative multipliers do not fully reflect the consequences for the government budget, since a rise in GDP raises tax revenues and the deficit therefore rises by less than the amount spent; here government revenue is one of the endogenous variables and its regime is controlled alongside that of expenditure, so the revenue response is part of the estimated system rather than an omitted offset. The results of Section \ref{s4.3} give this an empirical content: wherever the coefficients on the regime variable in the revenue equation are large, as in Spain, Italy and Poland, a model without the regime-representing intervention variables would have carried that variation into the identified shocks.

The estimates of the coefficients on the intervention variables measure the regimes themselves, and they show how differently the same kind of policy change is transmitted in different economies. In Spain, in Italy and in Poland the large values appear in the equations for both categories of government expenditure and in the equation for government revenue: changes in tax and economic conditions had a strong impact on these economies and a large to very large impact on public spending, and the autoregressive and discretionary parts of expenditure and revenue are correspondingly volatile. In the Polish case the estimates in the equation for gross capital formation are large as well, so that investment participates in the same volatility. In the Netherlands the estimates are small in every equation, the largest of them, in the equations for defence expenditure and for government revenue, remaining below four tenths and below one eighth respectively in absolute value; the Dutch regimes shifted the level of the fiscal variables only slightly and that of the private variables less still, so that the Dutch case belongs with the German one rather than with the Spanish or the Italian, and what distinguishes the Netherlands is not the size of $\hat{\Psi}$ but the size of the autoregressive part discussed in Section \ref{s4.4}. In France the large values are confined to the equation for defence expenditure: the regimes had a strong impact on the French economy, while the fiscal policy regime itself was stable, as the small estimates in the second equation show. In Germany the regimes altered the economy only to a small extent, and the large values in the first equation are concentrated in the last three years of the sample, which identifies a single large regime change at the end of the period against an otherwise quiet background. The matrix $\hat{\Psi}$ is therefore a direct measure of how much of the movement of the fiscal variables is regime-driven, and it separates economies in which public finances are turbulent from those in which they are not.

The estimates of the autoregressive matrices describe the standard policy rule, and their magnitudes measure the volatility of the series that the rule links. The lag order of one is a structural assumption and not a statistical convenience. Budgets are prepared and executed annually, and the systematic component of the fiscal rule relates this year's expenditure and revenue to the outcome of the immediately preceding year; a specification with lags of higher order would relate the budget to years that do not enter the budgetary decision, and would therefore no longer be structural. Read within this structure, the counts of the orders of magnitude reported in Section \ref{s4.4} are informative in their own right. The regular part of defence expenditure is very volatile in Germany, Italy, Spain, the Netherlands and Poland, and volatile in France, whereas the regular part of non-defence expenditure is very volatile in France, volatile in Germany, in Italy and in the Netherlands, and not very volatile in Spain and in Poland. The coefficients larger than one in absolute value on public spending and on government revenue in the equation for the gross growth rate of GDP reflect the composite effect of regular social transfers, subsidies, paid compensation of employees and pensions on consumption and on new debt-taking of the government, and their magnitude measures how strongly these flows move output from one year to the next. The same holds for the coefficients on the gross growth of investment, which show how large the annual swings of capital formation are relative to those of output. In the Italian case the negative coefficient on the gross growth rate of government revenue in the GDP equation is indicative of the negative impact of tax collections and debt-taking, and the negative impact of past consumption on present GDP can be indicative of spending on imports and of too little saving by households.

The estimates of the constants complete the picture of the policy rule. They differ by orders of magnitude between countries and between equations, and where they are of order $10^{0}$ and $10^{1}$ the regular, rule-based component of the corresponding variable is large relative to the regime and to the discretionary components; this is the case, in particular, for defence expenditure. In Italy, Spain and Poland, where the intervention variables form an exhaustive partition of the sample into mutually exclusive periods, and in France, where they do so once the nested regime is set aside, what the model identifies in each subperiod is the sum of the constant and of the corresponding column of the matrix of regime coefficients (plus the column of the nested regime in the year in which it is in force), and it is this sum that measures the level around which the variable fluctuates under a given regime. In Germany and the Netherlands, whose designs have full rank, the constant and the regime coefficients are identified separately, but the level in a given year is likewise the constant plus the coefficients of all the regimes in force in that year.

Under the interpretation of the model, the constant is the approximation of the exponential growth trend underlying an economy, while the regimes are its changes due to changes in economic, policy and legal conditions.

The impulse response functions and the multiplier paths reveal a further property of fiscal policy that is not usually reported. In several of the economies examined, the responses alternate in sign from one horizon to the next, and the multipliers inherit this alternation as a wave-like shape which damps out only slowly. This is a cyclicality of government spending, and it appears alongside the two-annual cyclical pattern of investment which is visible in the large negative coefficients on the gross growth rate of investment in the equations for GDP in Germany, in Spain and in the Netherlands. A single autoregressive equation of order one cannot produce such a pattern from its own dynamics; the cyclicality observed here is generated by the interaction of the equations of the system, which is to say that it is a property of the economy and not of the lag structure imposed on it. The structural AR(1) form is thus sufficient to expose cycles in public spending, and the cumulated, multiplicative form \ref{eq:multiplier} is what makes them legible, since it aggregates the entire path and yields values that are stable across all the horizons considered, whereas a multiplier evaluated at a single horizon may fall at either extreme of the oscillation. That the cyclicality of spending is visible at all is a consequence of the data used: seasonally adjusted and detrended series are spectrally distorted, and the alternation reported here is of exactly the kind that filtering removes.

The estimates for the Netherlands and for Poland show, finally, how much information is carried by movements that a filtered specification would discard. In both countries the coefficient on defence expenditure is large in absolute value, and in both cases this follows from the last few observations, which show abrupt growth; the shift is visible in the regression alongside the policy-economic regime variables. In the Dutch case the shift in the composition of general government expenditure, and the increasing growth of spending on a sector that had previously been very small, entail more than proportional gains from regular expenses in the initial stage of the development of the sector, as well as army restructuring and expansion; since the Dutch coefficients on the intervention variables are small, this movement is carried almost entirely by the autoregressive part and not by the regime term. In the Polish case the increasing growth of spending on the defence sector during the Russian aggression on Ukraine entails large gains of that sector and its expansion, which are transferred into the large coefficient from regular expenses. Not all nonstationary behaviour should therefore be dismissed: here it is the signal, and it measures the speed at which these economies are rearming.

Identification from natural experiments draws its credibility from a particular episode, and its results are correspondingly difficult to generalise to the ordinary conduct of fiscal policy, which is what a finance ministry has to plan. The narrative element used here operates differently: every year of the sample belongs to some regime, the dating comes from the legislative record rather than from an exceptional event, and the identification therefore covers the whole sample rather than a single episode. The estimates obtained are consequently about routine fiscal policy in six economies, and are usable for budget planning in the way that episode-based estimates are not. A second alternative is the one-step instrumental variables estimation of cumulative multipliers by local projections proposed in \textcite{ramey_macroeconomic_2016} and \textcite{ramey_government_2018}, in which cumulative GDP up to a horizon is regressed on cumulative government spending up to that horizon with an SVAR shock or a narrative variable as an instrument. Local projections and structural vector autoregressions theoretically produce the same impulse response functions, but the equivalence is one of population objects and not of estimators, and instrumental variables carry a finite sample bias that is severe in macroeconomic series, which are short. With the effective sample of 28 annual observations used here, that consideration carries considerable weight, and the multipliers are therefore computed using the impulse response functions and the new multiplier formula, appropriate for data in the form of gross growth rates.

Finally, the results bear on the persistent inconsistency between the small government spending multipliers and the very large tax change multipliers reported in the empirical literature. Considering the two separately may make little sense unless public debt is assumed to absorb the entire difference, and the very large output contractions after tax increases shown by impulse response functions imply equally large expansions after tax reductions, since impulse response functions are symmetric by construction. The model estimated here contains government expenditure, in two categories, and government revenue within a single system, and controls the tax regime and the wider policy regime jointly, so that the response of output to expenditure and its response to revenue are estimated under the same identification and remain mutually consistent. 

\section{Conclusions}\label{s6}

This paper has proposed a form of the fiscal multiplier consistent with data expressed as gross growth rates, and an SVAR that controls for policy and economic regimes by means of narratively dated intervention variables. The multiplier is multiplicative rather than additive, since the cumulation of responses of growth rates over a horizon is a product; both its numerator and its denominator are polynomials in the size of the impulse, so that the multiplier depends on how large the expenditure programme is. This dependence is a property of fiscal policy that additive multipliers cannot represent, and it is visible in the estimates: the multipliers obtained for an impulse of one percentage point differ from those obtained for an impulse of three percentage points at every horizon, and they differ nonlinearly. The multiplier requires no discounting, no ad hoc conversion factor and no auxiliary estimate of an average expenditure share of output, and it is computed from nominal, non-adjusted series, which are the quantities in which government budgets are expressed.

The intervention variables decompose the impact of government policy into three parts that the literature does not separate: the regime, measured by the coefficients on the narratively dated indicators; the standard policy rule, represented by the autoregressive equations and the constants, with a lag order of one imposed as the structural form of an annually prepared and executed budget; and discretionary policy, represented by the structural shocks to the government equations. Because the regime component is removed from the residuals, the identified discretionary shocks are smaller than those obtained from SVARs which do not control for regime changes, and the multipliers are correspondingly smaller: in most cases below one, and rarely crossing it, which closes the gap between the values obtained from SVAR and from DSGE models. 

The estimated coefficients on the regime-specific intervention variables and on the autoregressive matrices measure how volatile the fiscal variables and the economies themselves are. Public spending and government revenue are strongly regime-driven in Spain, Italy and Poland, with investment participating in the same volatility in the Polish case; in the Netherlands the regimes shift the levels of all six variables only slightly, while the autoregressive part of the fiscal equations is the largest of the six countries, so that Dutch fiscal policy is driven by the state of the economy in the preceding year rather than by the succession of frameworks; in France the economy responds strongly while the fiscal policy regime is itself stable; and in Germany the regimes have much smaller influence, apart from a single large change concentrated in the last three years of the sample. This result, however, demonstrates the relative stability of German economy in the studied sample.

The results also show that fiscal multipliers differ across the six largest economies of the European Union, that different types of government expenditure have unalike multipliers, and that the impulse responses and multiplier paths of several of these economies display a cyclicality of government spending, generated by the interaction of the equations of the system and legible only once the whole path is cumulated in multiplicative form. Movements that a filtered specification would discard, such as the abrupt growth of defence expenditure at the end of the Dutch and Polish samples, carry information about the speed at which these economies are rearming. Taken together, the approach yields multipliers that are nominal, computed from unadjusted data, sensitive to the size of the shock and purged of the regime component of fiscal policy, and it therefore provides a basis for empirical tools with direct relevance for budget planning and for the evaluation of the effects of fiscal policy.

\printbibliography      

@article{ramey_ten_2019,
	title = {Ten Years After the Financial Crisis: What Have We Learned from the Renaissance in Fiscal Research?},
	volume = {33},
	issn = {0895-3309},
	doi = {10.1257/jep.33.2.89},
	shorttitle = {Ten Years After the Financial Crisis},
	pages = {89--114},
	number = {2},
	journaltitle = {Journal of Economic Perspectives},
	shortjournal = {Journal of Economic Perspectives},
	author = {Ramey, Valerie A.},
	date = {2019},
	langid = {english},
}

@article{owyang_are_2013,
	title = {Are Government Spending Multipliers Greater during Periods of Slack? Evidence from Twentieth-Century Historical Data},
	volume = {103},
	issn = {0002-8282},
	doi = {10.1257/aer.103.3.129},
	shorttitle = {Are Government Spending Multipliers Greater during Periods of Slack?},
	pages = {129--134},
	number = {3},
	journaltitle = {American Economic Review},
	author = {Owyang, Michael T. and Ramey, Valerie A. and Zubairy, Sarah},
	urldate = {2026-06-29},
	date = {2013},
	langid = {english},
}

@article{mountford_what_2009,
	title = {What are the effects of fiscal policy shocks?},
	volume = {24},
	issn = {1099-1255},
	doi = {10.1002/jae.1079},
	pages = {960--992},
	number = {6},
	journaltitle = {Journal of Applied Econometrics},
	publisher = {John Wiley \& Sons, Ltd},
	author = {Mountford, Andrew and Uhlig, Harald},
	date = {2009},
	langid = {english},
}

@article{blanchard_empirical_2002,
	title = {An Empirical Characterization of the Dynamic Effects of Changes in Government Spending and Taxes on Output},
	volume = {117},
	doi = {10.1162/003355302320935043},
	pages = {1329--68},
	number = {4},
	journaltitle = {Quarterly Journal of Economics},
	author = {Blanchard, Olivier and Perotti, Roberto},
	date = {2002},
	langid = {english},
}

@article{house_phased-tax_2006,
	title = {Phased-In Tax Cuts and Economic Activity},
	volume = {96},
	pages = {1835--1849},
	number = {5},
	journaltitle = {American Economic Review},
	author = {House, Christopher L and Shapiro, Matthew D},
	date = {2006},
	langid = {english},
}

@article{mertens_dynamic_2013,
	title = {The Dynamic Effects of Personal and Corporate Income Tax Changes in the United States},
	volume = {103},
	issn = {0002-8282},
	doi = {10.1257/aer.103.4.1212},
	pages = {1212--1247},
	number = {4},
	journaltitle = {American Economic Review},
	author = {Mertens, Karel R. S. M. and Ravn, Morten O.},
	date = {2013},
	langid = {english},
}

@incollection{ramey_macroeconomic_2016,
	title = {Macroeconomic Shocks and Their Propagation},
	volume = {2},
	isbn = {978-0-444-59487-7},
	doi = {10.1016/bs.hesmac.2016.03.003},
	pages = {71--162},
	booktitle = {Handbook of Macroeconomics},
	publisher = {Elsevier},
	author = {Ramey, Valerie A.},
	date = {2016},
	langid = {english},
}

@article{sims_state-dependent_2018,
	title = {The state-dependent effects of tax shocks},
	volume = {107},
	issn = {0014-2921},
	doi = {10.1016/j.euroecorev.2018.05.002},
	pages = {57--85},
	journaltitle = {European Economic Review},
	shortjournal = {European Economic Review},
	author = {Sims, Eric and Wolff, Jonathan},
	date = {2018},
}

@article{sims_output_2018,
	title = {The Output and Welfare Effects of Government Spending Shocks Over the Business Cycle},
	volume = {59},
	issn = {1468-2354},
	doi = {10.1111/iere.12308},
	pages = {1403--1435},
	number = {3},
	journaltitle = {International Economic Review},
	author = {Sims, Eric and Wolff, Jonathan},
	date = {2018},
	langid = {english},
}

@article{ramey_government_2018,
	title = {Government Spending Multipliers in Good Times and in Bad: Evidence from {US} Historical Data},
	volume = {126},
	doi = {10.1086/696277},
	shorttitle = {Government Spending Multipliers in Good Times and in Bad},
	pages = {850--901},
	number = {2},
	journaltitle = {Journal of Political Economy},
	publisher = {The University of Chicago Press},
	author = {Ramey, Valerie A. and Zubairy, Sarah},
	date = {2018},
}

@article{pappa_effects_2009,
	title = {The effects of fiscal shocks on employment and the real wage},
	volume = {50},
	issn = {0020-6598, 1468-2354},
	doi = {10.1111/j.1468-2354.2008.00528.x},
	pages = {217--244},
	number = {1},
	journaltitle = {International Economic Review},
	shortjournal = {Int Economic Review},
	author = {Pappa, Evi},
	date = {2009},
	langid = {english},
}

@article{ramey_identifying_2011,
	title = {Identifying Government Spending Shocks: It's all in the Timing*},
	volume = {126},
	issn = {0033-5533, 1531-4650},
	doi = {10.1093/qje/qjq008},
	shorttitle = {Identifying Government Spending Shocks},
	pages = {1--50},
	number = {1},
	journaltitle = {The Quarterly Journal of Economics},
	author = {Ramey, Valerie A.},
	date = {2011},
	langid = {english},
}

@article{ramey_can_2011,
	title = {Can Government Purchases Stimulate the Economy?},
	volume = {49},
	issn = {0022-0515},
	doi = {10.1257/jel.49.3.673},
	pages = {673--685},
	number = {3},
	journaltitle = {Journal of Economic Literature},
	shortjournal = {Journal of Economic Literature},
	author = {Ramey, Valerie A.},
	date = {2011},
	langid = {english},
}

@article{auerbach_measuring_2012,
	title = {Measuring the Output Responses to Fiscal Policy},
	volume = {4},
	issn = {1945-7731, 1945-774X},
	doi = {10.1257/pol.4.2.1},
	pages = {1--27},
	number = {2},
	journaltitle = {American Economic Journal: Economic Policy},
	shortjournal = {American Economic Journal: Economic Policy},
	author = {Auerbach, Alan J and Gorodnichenko, Yuriy},
	date = {2012},
	langid = {english},
}

@article{romer_new_2004,
	title = {A New Measure of Monetary Shocks: Derivation and Implications},
	volume = {94},
	number = {4},
	journaltitle = {American Economic Review},
	author = {Romer, Christina D and Romer, David H},
	date = {2004},
	langid = {english},
}

@article{romer_macroeconomic_2010,
	title = {The Macroeconomic Effects of Tax Changes: Estimates Based on a New Measure of Fiscal Shocks},
	volume = {100},
	issn = {0002-8282},
	doi = {10.1257/aer.100.3.763},
	shorttitle = {The Macroeconomic Effects of Tax Changes},
	pages = {763--801},
	number = {3},
	journaltitle = {American Economic Review},
	shortjournal = {American Economic Review},
	author = {Romer, Christina D and Romer, David H},
	date = {2010},
	langid = {english},
}

@article{mertens_fiscal_2014,
	title = {Fiscal Policy in an Expectations-Driven Liquidity Trap},
	volume = {81},
	issn = {0034-6527, 1467-937X},
	doi = {10.1093/restud/rdu016},
	pages = {1637--1667},
	number = {4},
	journaltitle = {The Review of Economic Studies},
	shortjournal = {The Review of Economic Studies},
	author = {Mertens, Karel R. S. M. and Ravn, Morten O.},
	date = {2014},
	langid = {english},
}

@article{caldara_analytics_2017,
	title = {The Analytics of {SVARs}: A Unified Framework to Measure Fiscal Multipliers},
	volume = {84},
	issn = {0034-6527},
	doi = {10.1093/restud/rdx030},
	shorttitle = {The Analytics of {SVARs}},
	pages = {1015--1040},
	number = {3},
	journaltitle = {The Review of Economic Studies},
	shortjournal = {Rev Econ Stud},
	author = {Caldara, Dario and Kamps, Christophe},
	date = {2017},
}

@article{baumeister_drawing_2020,
	title = {Drawing conclusions from structural vector autoregressions identified on the basis of sign restrictions},
	volume = {109},
	issn = {02615606},
	doi = {10.1016/j.jimonfin.2020.102250},
	pages = {102250},
	journaltitle = {Journal of International Money and Finance},
	shortjournal = {Journal of International Money and Finance},
	author = {Baumeister, Christiane and Hamilton, James D.},
	date = {2020},
	langid = {english},
}

@article{baumeister_sign_2015,
	title = {Sign Restrictions, Structural Vector Autoregressions, and Useful Prior Information},
	volume = {83},
	issn = {0012-9682},
	doi = {10.3982/ECTA12356},
	pages = {1963--1999},
	number = {5},
	journaltitle = {Econometrica},
	shortjournal = {Econometrica},
	author = {Baumeister, Christiane and Hamilton, James D.},
	date = {2015},
	langid = {english},
}

@article{baumeister_inference_2018,
	title = {Inference in structural vector autoregressions when the identifying assumptions are not fully believed: Re-evaluating the role of monetary policy in economic fluctuations},
	volume = {100},
	issn = {03043932},
	doi = {10.1016/j.jmoneco.2018.06.005},
	shorttitle = {Inference in structural vector autoregressions when the identifying assumptions are not fully believed},
	pages = {48--65},
	journaltitle = {Journal of Monetary Economics},
	shortjournal = {Journal of Monetary Economics},
	author = {Baumeister, Christiane and Hamilton, James D.},
	date = {2018},
	langid = {english},
}

@article{mertens_empirical_2012,
	title = {Empirical Evidence on the Aggregate Effects of Anticipated and Unanticipated {US} Tax Policy Shocks},
	volume = {4},
	issn = {1945-7731, 1945-774X},
	doi = {10.1257/pol.4.2.145},
	pages = {145--181},
	number = {2},
	journaltitle = {American Economic Journal: Economic Policy},
	shortjournal = {American Economic Journal: Economic Policy},
	author = {Mertens, Karel R. S. M. and Ravn, Morten O.},
	date = {2012},
	langid = {english},
}

@article{ramey_costly_1998,
	title = {Costly capital reallocation and the effects of government spending},
	volume = {48},
	pages = {145--194},
	journaltitle = {Carnegie-Rochester Conference Series on Public Policy},
	author = {Ramey, Valerie A. and Shapiro, D},
	date = {1998},
	langid = {english},
}

@article{mumtaz_fiscal_2020,
	title = {Fiscal policy shocks and stock prices in the United States},
	volume = {129},
	issn = {0014-2921},
	doi = {10.1016/j.euroecorev.2020.103562},
	pages = {103562},
	journaltitle = {European Economic Review},
	shortjournal = {European Economic Review},
	author = {Mumtaz, Haroon and Theodoridis, Konstantinos},
	date = {2020},
}

@article{leeper_equilibria_1991,
	title = {Equilibria under ‘active’ and ‘passive’ monetary and fiscal policies},
	volume = {27},
	issn = {0304-3932},
	doi = {10.1016/0304-3932(91)90007-B},
	pages = {129--147},
	number = {1},
	journaltitle = {Journal of Monetary Economics},
	shortjournal = {Journal of Monetary Economics},
	author = {Leeper, Eric M.},
	date = {1991},
}

\appendix

\section{Differences of the values of fiscal multipliers for disjoint categories of total general government expenditure} \label{appndx}
\begin{figure}[htbp]
	\centering
	\includegraphics[width=\textwidth]{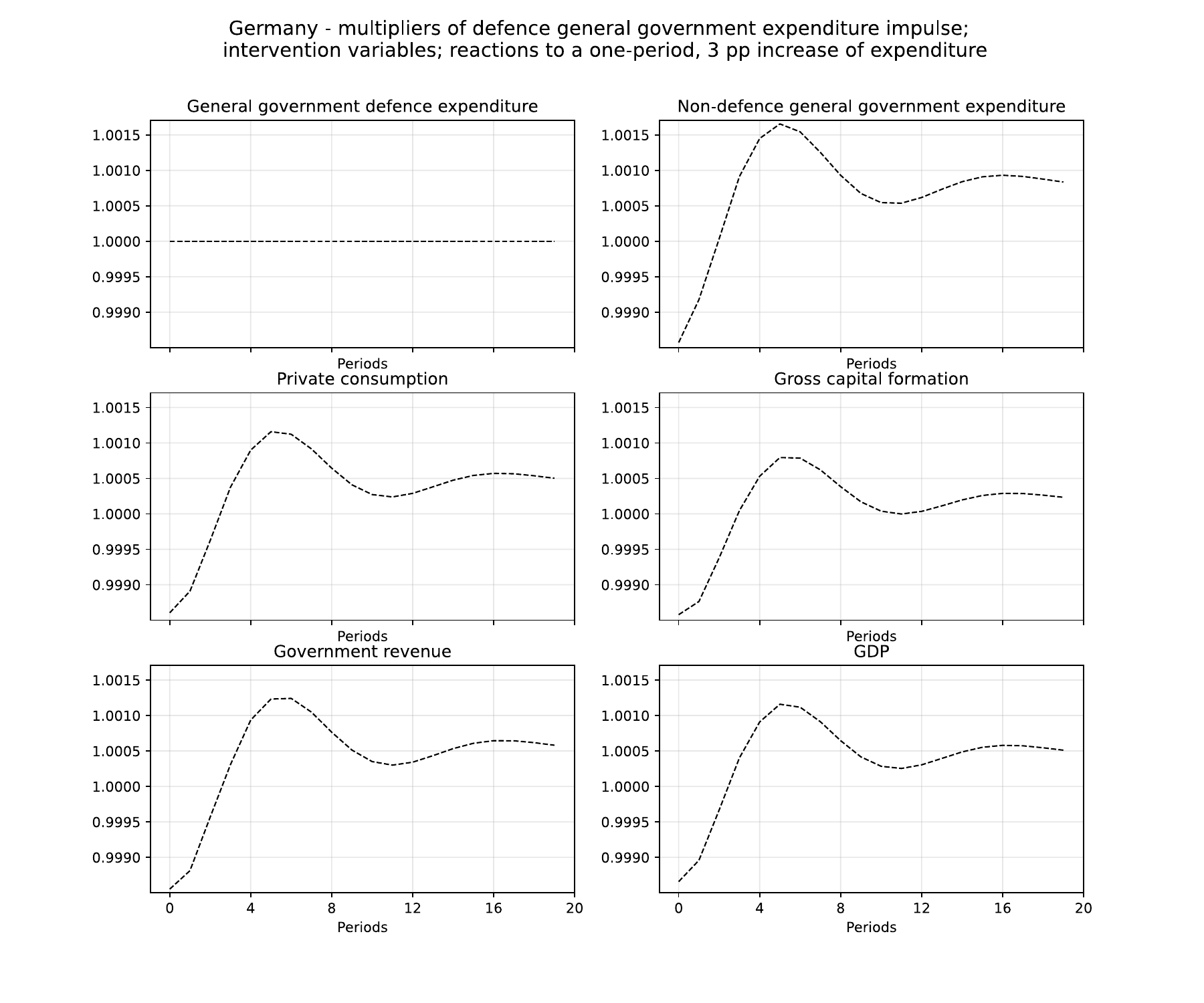}
	\caption{Germany. Multipliers of an impulse of defence general government expenditure, with intervention variables; reaction to a one-period increase of expenditure of 3 percentage points. The units of the impulse are percentage points.}
	\label{fig:de004}
\end{figure}
Multipliers of various categories of total government expenditure differ. This is illustrated by, e.g., figure \ref{fig:de003}, showing the multipliers for the increase of the gross growth rate of non-defence total general government expenditure by $0.03$ percentage points, and figure \ref{fig:de004}, demonstrating the values of multipliers for analogous increase of the gross growth rate of defence total general government expenditure.

\section{Proof of Proposition \ref{prop:minnorm}}\label{app:proof}

\subsection{Notation and assumptions}

Throughout, the model \ref{eq:3} is written as ${Y}={X}{B}+{U}$, with ${Y}$ of dimension $n\times N$, ${X}$ of dimension $n\times K$ and ${B}$ of dimension $K\times N$, where $n:=T-p$ is the number of effective observations, $K=1+Np+M$, and the rows of ${X}$ are
\begin{equation}\label{A:row}
	{x}_{t}^{\prime}=\bigl(1,\;{y}_{t-1}^{\prime},\dots,{y}_{t-p}^{\prime},\;{d}_{t}^{\prime}\bigr),\qquad t=p+1,\dots,T .
\end{equation}
The rows of ${B}$ are ordered conformably with ${x}_{t}^{\prime}$: the first row is ${c}^{\prime}$, rows $2+(l-1)N,\dots,1+lN$ form $\Phi_{l}^{\prime}$ for $l=1,\dots,p$, and row $k_{m}=1+Np+m$ is $\boldsymbol{\psi}_{m}^{\prime}$ for $m=1,\dots,M$. Hence
\begin{equation}\label{A:rows}
	{e}_{1}^{\prime}{B}={c}^{\prime},\qquad {e}_{k_{m}}^{\prime}{B}=\boldsymbol{\psi}_{m}^{\prime},\qquad
	{x}_{t}^{\prime}{e}_{1}=1,\qquad {x}_{t}^{\prime}{e}_{k_{m}}=d_{mt}.
\end{equation}
The index set $\mathcal{A}:=\{2,\dots,1+Np\}$ of the rows carrying the autoregressive coefficients is called the autoregressive block. The null vectors \ref{eq:11} are collected in
\begin{equation}\label{A:V}
	{V}:=\Bigl[\;{v}_{0}\;\Bigm|\;({e}_{k_{m}})_{m\in\mathcal{Z}}\;\Bigr]\in\mathbb{R}^{K\times q},
	\qquad
	{v}_{0}={e}_{1}-\sum_{m\in\mathcal{S}}{e}_{k_{m}},
	\qquad q:=1+\lvert\mathcal{Z}\rvert ,
\end{equation}
and $\widehat{{B}}:=({X}^{\prime}{X})^{+}{X}^{\prime}{Y}$ is the estimator \ref{eq:10}. $\lVert\cdot\rVert_{F}$ is the Frobenius norm, $\lVert {A}\rVert_{F}^{2}=\operatorname{tr}({A}^{\prime}{A})$.

In addition to \ref{eq:7} and $\mathcal{S}\cap\mathcal{Z}=\emptyset$, the following are assumed.
\begin{enumerate}
	\item[(A1)] The columns of ${V}$ span $\mathcal{N}({X})$. This is the hypothesis of the proposition: apart from the dependence created by the partition and by the empty regimes the design contains no other exact linear dependence. It is a property of the data and not a consequence of \ref{eq:7}.
	\item[(A2)] $n>K$, so that the divisor $n-K$ in \ref{eq:4} is positive, and $\widehat{\Sigma}_{u}$ is positive definite, so that the Cholesky factorisation \ref{eq:13} exists.
	\item[(A3)] The threshold $\tau$ in \ref{eq:9} retains exactly the nonzero singular values, so that $r=\operatorname{rank}({X})$.
\end{enumerate}
By (A3), ${X}$ has the compact singular value decomposition
\begin{equation}\label{A:svd}
	{X}={U}_{r}{S}_{r}{V}_{r}^{\prime},\qquad {U}_{r}^{\prime}{U}_{r}={I}_{r},\qquad {V}_{r}^{\prime}{V}_{r}={I}_{r},\qquad {S}_{r}=\operatorname{diag}(s_{1},\dots,s_{r}),\quad s_{j}>0 .
\end{equation}

\subsection{Two properties of the estimator}

\begin{lemma}\label{A:lem1}
	Under (A3):
	\begin{enumerate}
		\item[(a)] $({X}^{\prime}{X})^{+}{X}^{\prime}={V}_{r}{S}_{r}^{-1}{U}_{r}^{\prime}={X}^{+}$, so that $\widehat{{B}}={X}^{+}{Y}={V}_{r}{S}_{r}^{-1}{U}_{r}^{\prime}{Y}$;
		\item[(b)] the residuals $\widehat{{U}}={Y}-{X}\widehat{{B}}$ satisfy the normal equations ${X}^{\prime}\widehat{{U}}={0}$;
		\item[(c)] $\mathcal{N}({X})=\mathcal{N}({V}_{r}^{\prime})$, and ${w}^{\prime}\widehat{{B}}={0}$ for every ${w}\in\mathcal{N}({X})$.
	\end{enumerate}
\end{lemma}

\begin{proof}
	(a) By \ref{A:svd}, ${X}^{\prime}{X}={V}_{r}{S}_{r}{U}_{r}^{\prime}{U}_{r}{S}_{r}{V}_{r}^{\prime}={V}_{r}{S}_{r}^{2}{V}_{r}^{\prime}$. Put ${G}:={V}_{r}{S}_{r}^{-2}{V}_{r}^{\prime}$. Using ${V}_{r}^{\prime}{V}_{r}={I}_{r}$,
	\begin{align}
		({X}^{\prime}{X}){G}({X}^{\prime}{X}) &= {V}_{r}{S}_{r}^{2}{S}_{r}^{-2}{S}_{r}^{2}{V}_{r}^{\prime}={X}^{\prime}{X}, &
		{G}({X}^{\prime}{X}){G} &= {V}_{r}{S}_{r}^{-2}{S}_{r}^{2}{S}_{r}^{-2}{V}_{r}^{\prime}={G}, \label{A:pen1}\\
		({X}^{\prime}{X}){G} &= {V}_{r}{V}_{r}^{\prime}=\bigl(({X}^{\prime}{X}){G}\bigr)^{\prime}, &
		{G}({X}^{\prime}{X}) &= {V}_{r}{V}_{r}^{\prime}=\bigl({G}({X}^{\prime}{X})\bigr)^{\prime}. \label{A:pen2}
	\end{align}
	These are the four Moore--Penrose conditions, and since the Moore--Penrose inverse is unique, ${G}=({X}^{\prime}{X})^{+}$. Therefore, with ${X}^{\prime}={V}_{r}{S}_{r}{U}_{r}^{\prime}$,
	\begin{equation}\label{eq:pinv}
		({X}^{\prime}{X})^{+}{X}^{\prime}={V}_{r}{S}_{r}^{-2}{V}_{r}^{\prime}{V}_{r}{S}_{r}{U}_{r}^{\prime}={V}_{r}{S}_{r}^{-1}{U}_{r}^{\prime}.
	\end{equation}
	The same computation with ${H}:={V}_{r}{S}_{r}^{-1}{U}_{r}^{\prime}$ gives ${X}{H}{X}={X}$, ${H}{X}{H}={H}$, ${X}{H}={U}_{r}{U}_{r}^{\prime}$ and ${H}{X}={V}_{r}{V}_{r}^{\prime}$, both symmetric, so ${H}={X}^{+}$.
	
	(b) From (a),
	\begin{equation}\label{A:XtXXp}
		{X}^{\prime}{X}{X}^{+}={V}_{r}{S}_{r}^{2}{V}_{r}^{\prime}{V}_{r}{S}_{r}^{-1}{U}_{r}^{\prime}={V}_{r}{S}_{r}{U}_{r}^{\prime}={X}^{\prime},
	\end{equation}
	and therefore
	\begin{equation}\label{eq:normalequations}
		{X}^{\prime}\widehat{{U}}={X}^{\prime}{Y}-{X}^{\prime}{X}{X}^{+}{Y}={X}^{\prime}{Y}-{X}^{\prime}{Y}={0}.
	\end{equation}
	
	(c) If ${V}_{r}^{\prime}{a}={0}$ then ${X}{a}={U}_{r}{S}_{r}{V}_{r}^{\prime}{a}={0}$. Conversely, if ${X}{a}={0}$ then ${0}={S}_{r}^{-1}{U}_{r}^{\prime}{X}{a}={S}_{r}^{-1}{U}_{r}^{\prime}{U}_{r}{S}_{r}{V}_{r}^{\prime}{a}={V}_{r}^{\prime}{a}$. Hence $\mathcal{N}({X})=\mathcal{N}({V}_{r}^{\prime})$, and for ${w}\in\mathcal{N}({X})$,
	\begin{equation}\label{eq:orthogonal}
		{w}^{\prime}\widehat{{B}}={w}^{\prime}{V}_{r}{S}_{r}^{-1}{U}_{r}^{\prime}{Y}=({V}_{r}^{\prime}{w})^{\prime}{S}_{r}^{-1}{U}_{r}^{\prime}{Y}={0},
	\end{equation}
	that is, every column of $\widehat{{B}}$ is orthogonal to $\mathcal{N}({X})$.
\end{proof}

\subsection{The set of least squares solutions}

\begin{lemma}\label{A:lem2}
	Under \ref{eq:7}, $\mathcal{S}\cap\mathcal{Z}=\emptyset$, (A1) and (A3):
	\begin{enumerate}
		\item[(a)] ${X}{V}={0}$, and
		\begin{equation}\label{A:gram}
			{V}^{\prime}{V}=\operatorname{diag}\bigl(1+\lvert\mathcal{S}\rvert,\;{I}_{\lvert\mathcal{Z}\rvert}\bigr),
		\end{equation}
		so the columns of ${V}$ are mutually orthogonal and linearly independent, and $\operatorname{rank}({X})=K-q=K-1-\lvert\mathcal{Z}\rvert$;
		\item[(b)] the set of least squares solutions is
		\begin{equation}\label{eq:solutionset}
			\mathcal{B}:=\operatorname*{arg\,min}_{{B}}\lVert{Y}-{X}{B}\rVert_{F}^{2}=\bigl\{\widehat{{B}}+{V}{A}:\;{A}\in\mathbb{R}^{q\times N}\bigr\},
		\end{equation}
		and each member of $\mathcal{B}$ corresponds to exactly one ${A}$;
		\item[(c)] ${e}_{j}^{\prime}{V}={0}$ for every $j\in\mathcal{A}$.
	\end{enumerate}
\end{lemma}

\begin{proof}
	(a) By \ref{A:rows} and \ref{eq:7}, for every $t=p+1,\dots,T$,
	\begin{equation}\label{A:v0null}
		{x}_{t}^{\prime}{v}_{0}={x}_{t}^{\prime}{e}_{1}-\sum_{m\in\mathcal{S}}{x}_{t}^{\prime}{e}_{k_{m}}=1-\sum_{m\in\mathcal{S}}d_{mt}=0,
	\end{equation}
	and for $m\in\mathcal{Z}$, ${x}_{t}^{\prime}{e}_{k_{m}}=d_{mt}=0$ by the definition of $\mathcal{Z}$. Hence ${X}{V}={0}$. Since ${e}_{1}$ and the ${e}_{k_{m}}$ are distinct unit vectors ($k_{m}\geq 2+Np>1$),
	\begin{align}
		{v}_{0}^{\prime}{v}_{0}&={e}_{1}^{\prime}{e}_{1}+\sum_{m\in\mathcal{S}}\sum_{m^{\prime}\in\mathcal{S}}{e}_{k_{m}}^{\prime}{e}_{k_{m^{\prime}}}=1+\lvert\mathcal{S}\rvert, \label{A:g1}\\
		{v}_{0}^{\prime}{e}_{k_{m}}&={e}_{1}^{\prime}{e}_{k_{m}}-\sum_{m^{\prime}\in\mathcal{S}}{e}_{k_{m^{\prime}}}^{\prime}{e}_{k_{m}}=0-\mathbf{1}\{m\in\mathcal{S}\}=0,\qquad m\in\mathcal{Z}, \label{A:g2}\\
		{e}_{k_{m}}^{\prime}{e}_{k_{m^{\prime}}}&=\mathbf{1}\{m=m^{\prime}\},\qquad m,m^{\prime}\in\mathcal{Z}, \label{A:g3}
	\end{align}
	where \ref{A:g2} uses $\mathcal{S}\cap\mathcal{Z}=\emptyset$; this is the only place where that hypothesis is used. Together these give \ref{A:gram}. As ${V}^{\prime}{V}$ is diagonal with positive entries it is nonsingular, and ${V}{a}={0}$ implies ${V}^{\prime}{V}{a}={0}$ and so ${a}={0}$: the $q$ columns are independent. By (A1) they form a basis of $\mathcal{N}({X})$, so $\dim\mathcal{N}({X})=q$ and, by the rank--nullity theorem, $\operatorname{rank}({X})=K-q=K-1-\lvert\mathcal{Z}\rvert$.
	
	(b) For any ${B}$ put ${\Delta}:={B}-\widehat{{B}}$, so that ${Y}-{X}{B}=\widehat{{U}}-{X}{\Delta}$ and
	\begin{align}
		\lVert{Y}-{X}{B}\rVert_{F}^{2}
		&=\operatorname{tr}\bigl[(\widehat{{U}}-{X}{\Delta})^{\prime}(\widehat{{U}}-{X}{\Delta})\bigr]
		=\lVert\widehat{{U}}\rVert_{F}^{2}-2\operatorname{tr}\bigl({\Delta}^{\prime}{X}^{\prime}\widehat{{U}}\bigr)+\lVert{X}{\Delta}\rVert_{F}^{2} \nonumber\\
		&=\lVert\widehat{{U}}\rVert_{F}^{2}+\lVert{X}{\Delta}\rVert_{F}^{2}, \label{A:cross}
	\end{align}
	the cross term vanishing by \ref{eq:normalequations}. The sum of squares is therefore minimised exactly when ${X}{\Delta}={0}$, that is, when every column of ${\Delta}$ lies in $\mathcal{N}({X})$. By (a) and (A1) this holds if and only if ${\Delta}={V}{A}$ for some ${A}\in\mathbb{R}^{q\times N}$, which is unique because the columns of ${V}$ are independent. In particular $\widehat{{B}}\in\mathcal{B}$ (take ${A}={0}$).
	
	(c) The support of ${v}_{0}$ is $\{1\}\cup\{k_{m}:m\in\mathcal{S}\}$ and that of ${e}_{k_{m}}$ is $\{k_{m}\}$. Since $1<2$ and $k_{m}=1+Np+m>1+Np$, none of these indices belongs to $\mathcal{A}=\{2,\dots,1+Np\}$.
\end{proof}

\subsection{Part 1}

Let $\widetilde{{B}}=\widehat{{B}}+{V}{A}$ be an arbitrary least squares solution. Since ${X}{V}={0}$,
\begin{equation}\label{A:fit}
	{X}\widetilde{{B}}={X}\widehat{{B}}+{X}{V}{A}={X}\widehat{{B}}={X}{X}^{+}{Y}={U}_{r}{U}_{r}^{\prime}{Y},
	\qquad
	\widetilde{{U}}:={Y}-{X}\widetilde{{B}}={Y}-{X}\widehat{{B}}=\widehat{{U}},
\end{equation}
so the fitted values, which are the orthogonal projection of ${Y}$ on the column space of ${X}$, and the residuals are the same for every solution. The covariance estimator of \ref{eq:4} depends on the residuals only,
\begin{equation}\label{A:sigma}
	\frac{\widetilde{{U}}^{\prime}\widetilde{{U}}}{n-K}=\frac{\widehat{{U}}^{\prime}\widehat{{U}}}{n-K}=\widehat{\Sigma}_{u},
\end{equation}
and is invariant. Since the first column of ${X}$ is the vector of ones $\boldsymbol{\iota}_{n}$, the first row of \ref{eq:normalequations} reads
\begin{equation}\label{A:mean}
	\boldsymbol{\iota}_{n}^{\prime}\widehat{{U}}={0}^{\prime},\qquad\text{that is}\qquad \overline{\widehat{{u}}}=\frac{1}{n}\sum_{t=p+1}^{T}\widehat{{u}}_{t}={0},
\end{equation}
so the residuals have mean zero exactly.

\emph{Uniqueness of the Cholesky factor.} Suppose $\widehat{\Sigma}_{u}={P}_{1}{P}_{1}^{\prime}={P}_{2}{P}_{2}^{\prime}$ with ${P}_{1},{P}_{2}$ lower triangular with strictly positive diagonals, hence invertible. Multiplying ${P}_{1}{P}_{1}^{\prime}={P}_{2}{P}_{2}^{\prime}$ on the left by ${P}_{2}^{-1}$ and on the right by $({P}_{1}^{\prime})^{-1}$,
\begin{equation}\label{A:chol1}
	{D}:={P}_{2}^{-1}{P}_{1}={P}_{2}^{\prime}({P}_{1}^{\prime})^{-1}.
\end{equation}
The inverse of a lower triangular matrix is lower triangular and the product of lower triangular matrices is lower triangular, so the left-hand side is lower triangular; by the same argument for upper triangular matrices the right-hand side is upper triangular. Hence ${D}$ is diagonal, with entries $D_{ii}=[{P}_{1}]_{ii}/[{P}_{2}]_{ii}>0$. Moreover
\begin{equation}\label{A:chol2}
	{D}{D}^{\prime}={P}_{2}^{-1}{P}_{1}{P}_{1}^{\prime}({P}_{2}^{\prime})^{-1}={P}_{2}^{-1}\widehat{\Sigma}_{u}({P}_{2}^{\prime})^{-1}={P}_{2}^{-1}{P}_{2}{P}_{2}^{\prime}({P}_{2}^{\prime})^{-1}={I}_{N},
\end{equation}
so $D_{ii}^{2}=1$ and $D_{ii}>0$ give ${D}={I}_{N}$ and ${P}_{1}={P}_{2}$. The factor ${P}$ of \ref{eq:13} is thus a function of $\widehat{\Sigma}_{u}$ alone and is invariant, and so are the structural shocks $\widehat{\boldsymbol{\varepsilon}}_{t}={P}^{-1}\widehat{{u}}_{t}$. By \ref{A:mean} their sample mean is zero, and their sample second moment, with the divisor of \ref{eq:4}, is
\begin{equation}\label{A:epsmoment}
	\frac{1}{n-K}\sum_{t=p+1}^{T}\widehat{\boldsymbol{\varepsilon}}_{t}\widehat{\boldsymbol{\varepsilon}}_{t}^{\prime}
	={P}^{-1}\frac{\widehat{{U}}^{\prime}\widehat{{U}}}{n-K}({P}^{\prime})^{-1}
	={P}^{-1}{P}{P}^{\prime}({P}^{\prime})^{-1}={I}_{N},
\end{equation}
which is the identity reported after \ref{eq:13}; in particular $\sigma_{j}=1$ in \ref{eq:16}.

\emph{Autoregressive matrices and impulse responses.} By Lemma \ref{A:lem2}(c), for every $j\in\mathcal{A}$,
\begin{equation}\label{A:ar}
	{e}_{j}^{\prime}\widetilde{{B}}={e}_{j}^{\prime}\widehat{{B}}+({e}_{j}^{\prime}{V}){A}={e}_{j}^{\prime}\widehat{{B}},
	\qquad\text{hence}\qquad
	\widetilde{\Phi}_{l}=\widehat{\Phi}_{l},\quad l=1,\dots,p .
\end{equation}
The recursion \ref{eq:14} generates $\Xi_{h}$ from $\Xi_{0}={I}_{N}$ and $\Phi_{1},\dots,\Phi_{p}$ only. By induction on $h$: $\Xi_{0}$ does not depend on the solution; if $\Xi_{0},\dots,\Xi_{h-1}$ do not, then neither does $\Xi_{h}=\sum_{l=1}^{\min(h,p)}\Phi_{l}\Xi_{h-l}$, by \ref{A:ar}. By \ref{eq:15}, $\Theta_{h}=\Xi_{h}{P}$ is invariant for every $h=0,1,\dots,H$. This proves part 1.

\begin{remark}\label{A:rem1}
	The invariance of the fitted values, residuals, $\widehat{\Sigma}_{u}$, ${P}$ and $\widehat{\boldsymbol{\varepsilon}}_{t}$ follows from \ref{A:cross} alone and holds for any two least squares solutions whatever $\mathcal{N}({X})$ is. The invariance of $\Phi_{l}$, and hence of $\Theta_{h}$, is where (A1) is needed: a null direction with a nonzero entry in $\mathcal{A}$ would make $\Phi_{l}$ indeterminate. (A1) secures that the dynamics are point identified and that only the split of the deterministic level is not.
\end{remark}

\subsection{Part 2}

A linear function ${\lambda}^{\prime}{B}\in\mathbb{R}^{1\times N}$, ${\lambda}\in\mathbb{R}^{K}$, is called identified by the data if ${\lambda}^{\prime}{B}_{1}={\lambda}^{\prime}{B}_{2}$ whenever ${X}{B}_{1}={X}{B}_{2}$, that is, if it is determined by the fitted values.

\emph{Observational equivalence.} Two parameter matrices with ${X}{B}_{1}={X}{B}_{2}$ are indistinguishable, and this holds not only on the observed sample. By Lemma \ref{A:lem2}(c) the rows of ${V}{A}$ in $\mathcal{A}$ are zero, so, writing ${a}_{0}^{\prime}$ for the first row of ${A}$ and ${a}_{m}^{\prime}$ for the row that multiplies ${e}_{k_{m}}$, for every $t=p+1,\dots,T$ and \emph{every} value of the lagged endogenous variables,
\begin{equation}\label{A:obsequiv}
	{x}_{t}^{\prime}{V}{A}=\Bigl(1-\sum_{m\in\mathcal{S}}d_{mt}\Bigr){a}_{0}^{\prime}+\sum_{m\in\mathcal{Z}}d_{mt}\,{a}_{m}^{\prime}={0}^{\prime}.
\end{equation}
Hence ${B}$ and ${B}+{V}{A}$ assign the same conditional mean ${x}_{t}^{\prime}{B}$ to ${y}_{t}$ given its past in every period and on every path, and, with the same $\Sigma_{u}$, generate the same distribution of $({y}_{p+1},\dots,{y}_{T})$ given the initial values and the intervention variables, whatever the distribution of ${u}_{t}$.

\emph{Criterion.} By Lemma \ref{A:lem2}, ${X}{B}_{1}={X}{B}_{2}$ if and only if ${B}_{1}-{B}_{2}={V}{A}$ for some ${A}$. If ${\lambda}^{\prime}{V}={0}$, then ${\lambda}^{\prime}({B}+{V}{A})={\lambda}^{\prime}{B}$ for every ${A}$. If ${\lambda}^{\prime}{V}\neq{0}$, let ${v}$ be a column of ${V}$ with ${\lambda}^{\prime}{v}\neq 0$ and take ${A}$ with an arbitrary row ${a}^{\prime}$ in the position of ${v}$ and zeros elsewhere; then
\begin{equation}\label{A:notid}
	{\lambda}^{\prime}({B}+{V}{A})-{\lambda}^{\prime}{B}=({\lambda}^{\prime}{v})\,{a}^{\prime},
\end{equation}
which ranges over the whole of $\mathbb{R}^{1\times N}$ as ${a}$ varies. Therefore
\begin{equation}\label{A:crit}
	{\lambda}^{\prime}{B}\ \text{is identified}
	\iff {\lambda}^{\prime}{V}={0}
	\iff {\lambda}\perp\mathcal{N}({X})
	\iff {\lambda}\in\operatorname{col}({X}^{\prime}),
\end{equation}
the second equivalence by (A1) and the third because $\operatorname{col}({X}^{\prime})=\mathcal{N}({X})^{\perp}$. Written out over the columns of ${V}$, using ${\lambda}^{\prime}{v}_{0}=\lambda_{1}-\sum_{m\in\mathcal{S}}\lambda_{k_{m}}$ and ${\lambda}^{\prime}{e}_{k_{m}}=\lambda_{k_{m}}$, the condition is
\begin{equation}\label{eq:estimability}
	\lambda_{1}=\sum_{m\in\mathcal{S}}\lambda_{k_{m}}
	\qquad\text{and}\qquad
	\lambda_{k_{m}}=0\quad\text{for every }m\in\mathcal{Z},
\end{equation}
and it imposes nothing on the coordinates in $\mathcal{A}$.

\emph{Reading off the identified quantities.} Applying \ref{eq:estimability}:
\begin{enumerate}
	\item[(i)] every entry of $\Phi_{1},\dots,\Phi_{p}$ corresponds to ${\lambda}={e}_{j}$, $j\in\mathcal{A}$, with $\lambda_{1}=0=\sum_{m\in\mathcal{S}}\lambda_{k_{m}}$: identified, in agreement with part 1;
	\item[(ii)] for $m\in\mathcal{S}$, ${\lambda}={e}_{1}+{e}_{k_{m}}$ gives $\lambda_{1}=1=\sum_{m^{\prime}\in\mathcal{S}}\lambda_{k_{m^{\prime}}}$: ${c}+\boldsymbol{\psi}_{m}$ is identified;
	\item[(iii)] ${\lambda}={e}_{1}$ gives $1\neq 0$, and ${\lambda}={e}_{k_{m}}$, $m\in\mathcal{S}$, gives $0\neq 1$: neither ${c}$ nor $\boldsymbol{\psi}_{m}$, $m\in\mathcal{S}$, is identified on its own;
	\item[(iv)] for $m,m^{\prime}\in\mathcal{S}$, ${\lambda}={e}_{k_{m}}-{e}_{k_{m^{\prime}}}$ gives $0=1-1$: $\boldsymbol{\psi}_{m}-\boldsymbol{\psi}_{m^{\prime}}$ is identified;
	\item[(v)] for a nested regime $m\notin\mathcal{S}\cup\mathcal{Z}$, ${\lambda}={e}_{k_{m}}$ gives $0=0$, while ${\lambda}={e}_{1}+{e}_{k_{m}}$ gives $1\neq 0$: $\boldsymbol{\psi}_{m}$ is identified on its own and ${c}+\boldsymbol{\psi}_{m}$ is not;
	\item[(vi)] for $m\in\mathcal{Z}$, ${e}_{k_{m}}$ violates the second condition and ${e}_{1}+{e}_{k_{m}}$ the first; more strongly, for any ${\boldsymbol{\psi}}^{\ast}\in\mathbb{R}^{N}$ the matrix $\widehat{{B}}+{e}_{k_{m}}({\boldsymbol{\psi}}^{\ast}-\widehat{\boldsymbol{\psi}}_{m})^{\prime}\in\mathcal{B}$ has $\boldsymbol{\psi}^{\ast}$ in row $k_{m}$, so nothing whatever is identified about $\boldsymbol{\psi}_{m}$;
	\item[(vii)] for every period $t=p+1,\dots,T$ the deterministic level
	\begin{equation}\label{A:deltat}
		\boldsymbol{\delta}_{t}:={c}+\Psi{d}_{t}={c}+\sum_{m=1}^{M}\boldsymbol{\psi}_{m}d_{mt}
	\end{equation}
	corresponds to ${\lambda}={e}_{1}+\sum_{m}d_{mt}{e}_{k_{m}}$, for which $\lambda_{1}=1=\sum_{m\in\mathcal{S}}d_{mt}$ by \ref{eq:7} and $\lambda_{k_{m}}=d_{mt}=0$ for $m\in\mathcal{Z}$: it is identified. In a period in which a regime $m^{\prime}\in\mathcal{S}$ and a regime $m$ nested in it are both in force, $\boldsymbol{\delta}_{t}=({c}+\boldsymbol{\psi}_{m^{\prime}})+\boldsymbol{\psi}_{m}$, the sum of the quantities identified in (ii) and (v).
\end{enumerate}
When $\mathcal{S}\cup\mathcal{Z}=\{1,\dots,M\}$, as for disjoint and exhaustive regimes, (v) is void and the identified deterministic quantities are exactly ${c}+\boldsymbol{\psi}_{m}$, $m\notin\mathcal{Z}$. When one indicator has to be left out of $\mathcal{S}$, as in the French dating, (v) applies to it. This proves part 2.

\emph{Stationary means.} If ${I}_{N}-\sum_{l=1}^{p}\Phi_{l}$ is nonsingular, the implied stationary mean \ref{eq:20} of a regime $m\in\mathcal{S}$,
\begin{equation}\label{A:mu}
	\boldsymbol{\mu}_{m}=\Bigl({I}_{N}-\sum_{l=1}^{p}\Phi_{l}\Bigr)^{-1}({c}+\boldsymbol{\psi}_{m}),
\end{equation}
is a product of a matrix identified by (i) and a vector identified by (ii), and is invariant. For a period in which a nested regime $m$ is in force inside $m^{\prime}\in\mathcal{S}$, the invariant object is $\bigl({I}_{N}-\sum_{l}\Phi_{l}\bigr)^{-1}({c}+\boldsymbol{\psi}_{m^{\prime}}+\boldsymbol{\psi}_{m})$, and in general $\bigl({I}_{N}-\sum_{l}\Phi_{l}\bigr)^{-1}\boldsymbol{\delta}_{t}$ for any $t$, by (vii); the expression \ref{eq:20} evaluated with ${c}+\boldsymbol{\psi}_{m}$ for the nested $m$ alone is not invariant.

\subsection{Part 3}

By Lemma \ref{A:lem1}(c) applied to each column of ${V}\subset\mathcal{N}({X})$,
\begin{equation}\label{eq:Vorthogonality}
	{V}^{\prime}\widehat{{B}}={0}.
\end{equation}
Let $\widetilde{{B}}=\widehat{{B}}+{V}{A}\in\mathcal{B}$. Then
\begin{align}
	\lVert\widetilde{{B}}\rVert_{F}^{2}
	&=\lVert\widehat{{B}}\rVert_{F}^{2}+2\operatorname{tr}\bigl(\widehat{{B}}^{\prime}{V}{A}\bigr)+\lVert{V}{A}\rVert_{F}^{2}
	=\lVert\widehat{{B}}\rVert_{F}^{2}+2\operatorname{tr}\bigl(({V}^{\prime}\widehat{{B}})^{\prime}{A}\bigr)+\operatorname{tr}\bigl({A}^{\prime}{V}^{\prime}{V}{A}\bigr) \nonumber\\
	&=\lVert\widehat{{B}}\rVert_{F}^{2}+(1+\lvert\mathcal{S}\rvert)\lVert{a}_{0}\rVert^{2}+\sum_{m\in\mathcal{Z}}\lVert{a}_{m}\rVert^{2}, \label{A:pyth}
\end{align}
by \ref{eq:Vorthogonality} and \ref{A:gram}, where ${a}_{0}^{\prime},({a}_{m}^{\prime})_{m\in\mathcal{Z}}$ are the rows of ${A}$. The last two terms are nonnegative and vanish only if ${A}={0}$. Hence $\widehat{{B}}$ is the unique least squares solution of smallest Frobenius norm. It is also the unique solution orthogonal to $\mathcal{N}({X})$: if ${V}^{\prime}\widetilde{{B}}={0}$ then, by \ref{eq:Vorthogonality}, ${V}^{\prime}{V}{A}={0}$ and so ${A}={0}$ because ${V}^{\prime}{V}$ is nonsingular.

Writing ${V}^{\prime}\widetilde{{B}}={0}$ row by row, with \ref{A:rows},
\begin{equation}\label{A:rowwise}
	{v}_{0}^{\prime}\widetilde{{B}}=\widetilde{{c}}^{\prime}-\sum_{m\in\mathcal{S}}\widetilde{\boldsymbol{\psi}}_{m}^{\prime}={0}^{\prime},
	\qquad
	{e}_{k_{m}}^{\prime}\widetilde{{B}}=\widetilde{\boldsymbol{\psi}}_{m}^{\prime}={0}^{\prime},\quad m\in\mathcal{Z},
\end{equation}
so orthogonality to $\mathcal{N}({X})$ is equivalent to \ref{eq:12}, and $\widehat{{B}}$ is the unique member of $\mathcal{B}$ satisfying \ref{eq:12}. Because the columns of ${V}$ are mutually orthogonal, the two sets of conditions do not interact: the constant equals the sum of the coefficients of the regimes in the partition alone, and the coefficients of the empty regimes are set to zero rather than absorbed into it.

\emph{Explicit form.} For any $\widetilde{{B}}\in\mathcal{B}$, premultiplying $\widetilde{{B}}=\widehat{{B}}+{V}{A}$ by ${V}^{\prime}$ and using \ref{eq:Vorthogonality} gives ${A}=({V}^{\prime}{V})^{-1}{V}^{\prime}\widetilde{{B}}$, so
\begin{equation}\label{A:proj}
	\widehat{{B}}=\bigl({I}_{K}-{V}({V}^{\prime}{V})^{-1}{V}^{\prime}\bigr)\widetilde{{B}},
\end{equation}
and by \ref{A:gram}, with $\widetilde{\boldsymbol{\ell}}:=\bigl(\widetilde{{c}}-\sum_{m\in\mathcal{S}}\widetilde{\boldsymbol{\psi}}_{m}\bigr)/(1+\lvert\mathcal{S}\rvert)$,
\begin{equation}\label{A:explicit}
	\widehat{{c}}=\widetilde{{c}}-\widetilde{\boldsymbol{\ell}},\qquad
	\widehat{\boldsymbol{\psi}}_{m}=\widetilde{\boldsymbol{\psi}}_{m}+\widetilde{\boldsymbol{\ell}}\ \ (m\in\mathcal{S}),\qquad
	\widehat{\boldsymbol{\psi}}_{m}={0}\ \ (m\in\mathcal{Z}),\qquad
	\widehat{\boldsymbol{\psi}}_{m}=\widetilde{\boldsymbol{\psi}}_{m}\ \ (m\notin\mathcal{S}\cup\mathcal{Z}),
\end{equation}
the autoregressive rows being unchanged. In particular, taking for $\widetilde{{B}}$ the solution with $\widetilde{{c}}={0}$, in which $\widetilde{\boldsymbol{\psi}}_{m}={c}+\boldsymbol{\psi}_{m}=:\boldsymbol{\ell}_{m}$ are the identified regime levels of part 2,
\begin{equation}\label{A:chat}
	\widehat{{c}}=\frac{1}{1+\lvert\mathcal{S}\rvert}\sum_{m\in\mathcal{S}}\boldsymbol{\ell}_{m},
	\qquad
	\widehat{\boldsymbol{\psi}}_{m}=\boldsymbol{\ell}_{m}-\widehat{{c}},\quad m\in\mathcal{S}.
\end{equation}
This completes the proof. \hfill$\square$

\subsection{Consequences}

\emph{The multiplier.} By \ref{A:epsmoment}, $\sigma_{j}=1$, so the scaled responses \ref{eq:16} are $\theta^{(\delta)}_{i,j,h}=\delta[\Theta_{h}]_{ij}$ and are invariant by part 1. The multiplier \ref{eq:multiplier} is a function of $\Theta_{0},\dots,\Theta_{H}$ and $\delta$ only, and is therefore invariant for every variable, horizon and impulse size at which its denominator does not vanish. The nonlinearity in $\delta$ of Section \ref{s3.6} is a property of the estimated system and not of the normalisation.

\emph{The bootstrap.} Each pseudo-sample \ref{eq:19} is generated from $\widehat{{c}}+\widehat{\Psi}{d}_{t}=\widehat{\boldsymbol{\delta}}_{t}$, identified by part 2(vii), from $\widehat{\Phi}_{l}$, identified by part 1, and from resampled residuals, which are invariant; by \ref{A:mean} the centring $\widehat{{u}}^{c}_{t}=\widehat{{u}}_{t}-\overline{\widehat{{u}}}$ leaves them unchanged. Conditionally on the resampling indices the pseudo-samples are therefore the same whichever least squares solution is used to generate them. As the intervention variables are held fixed, \ref{A:v0null} holds for the design ${X}^{(b)}$ of every replication whatever the values of ${y}^{(b)}_{t}$, and so does ${X}^{(b)}{e}_{k_{m}}={0}$ for $m\in\mathcal{Z}$; hence ${V}\subset\mathcal{N}({X}^{(b)})$, every replication is subject to the same rank deficiency and to the same normalisation \ref{eq:12}, and, whenever (A1) holds for ${X}^{(b)}$ (equivalently $\operatorname{rank}{X}^{(b)}=K-1-\lvert\mathcal{Z}\rvert$), part 1 applies to it and the resampled responses are comparable across replications.

\emph{The divisor in \ref{eq:4}.} With a rank-deficient design the residual sum of squares has $n-\operatorname{rank}({X})=n-K+1+\lvert\mathcal{Z}\rvert$ degrees of freedom. Let $\widehat{\Sigma}_{u}^{\ast}:=\widehat{{U}}^{\prime}\widehat{{U}}/(n-K+1+\lvert\mathcal{Z}\rvert)$ and
\begin{equation}\label{A:kappa}
	\kappa:=\Bigl(\frac{n-K}{n-K+1+\lvert\mathcal{Z}\rvert}\Bigr)^{1/2}\in(0,1),
	\qquad\text{so that}\qquad \widehat{\Sigma}_{u}^{\ast}=\kappa^{2}\widehat{\Sigma}_{u}.
\end{equation}
The matrix $\kappa{P}$ is lower triangular with positive diagonal and $(\kappa{P})(\kappa{P})^{\prime}=\kappa^{2}\widehat{\Sigma}_{u}$, so by the uniqueness of the Cholesky factor ${P}^{\ast}=\kappa{P}$ and $\Theta_{h}^{\ast}=\Xi_{h}{P}^{\ast}=\kappa\Theta_{h}$. Both choices are invariant to the normalisation, being functions of the invariant residuals, but they do not give the same multiplier. Writing $m_{i}(H,\delta;\Theta)$ for \ref{eq:multiplier} computed from $\Theta_{0},\dots,\Theta_{H}$,
\begin{equation}\label{A:multscale}
	m_{i}(H,\delta;\kappa\Theta)
	=\frac{\prod_{h=0}^{H}\bigl(1+\delta\kappa[\Theta_{h}]_{ij}\bigr)}{\prod_{h=0}^{H}\bigl(1+\delta\kappa[\Theta_{h}]_{jj}\bigr)}
	=m_{i}(H,\kappa\delta;\Theta),
\end{equation}
so replacing the divisor is equivalent to evaluating the multiplier at the impulse $\kappa\delta$ instead of $\delta$. By \ref{eq:18}, $m_{i}(H,0;\Theta)=1$ and
\begin{equation}\label{A:deriv}
	\frac{\partial m_{i}(H,\delta;\Theta)}{\partial\delta}\bigg|_{\delta=0}=\sum_{h=0}^{H}[\Theta_{h}]_{ij}-\sum_{h=0}^{H}[\Theta_{h}]_{jj},
\end{equation}
which is nonzero in general, so $m_{i}$ is not constant in $\delta$ and the rescaling does not cancel between numerator and denominator. In this paper the divisor $n-K=T-p-K$ is used, which gives the unit sample covariance of the structural shocks in \ref{A:epsmoment}.
\end{document}